\documentclass[11pt,onecolumn]{article}
\usepackage[margin=1.05in,top=1.1in,bottom=1.1in]{geometry}

\usepackage{amsmath, amssymb, amsfonts}
\usepackage{amsthm}
\usepackage{bm}
\usepackage{booktabs}
\usepackage{array}
\usepackage{tabularx}
\usepackage{graphicx}
\usepackage{xcolor}
\usepackage{hyperref}
\usepackage{cite}
\usepackage{algorithm}
\usepackage{algorithmic}
\usepackage{mathtools}
\usepackage{siunitx}
\usepackage[utf8]{inputenc}
\usepackage[T1]{fontenc}
\usepackage{amssymb} 
\usepackage{setspace}

\hypersetup{
  colorlinks  = true,
  linkcolor   = blue!65!black,
  citecolor   = blue!65!black,
  urlcolor    = blue!65!black,
  pdftitle    = {Geometry-Driven Islanding Detection for Grid-Forming Inverters},
  pdfauthor   = {Arekar et al.},
  pdfsubject  = {arXiv Preprint},
  pdfkeywords = {Grid-forming inverter, islanding detection, NAIM, Fenichel theory,
                 chi-squared threshold, Lyapunov equation, modal fault attribution}
}

\newtheorem{theorem}{Theorem}[section]
\newtheorem{lemma}[theorem]{Lemma}
\newtheorem{corollary}[theorem]{Corollary}
\newtheorem{proposition}[theorem]{Proposition}
\newtheorem{definition}[theorem]{Definition}
\newtheorem{remark}[theorem]{Remark}

\newcommand{\Mzero}{\mathcal{M}_0}
\newcommand{\xip}{\boldsymbol{\xi}_\perp}
\newcommand{\bxipt}{\bar{\boldsymbol{\xi}}_{\perp,t}}
\newcommand{\bxip}{\bar{\boldsymbol{\xi}}_\perp}
\newcommand{\Siginst}{\boldsymbol{\Sigma}}
\newcommand{\Sigmlong}{\boldsymbol{\Sigma}_{\mathrm{long}}}
\newcommand{\Aperp}{\mathbf{A}_\perp}
\newcommand{\PF}{\mathcal{P}_F}
\newcommand{\Dt}{\mathcal{D}_t}

\begin{document}

\title{\textbf{Geometry-Driven Islanding Detection and Fault Classification 
       for Grid-Forming Inverters: A Normally Hyperbolic
       Invariant Manifold Framework with Physics-Derived Thresholds}}

\author{Shubh Arekar$^{1}$, Aditi Ramteke$^{1}$, Kshitij Gaikwad$^{1}$,\\
        Tejal Jadhav$^{1}$, Shashank Verma$^{1}$, Madhavi Parimi$^{1}$, Sushama Wagh$^{1}$\\[6pt]
        \small $^1$EMC$^2$ Laboratory, Department of Electrical Engineering,\\
        \small Veermata Jijabai Technological Institute (VJTI), Mumbai 400019, India\\
        \small Corresponding author: \texttt{sparekar\_m24@ee.vjti.ac.in}
}

\date{Preprint --- \today}

\maketitle
\thispagestyle{empty}

\begin{center}
\small\itshape
This is a preprint submitted to arXiv. The work presents a geometry-driven
islanding detection framework for grid-forming inverters. Comments welcome.
\end{center}
\vspace{6pt}

\begin{abstract}
\noindent
This paper develops a geometry-driven detection and fault-classification
framework for grid-forming (GFM) inverters, grounded in the theory of normally
hyperbolic invariant manifolds (NAIM) and stochastic hypothesis testing. The
GFM droop manifold $\mathcal{M}_0$ is identified as a NAIM of the closed-loop
droop dynamics; its transverse fluctuations under grid noise are characterised
as an Ornstein--Uhlenbeck process whose long-run covariance
$\boldsymbol{\Sigma}_{\mathrm{long}}$ is derived analytically from the
algebraic Lyapunov equation. The detection statistic
$\mathcal{D}_t = T_w\bar{\boldsymbol{\xi}}_\perp^\top
\boldsymbol{\Sigma}_{\mathrm{long}}^{-1}\bar{\boldsymbol{\xi}}_\perp$
converges to $\chi^2(2)$ under the null hypothesis, yielding the closed-form,
tuning-free threshold $D_\alpha = -2\ln\alpha$ with an asymptotically exact
false-alarm rate (FAR) of $\alpha$. A definitive correction that removes a
factor-of-2 error present in earlier formulations is established analytically
and validated by $N_{\mathrm{MC}} = 8{,}000$ Monte Carlo realisations across
nine window lengths and three significance levels
($\alpha \in \{0.01, 0.05, 0.10\}$). The Berry--Esseen convergence rate
$d_{\mathrm{KS}} \leq 1.6704/(\beta T_w)$ is empirically confirmed, with the
minimum window condition $T_w \geq 10/\beta_{\min} \approx 1.0$~s (using the
slower channel $\beta_{\min} = \min(\omega_f, \omega_v)$) satisfying the
IEEE~1547-2018 two-second detection mandate by a factor of two. A co-design
theorem proves that increasing the filter bandwidths $(\omega_f, \omega_v)$
simultaneously maximises the Fenichel spectral gap (dynamic stability),
tightens the null covariance (detection sensitivity), and minimises the
FAR at the grid code threshold---with no inherent trade-off. Modal decomposition
of $\mathcal{D}_t$ into frequency and voltage contributions classifies islanding
(frequency-dominated) from voltage faults (voltage-dominated) without additional
sensors. Three fully worked scenarios confirm: normal operation is accepted
($\mathcal{D} = 1.45$, $p = 0.484$); soft islanding is overwhelmingly rejected
with detection power saturating to unity at $T_w < 0.1$~s; and a 10\% voltage
sag is detected and correctly classified as a grid fault
($d_V/\mathcal{D} = 99.2\%$).

\end{abstract}

\newpage
\tableofcontents
\newpage

\sloppy
\section{Introduction}
\label{sec:intro}

\subsection{Background and Motivation}

The fast growth of converter-based electricity generation worldwide due to
technologies such as solar photovoltaic plants, offshore wind farms, and
battery storage is changing the nature of dynamic behavior in electrical grids.
The grid-forming (GFM) inverter technology is what makes possible renewable
energy sources penetration in high levels: contrary to the operation principle of
grid-following (GFL) inverters, which rely on an external voltage reference in a
phase-locked loop, GFM inverters produce both voltage magnitude and frequency
independently of external influence through virtual synchronous machine (VSM)
or droop control algorithms \cite{Rocabert2012, Milano2018, Tayyebi2020}.
In addition to the independence from voltage references, the grid-forming inverters
provide frequency support and black-start capability.

A key safety and reliability issue in GFM-based networks is
\emph{islanding detection} that involves the determination of whether an
inverter or group of inverters has been electrically disconnected from the
rest of the grid despite supplying power to a locally connected load island.
Islanding undetected is a violation of the IEEE 1547-2018 standard that
calls for detection in less than two seconds of isolation occurrence; the
consequences of undetected islanding can be fatal, namely, phase mismatches,
utility workers' injury, equipment damage, and potential fire. Islanding
detection represents a particular challenge for GFM networks since
frequency/voltage control based on droop control ensures maintenance of both
voltage and frequency after islanding; hence, compared to GFL networks, ROCOF
and magnitude of voltage may display much lower variation after islanding, as
shown in the literature \cite{Lopes2006, Pogaku2007}.

Another complication in this regard involves stochastic perturbations
in actual power systems. Such things as load variations, voltage
noise in the grid, variability in renewable power supply, and
measurement noise all cause a stochastic oscillation in the inverter’s
frequency and voltage. Thus, an islanding detection algorithm must be
able to differentiate between the stochastic movement of a grid-tied
inverter system and the deterministic off manifold movement of an
isolated one—an inherently statistical problem that requires theory about
the false-alarm rate (FAR).

\subsection{Limitations of Existing Methods}

Current schemes for islanding detection can be broadly classified into four groups,
each having fundamental drawbacks motivating our proposed solution.

\textbf{Passive approaches:} The rate of change of frequency (ROCOF)
~\cite{OKane1997} and the vector surge approach ~\cite{Freitas2005}
compare the value of either voltage or frequency with an empirically-defined
threshold. These are relatively quick to operate and do not introduce any power
quality problem, but the drawback here is their inability to detect islanding
in a range of load-to-generation matching referred to as \emph{non-detection
zone} (NDZ). Importantly, thresholds used in passive schemes have no relation
to closed-loop inverter dynamics and, therefore, do not guarantee any FAR
figure.

\textbf{Active approaches:} The active frequency drift (AFD)
~\cite{Ropp1999}, as well as the Sandia frequency shift approach (SFS)
, rely on introducing external perturbations aimed at eliciting
the response that will serve as a basis for detecting islanding condition.
While these techniques reduce the size of NDZ, they continuously create power
quality problems (harmonic injections), need careful tuning of parameters, and,
fundamentally, contradict.

\textbf{Model-based techniques:} Methods based on Kalman filter innovation
techniques~\cite{KarimiNikk2008} evaluate model-predicted and measured voltages using
innovation analysis. If the model is precise, innovations are distributed
accordingly and the technique produces a well-founded hypothesis. However, in
reality, deviations from the ideal case due to grid impedances,
load switching, and harmonics make innovation covariance uncertain. The
resulting FAR will increase exponentially in time due to trajectory-related
errors.

\textbf{Data-driven techniques:} SVM and intelligent classifier~\cite{ElArroudi2007}
provide excellent performance on test data sets. However, both techniques
require large training data sets, do not provide easily interpretable
detection thresholds, do not offer any guarantees of generalization to unseen
operating conditions, and fail to link detection results to inverters'
physics.

\textbf{The inherent limitation:} The detection threshold in all prior art
works is not derived from an analysis of inverter closed-loop behavior.
This paper solves this problem and introduces an analysis framework based
on the theory of normally hyperbolic invariant manifolds (NAIM)~\cite{Fenichel1979, HirschPughShub1977}.
GFM droop manifold geometry plays a central role.

\subsection{Proposed Approach}

The NAIM model answers the following important theoretical question that we
are unaware has yet been investigated in islanding detection research:

\begin{quote}
\itshape
What is the theoretically optimal threshold for detecting islanding using a
GFM inverter, and how does it depend---in theory---on the droop filter
bandwidths, droop gains, and grid noise?
\end{quote}

\noindent There are four steps to the approach which are mathematically
rigorous: (i)~characterize the GFM droop manifold as a NAIM of the closed-loop
dynamics system, and calculate the Fenichel spectral gap $\mathcal{P}_F$;
(ii)~prove that transverse fluctuations inside the boundary layer follow an
Ornstein--Uhlenbeck (OU) process where $\boldsymbol{\Sigma}_{\mathrm{long}}$, the
covariance matrix after many time periods, is the unique solution to a continuous
algebraic Lyapunov equation; (iii)~calculate a bivariate $\chi^2(2)$ test statistic
from the time-averaged transverse state and the correct covariance matrix; and
(iv)~the threshold $D_\alpha = -2 \ln \alpha$ can be derived analytically from the
tail of a $\chi^2(2)$ random variable.The resulting statistic is validated by large-scale
Monte Carlo simulation across three operational scenarios and three
significance levels.

\subsection{Main Contributions}

The main contributions of this paper are:

\begin{itemize}
  \item[\textbf{C1.}] \textbf{NAIM identification:} We prove that the GFM
    droop manifold $\mathcal{M}_0$ is normally hyperbolic
    (Theorem~\ref{thm:nh}) with explicit transverse linearisation
    $\mathbf{A}_\perp(k_P, k_Q)$ and Fenichel spectral gap
    $\mathcal{P}_F = (\omega_f + \omega_v)/2 - \hat\rho$, independent
    of the droop gains. For typical operating parameters ($k_Pk_Q \ll 1$),
    the eigenvalues of $\mathbf{A}_\perp$ are real and distinct
    ($\approx -\omega_f$, $-\omega_v$); $\beta = (\omega_f+\omega_v)/2$
    is the mean contraction rate.

  \item[\textbf{C2.}] \textbf{Definitive detection statistic:} We establish
    the correct statistic (Definition~\ref{def:statistic}) using the properly
    derived long-run covariance
    $\boldsymbol{\Sigma}_{\mathrm{long}} = \sigma^2\,
    \mathrm{diag}(\omega_f^{-2}, \omega_v^{-2})$ (Theorem~\ref{thm:lyapunov}).
    The leading-order formula is:
    \begin{equation*}
      \mathcal{D}_t \approx \frac{T_w}{\sigma^2}
        \bigl[\omega_f^2(\delta\omega)^2 + \omega_v^2(\delta V)^2\bigr].
    \end{equation*}
    We identify and correct a factor-of-2 error in earlier formulations;
    the correction is analytically derived and Monte Carlo validated.

  \item[\textbf{C3.}] \textbf{Physics-derived closed-form threshold:} The
    threshold $D_\alpha = -2\ln\alpha$ (Theorem~\ref{thm:threshold}) requires
    no empirical tuning, carries an asymptotically exact FAR of $\alpha$, and
    admits the Berry--Esseen finite-window correction
    $d_{\mathrm{KS}}(\mathcal{D}_t, \chi^2(2)) \leq C/(\beta T_w)$ with
    empirically calibrated $C = 1.6704$. The minimum window
    $T_w^{\min} = 10/\beta_{\min} \approx 1.0$~s (using the slower channel
    $\beta_{\min} = \min(\omega_f,\omega_v)$) satisfies the IEEE~1547-2018
    two-second detection mandate by a factor of two.

  \item[\textbf{C4.}] \textbf{Filter bandwidth co-design theorem:} Increasing
    $\omega_f$ and $\omega_v$ simultaneously maximises $\mathcal{P}_F$
    (stability), tightens the null covariance (detection sensitivity), and
    minimises the FAR at the grid code threshold
    (Theorem~\ref{thm:codesign}). Stability and detection are the
    \emph{same} design problem.

  \item[\textbf{C5.}] \textbf{Modal fault attribution:} The decomposition
    $\mathcal{D}_t = d_\omega + d_V$ classifies islanding
    (frequency-dominated), voltage faults (voltage-dominated), and mixed
    events without additional sensors (Proposition~\ref{prop:modal}).

  \item[\textbf{C6.}] \textbf{Quantified robustness.} FAR inflation under
    model mismatch $\delta\mathbf{A}$ is $O(\|\delta\mathbf{A}\|/\beta)$
    (Proposition~\ref{prop:robust}), providing an explicit perturbation
    budget with uniform-in-time validity.

  \item[\textbf{C7.}] \textbf{Large-scale Monte Carlo validation.} Corrected value of the test statistic is verified using $N_{\mathrm{MC}}=8{,}000$ samples for three settings and three levels of significance $\alpha=\{0.01, 0.05, 0.10\}$, demonstrating empiric convergence of FAR, decay of
\end{itemize}

\subsection{Paper Organisation}

Section~\ref{sec:model} develops the GFM inverter model proves the NAIM model structture. section~\ref{sec:stat} derives the corrected detection statistic 
and threshold. section~\ref{sec:codesign} establishes the co-desigen theorem.
section~\ref{sec:robust} analyses robustness section~\ref{sec:comparison}
compares against five standard methods and characterises the NDZ analytically.
section~\ref{sec:scenarios} presents three fully worked numerical scenarios.
section~\ref{sec:simulations} presents the completed simulation campaign
(S1, S3, S4, and S5). section~\ref{sec:sim_protocols} specifies the remaining
simulation protocols (S2 and S6) scection~\ref{sec:limitations} discusses current limitations
current limitations. Section~\ref{sec:conclusion} concludes.

\sloppy
\section{GFM Inverter Model and NAIM Structure}
\label{sec:model}

\subsection{Three-Phase GFM Power Circuit and dq-Frame Equations}

A three-phase two-level VSC running in a GFM configuration with rating $S_{\mathrm{rated}}$ and frequency $\omega^*$, at a nominal voltage of $V^*$ is composed of the following circuit elements. The elements consist of the bridge circuit, an output LC filter ($L_f$ and $C_f$), an interface transformer ($L_T$ and $n$), and the PCC bus. In the synchronously rotating dq-frame attached to the
inverter angle $\theta(t)$, the filter and transformer dynamics are:
\begin{align}
  L_f \dot{i}_{d} &= v_{od} - v_{d} + \omega L_f i_{q}, \label{eq:filter_d}\\
  L_f \dot{i}_{q} &= v_{oq} - v_{q} - \omega L_f i_{d}, \label{eq:filter_q}\\
  C_f \dot{v}_{d} &= i_{d} - i_{Ld} + \omega C_f v_{q}, \label{eq:cap_d}\\
  C_f \dot{v}_{q} &= i_{q} - i_{Lq} - \omega C_f v_{d}. \label{eq:cap_q}
\end{align}
A standard inner cascaded control loop (bandwidth $\omega_c \gg \omega_f,
\omega_v$) regulates the filter current and output voltage, rendering the
inner-loop dynamics effectively instantaneous on the droop timescale. The
active and reactive power delivered to the grid are measured by first-order
low-pass filters:
\begin{align}
  \dot{P} &= -\omega_c(P - P_{\mathrm{grid}}), \label{eq:Pfilt}\\
  \dot{Q} &= -\omega_c(Q - Q_{\mathrm{grid}}), \label{eq:Qfilt}
\end{align}
where $P_{\mathrm{grid}} = v_d i_{Ld} + v_q i_{Lq}$ and
$Q_{\mathrm{grid}} = v_q i_{Ld} - v_d i_{Lq}$ at the transformer secondary.

\subsection{Averaged Closed-Loop Droop Dynamics}
\label{sec:droop}

After closing the inner current and voltage loops (time-scale separation:
$\omega_c \gg \omega_f, \omega_v$), the averaged closed-loop dynamics in the
outer droop layer reduce to:
\begin{align}
  \dot{\omega} &= -\omega_f\bigl[(\omega - \omega^*) + k_P(P - P^*)\bigr],
    \label{eq:droop_omega}\\
  \dot{V} &= -\omega_v\bigl[(V - V^*) + k_Q(Q - Q^*)\bigr],
    \label{eq:droop_V}
\end{align}
where $\omega_f, \omega_v > 0$~(rad/s) are the frequency and voltage filter
bandwidths; $k_P \geq 0$~(pu/Hz) and $k_Q \geq 0$~(pu/kV) are the active and
reactive power droop gains; and $(\omega^*, V^*, P^*, Q^*)$ is the nominal
operating set-point. The power measurements $(P, Q)$ are treated as quasi-static
inputs to the droop equations on the inner timescale.

The state vector on the droop timescale is
$x = (\omega, V, P, Q)^\top \in \mathbb{R}^4$.
Define deviation variables:
$\delta\omega = \omega - \omega^*$,
$\delta V = V - V^*$,
$\delta P = P - P^*$,
$\delta Q = Q - Q^*$.

\subsection{The Droop Manifold and Its Normal Hyperbolicity}
\label{sec:naim}

\begin{definition}[Droop Manifold]
The GFM droop manifold $\Mzero$ is the two-dimensional invariant set on which
the droop equations are satisfied exactly:
\begin{equation}
  \Mzero = \bigl\{(\omega, V, P, Q) :
    \delta\omega = -k_P\,\delta P,\;
    \delta V = -k_Q\,\delta Q \bigr\}.
  \label{eq:M0}
\end{equation}
\end{definition}

\begin{theorem}[Normal Hyperbolicity of $\Mzero$]
\label{thm:nh}
Under $\omega_f, \omega_v, \omega_c > 0$ and $k_P, k_Q \geq 0$, the manifold
$\Mzero$ is normally hyperbolic with transverse state
$\xip = (\delta\omega, \delta V)^\top \in \mathbb{R}^2$ and transverse
linearisation:
\begin{equation}
  \Aperp(k_P, k_Q) =
  \begin{pmatrix}
    -\omega_f & k_P\omega_f \\
    -k_Q\omega_v & -\omega_v
  \end{pmatrix}.
  \label{eq:Aperp}
\end{equation}
The eigenvalues $\lambda_{1,2}$ of $\Aperp$ are \emph{real and negative}
whenever $k_Pk_Q < (\omega_f - \omega_v)^2/(4\omega_f\omega_v)$ (discriminant
positive), and complex-conjugate with $\mathrm{Re}(\lambda_{1,2}) = -\beta < 0$
otherwise, where $\beta = (\omega_f + \omega_v)/2$ is the \emph{mean}
transverse contraction rate. In both cases $\mathrm{Re}(\lambda_{1,2}) < 0$,
confirming normal hyperbolicity. The \emph{minimum} transverse contraction rate
(spectral abscissa of $-\Aperp$) is
\begin{equation}
  \beta_{\min} = \min\!\bigl\{|\mathrm{Re}(\lambda_1)|,\,|\mathrm{Re}(\lambda_2)|\bigr\}
  \approx \min(\omega_f, \omega_v) \quad\text{(for }k_Pk_Q \ll 1\text{)}.
  \label{eq:betamin}
\end{equation}
The Fenichel spectral gap is defined using $\beta = (\omega_f + \omega_v)/2$:
\begin{equation}
  \PF = \beta - \hat\rho = \frac{\omega_f + \omega_v}{2} - \hat\rho,
  \label{eq:PF}
\end{equation}
where $\hat\rho < \beta$ is the slow synchronisation rate on $\Mzero$.
\end{theorem}

\begin{proof}
Linearise~\eqref{eq:droop_omega}--\eqref{eq:droop_V} about a point on $\Mzero$
with $\delta P = \delta Q = 0$ (valid for $\omega_c \gg \beta$). The transverse
Jacobian is $\Aperp$ in~\eqref{eq:Aperp}. Its characteristic polynomial is
$\lambda^2 + (\omega_f + \omega_v)\lambda + \omega_f\omega_v(1 + k_Pk_Q) = 0$,
with discriminant $\Delta_{\mathrm{eig}} = (\omega_f - \omega_v)^2 -
4k_Pk_Q\omega_f\omega_v$. The sum of eigenvalues equals
$-(\omega_f + \omega_v) = -2\beta < 0$ in all cases, so
$\mathrm{Re}(\lambda_1) + \mathrm{Re}(\lambda_2) = -2\beta < 0$, guaranteeing
both real parts are negative. Since $\beta > 0$ and $\hat\rho < \beta$ by
assumption, $\PF > 0$: $\Mzero$ is normally hyperbolic~\cite{Fenichel1979}.
\end{proof}

\begin{remark}[Eigenvalue Structure and the Role of $\beta$]
\label{rem:gains}
For typical GFM operating parameters with
$k_Pk_Q \ll (\omega_f - \omega_v)^2/(4\omega_f\omega_v)$ (e.g., baseline
$k_Pk_Q = 0.0025 \ll 0.125$ for $\omega_f = 10$, $\omega_v = 20$~rad/s),
the eigenvalues of $\Aperp$ are \emph{real and distinct}, approximately
$\lambda_1 \approx -\omega_f$ and $\lambda_2 \approx -\omega_v$. In this regime
$\beta = (\omega_f + \omega_v)/2$ is the arithmetic \emph{mean} contraction rate;
the minimum contraction rate is $\beta_{\min} \approx \min(\omega_f, \omega_v)$.
The Fenichel gap $\PF = \beta - \hat\rho$ is independent of the droop gains
$(k_P, k_Q)$: the gains affect the imaginary parts (if complex) and the
off-diagonal entries of $\Sigmlong$, but not $\beta$. For the Berry--Esseen
minimum window condition and detection delay analysis, $\beta_{\min}$ governs
the slowest-decaying mode; see Remark~\ref{rem:tmin}.
\end{remark}

\subsection{Geometric Chain: From Manifold to Statistical Test}

The normally hyperbolic structure guarantees a boundary layer
$\mathcal{T}_\delta(\Mzero)$ of thickness $\delta$ around $\Mzero$ such that
trajectories starting inside $\mathcal{T}_\delta$ converge exponentially to
$\Mzero$ at rate $\beta$ in the absence of external forcing. Under grid-connected
operation, stochastic forcing maintains $\xip(t)$ inside $\mathcal{T}_\delta$
in a stationary distribution. Under islanding or fault conditions, the
steady-state mean of $\xip$ shifts to ${\bar{\boldsymbol{\xi}}_{\perp,\mathrm{is}}} =
-\Aperp^{-1}\mathbf{b}$, potentially pushing $\xip$ outside
$\mathcal{T}_\delta$. The detection test is a hypothesis test for off-manifold
motion:

\begin{equation*}
  \begin{aligned}
    \Mzero
    &\xrightarrow{\;\substack{\text{OU process}\\ \text{in }\mathcal{T}_\delta}\;}
    \xip \sim \mathcal{N}(\mathbf{0}, \Sigmlong/T_w) \\
    &\xrightarrow{\;\chi^2\text{ statistic}\;}
    p = \Pr\!\bigl(\chi^2(2) \geq \mathcal{D}_t\bigr).
  \end{aligned}
\end{equation*}

\sloppy
\section{The NAIM Detection Statistic and Threshold}
\label{sec:stat}

\subsection{Stochastic Transverse Dynamics}
\label{sec:ou}

Under grid-connected operation with the droop control active, the transverse
state $\xip(t) \in \mathbb{R}^2$ satisfies an Ornstein--Uhlenbeck (OU)
stochastic differential equation driven by the aggregate effect of grid voltage
noise, load fluctuations, and measurement noise:
\begin{equation}
  d\xip = \Aperp\,\xip\,dt + \sigma\,dW_t, \quad \xip \in \mathbb{R}^2,
  \label{eq:OU}
\end{equation}
where $\sigma > 0$ is the effective noise intensity (units: pu) and $W_t$ is
standard 2-dimensional Brownian motion. Since $\Aperp$ is Hurwitz
(Theorem~\ref{thm:nh}), the process is ergodic with unique stationary
distributions $\mathcal{N}(\mathbf{0}, \Siginst)$.

\subsection{Instantaneous and Long-Run Covariances}
\label{sec:lyapunov}

\begin{theorem}[Lyapunov Covariances]
\label{thm:lyapunov}
The OU process~\eqref{eq:OU} with stable drift $\Aperp$ has two distinct
covariance matrices:

\textbf{(a) Instantaneous covariance $\Siginst$:} The unique symmetric
positive-definite (SPD) solution of:
\begin{equation}
  \Aperp\Siginst + \Siginst\Aperp^\top + \sigma^2 I_2 = 0. \label{eq:Lyapunov}
\end{equation}
With $s = \omega_f + \omega_v$ and $\Delta = \omega_f\omega_v(1 + k_Pk_Q)$,
obtained by solving the equivalent $3\times3$ linear system
(Appendix~\ref{app:lyapunov}, Eqs.~\eqref{eq:sys1}--\eqref{eq:sys3}):
\begin{equation}
  \Siginst = \frac{\sigma^2}{2\Delta s}
  \begin{pmatrix}
    \omega_v s + k_P^2\omega_f^2 + k_Pk_Q\omega_f\omega_v &
      k_P\omega_f^2 - k_Q\omega_v^2 \\[4pt]
    k_P\omega_f^2 - k_Q\omega_v^2 &
      \omega_f s + k_Q^2\omega_v^2 + k_Pk_Q\omega_f\omega_v
  \end{pmatrix}.
  \label{eq:Sigma_exact}
\end{equation}
The off-diagonal entry $(\Siginst)_{12} = \sigma^2(k_P\omega_f^2 -
k_Q\omega_v^2)/(2\Delta s)$ is \emph{negative} for the baseline parameters
($k_P = k_Q$, $\omega_v > \omega_f$).
At $k_P = k_Q = 0$: $\Siginst = (\sigma^2/2)\,\mathrm{diag}(\omega_f^{-1}, \omega_v^{-1})$.

\textbf{(b) Long-run covariance $\Sigmlong$:}
The integral of the autocovariance function of $\xip$:
\begin{equation}
  \Sigmlong := \int_{-\infty}^{\infty} \mathrm{Cov}\!\bigl(\xip(0),\,\xip(\tau)\bigr)\,d\tau
  = (-\Aperp)^{-1}\Siginst + \Siginst(-\Aperp)^{-\top}.
  \label{eq:Sigmalong}
\end{equation}
For $k_Pk_Q \ll 1$ (leading order in droop coupling):
\begin{equation}
  \boxed{
  \Sigmlong \approx \sigma^2\,\mathrm{diag}\!\left(\frac{1}{\omega_f^2},\,
  \frac{1}{\omega_v^2}\right), \qquad
  \Sigmlong^{-1} \approx \frac{1}{\sigma^2}\,\mathrm{diag}(\omega_f^2, \omega_v^2).
  }
  \label{eq:Sigmalong_approx}
\end{equation}
\end{theorem}

\begin{proof}
\textbf{part (a):} Existence and uniqueness follow from standard Lyapunov
theory~\cite{Bernstein2009}; the closed from~\eqref{eq:Sigma_exact} is obtained
via carmer's rule on the vectorised Lyapunov equation.

\textbf{Part (b):} By the Wiener--Khinchin theorem, the spectral density of the 
stationary OU Process is $s(\omega) =  (j\omega I_2 - \Aperp)^{-1}\sigma^2 I_2
(j\omega I_2 - \Aperp)^{-*}$, and $\Sigmlong = \int S(\omega)\,d\omega/(2\pi)
= (-\Aperp)^{-1}\Siginst + \Siginst(-\Aperp)^{-\top}$.

At $k_P = k_Q = 0$: $(-\Aperp)^{-1} = \mathrm{diag}(\omega_f^{-1},
\omega_v^{-1})$ and $\Siginst = (\sigma^2/2)\,\mathrm{diag}(\omega_f^{-1},
\omega_v^{-1})$. Each of the two equal terms in~\eqref{eq:Sigmalong}
contributes $(\sigma^2/2)\,\mathrm{diag}(\omega_f^{-2}, \omega_v^{-2})$;
their sum gives~\eqref{eq:Sigmalong_approx}. The factor of 2 that appeared in
earlier formulations arose from computing only one of the two equal terms; the
correct derivation sums both, yielding $\sigma^2$ (not $\sigma^2/2$) as the
diagonal coeffcient.
\end{proof}

\begin{remark}[Definitive Covariance Values for GFM Parameters]
\label{rem:params}
For the baseline parameters $\omega_f = 10$~rad/s, $\omega_v = 20$~rad/s,
$\sigma = 0.02$~pu, $k_P = k_Q = 0.05$~pu (so $k_Pk_Q = 0.0025 \ll 1$):
\begin{align*}
  \Siginst &= \mathrm{diag}(2\!\times\!10^{-5},\;10^{-5})\;\mathrm{pu}^2,\\
  \Sigmlong &= \sigma^2\,\mathrm{diag}(\omega_f^{-2}, \omega_v^{-2})
           = \mathrm{diag}(4\!\times\!10^{-6},\;10^{-6})\;\mathrm{pu}^2,\\
  \Sigmlong^{-1} &= \frac{1}{\sigma^2}\mathrm{diag}(\omega_f^2, \omega_v^2)
           = \mathrm{diag}(2.5\!\times\!10^{5},\;10^{6})\;\mathrm{pu}^{-2}.
\end{align*}
The null standard deviations are:
$\sigma^\mathrm{null}_\omega = \sigma/\omega_f = 0.002$~pu $= 0.10$~Hz and
$\sigma^\mathrm{null}_V = \sigma/\omega_v = 0.001$~pu $= 0.40$~V.
The correction from the exact Lyapunov solution to the leading-order formula
is $< 0.5\%$ for these parameters (Appendix~\ref{app:lyapunov}).
\end{remark}

\begin{remark}[Why $\Sigmlong \neq \Siginst$]
\label{rem:Siglong}
For the OU process, successive samples $\xip(t)$ and $\xip(t+h)$ are
correlated with correlation time $1/\beta$. The time-averaged mean
$\bxipt := T_w^{-1}\int_{t-T_w}^t \xip(s)\,ds$ has variance
$\Sigmlong/T_w$, not $\Siginst/T_w$. Using $\Siginst^{-1}$ in the statistic
gives a test that converges to a \emph{scaled} $\chi^2(2)$ distribution rather
than the standard $\chi^2(2)$, inflating the FAR by the exact factor
$\exp(D_\alpha/4) = 4.47$ at $\alpha = 0.05$ (FAR $\approx 0.224$
instead of $0.05$) for typical GFM parameters. The long-run covariance $\Sigmlong$ is the correct normalisation,
and differs from $\Siginst$ by the factor $(-\Aperp)^{-1} + (-\Aperp)^{-\top}$,
which equals $2\beta^{-1}I_2$ only when $\Aperp = -\beta I_2$ (uncoupled
channels). In general, $\Sigmlong$ is determined by the full Lyapunov solution.
\end{remark}

\subsection{Detection Statistic: Definition and Null Distribution}
\label{sec:statistic}

\begin{definition}[NAIM Detection Statistic]
\label{def:statistic}
Given measurement window $[t - T_w, t]$, define the window-mean transverse state:
\begin{equation}
  \bxipt := \frac{1}{T_w}\int_{t-T_w}^t \xip(s)\,ds
              = \frac{1}{T_w}\int_{t-T_w}^t
                \begin{pmatrix}\delta\omega(s)\\\delta V(s)\end{pmatrix}ds.
  \label{eq:winmean}
\end{equation}
The NAIM detection statistic is:
\begin{equation}
  \mathcal{D}_t := T_w\,{\bxipt}^\top \Sigmlong^{-1} \bxipt.
  \label{eq:Dt}
\end{equation}
\end{definition}

\begin{theorem}[Asymptotic Null Distribution]
\label{thm:stat}
Under $H_0$ (grid-connected operation) and the OU model~\eqref{eq:OU},
as $T_w \to \infty$:
\begin{equation}
  \mathcal{D}_t \xrightarrow{d} \chi^2(2). \label{eq:chi2}
\end{equation}
The finite-window Kolmogorov--Smirnov (KS) approximation error satisfies:
\begin{equation}
  d_{\mathrm{KS}}\!\left(\mathcal{D}_t,\, \chi^2(2)\right) \leq \frac{C}{\beta T_w},
  \label{eq:ks}
\end{equation}
for a universal constant $C > 0$. The condition $T_w \geq 10/\beta$ gives
$d_{\mathrm{KS}} < 1\%$ relative to the mean rate $\beta = 15$~rad/s, giving
$T_w^{\min} = 667$~ms. For the slowest-decaying mode
($\beta_{\min} = \min(\omega_f, \omega_v) = 10$~rad/s at baseline), the
conservative window is $T_w \geq 10/\beta_{\min} = 1.0$~s; both are well
within the IEEE~1547-2018 two-second detection mandate
(see Remark~\ref{rem:tmin}).
\end{theorem}

\begin{proof}
By the ergodic central limit theorem for stationary OU processes~\cite{Durrett2010}:
$\sqrt{T_w}\,\bxipt \xrightarrow{d} \mathcal{N}(\mathbf{0}, \Sigmlong)$.
Standardising: $\Sigmlong^{-1/2}\sqrt{T_w}\,\bxipt \xrightarrow{d}
\mathcal{N}(\mathbf{0}, I_2)$.
Taking the squared Euclidean norm:
$\mathcal{D}_t = \|\Sigmlong^{-1/2}\sqrt{T_w}\,{\bxipt}\|^2
\xrightarrow{d} \chi^2(2)$.
The Berry--Esseen rate~\eqref{eq:ks} follows from the mixing rate of the OU
process, controlled by the spectral gap $\beta$ of $-\Aperp$.
\end{proof}

\begin{remark}[Minimum Window and Spectral Gap]
\label{rem:tmin}
The bound~\eqref{eq:ks} uses $\beta = (\omega_f + \omega_v)/2$ as the mean
contraction rate; the slowest-decaying mode has rate
$\beta_{\min} = \min(\omega_f, \omega_v) \approx \omega_f$ for the baseline
parameters ($\omega_f = 10 < \omega_v = 20$~rad/s). The conservative minimum
window condition for the slowest mode is therefore
$T_w \geq 10/\beta_{\min} = 10/\omega_f \approx 1.0$~s, which still satisfies
the IEEE~1547-2018 two-second detection mandate by a factor of two. Using the
mean rate $\beta = 15$~rad/s gives the more optimistic $T_w^{\min} = 667$~ms;
both windows are well within the mandate.
\end{remark}

\textbf{Leading-order formula.} Using~\eqref{eq:Sigmalong_approx}
in~\eqref{eq:Dt}:
\begin{equation}
  \boxed{
  \mathcal{D}_t \approx \frac{T_w}{\sigma^2}\bigl[\omega_f^2(\delta\omega)^2
    + \omega_v^2(\delta V)^2\bigr],
  }
  \label{eq:Dt_leading}
\end{equation}
valid to within $< 1\%$ error for $k_Pk_Q < 0.01$ and $< 7\%$ for
$k_Pk_Q = 0.0025$ (the baseline parameter value). Each channel is weighted by
the square of its filter bandwidth; the two equal contributions to $\Sigmlong$
from the two terms in~\eqref{eq:Sigmalong} combine to give $\sigma^2$ (not
$\sigma^2/2$) as the correct normalisation coefficient. Monte Carlo validation
of this formula is presented in Section~\ref{sec:mc_validation}.

\textbf{Why $\omega_f^2$ (not $\omega_f$)?} The instantaneous null variance
$(\Siginst)_{11} = \sigma^2/(2\omega_f)$ carries one power of $\omega_f$. The
long-run variance $(\Sigmlong)_{11} = \sigma^2/\omega_f^2$ carries two powers:
the autocorrelation time $1/\omega_f$ reduces the effective number of
independent samples in window $T_w$ by the factor $\omega_f T_w$, giving
variance $[\sigma^2/(2\omega_f)] \times [2/(\omega_f T_w)] =
\sigma^2/(\omega_f^2 T_w)$. The $\chi^2$ statistic normalises by
$(\Sigmlong)_{11}/T_w = \sigma^2/\omega_f^2$, yielding $\omega_f^2$ after inversion.

\subsection{Physics-Derived Detection Threshold}
\label{sec:threshold}

\begin{theorem}[Closed-Form Threshold]
\label{thm:threshold}
For the $\chi^2(2)$ asymptotic null distribution, the detection threshold at
asymptotic false-alarm rate $\alpha \in (0,1)$ is:
\begin{equation}
  D_\alpha = -2\ln\alpha, \label{eq:threshold}
\end{equation}
since $P(\chi^2(2) \geq d) = e^{-d/2}$ for all $d \geq 0$. The asymptotic FAR:
\begin{equation}
  P_{H_0}(\mathcal{D}_t \geq D_\alpha) \xrightarrow{T_w \to \infty} \alpha.
  \label{eq:FAR}
\end{equation}
Numerically: $D_{0.10} = 4.605$; $D_{0.05} = 5.991$; $D_{0.01} = 9.210$.
\end{theorem}

\begin{remark}[Empirical Constant $C$ in the Berry--Esseen Bound]
The Monte Carlo S1 results of Section~\ref{sec:mc_validation} empirically
estimate $C = 1.6704$ by fitting the observed $d_{\mathrm{KS}}$ values to the
functional form $C/(\beta T_w)$ across $T_w \in [0.4, 20]$~s. The 95\%
envelope is $2C/(\beta T_w)$; all empirical $d_{\mathrm{KS}}$ values lie
strictly below this envelope, confirming the bound.
\end{remark}

\subsection{Modal Fault Attribution}
\label{sec:modal}

\begin{proposition}[Modal Decomposition]
\label{prop:modal}
Define the frequency and voltage modal contributions:
\begin{align}
  d_\omega &:= \frac{T_w\,(\overline{\delta\omega})^2}{(\Sigmlong)_{11}}
             = \frac{T_w\omega_f^2}{\sigma^2}(\overline{\delta\omega})^2,
  \label{eq:domega}\\
  d_V &:= \frac{T_w\,(\overline{\delta V})^2}{(\Sigmlong)_{22}}
        = \frac{T_w\omega_v^2}{\sigma^2}(\overline{\delta V})^2,
  \label{eq:dV}
\end{align}
so that $\mathcal{D}_t \approx d_\omega + d_V$. Under $H_0$, each
$d_\omega, d_V \xrightarrow{d} \chi^2(1)$ independently as $T_w\to\infty$.
The attribution fractions classify the event:
\begin{itemize}
  \item $d_\omega/\mathcal{D}_t > 0.70$: \emph{frequency-dominated}
    $\Rightarrow$ islanding likely.
  \item $d_V/\mathcal{D}_t > 0.70$: \emph{voltage-dominated}
    $\Rightarrow$ grid voltage fault likely.
  \item Both $< 0.70$: \emph{mixed} event (simultaneous active and reactive imbalance).
\end{itemize}
\end{proposition}

\subsection{Four-Step Online Algorithm}
\label{sec:algorithm}

\begin{algorithm}
\caption{NAIM Islanding Detection and Fault Classification}
\label{alg:naim}
\begin{algorithmic}[1]
\STATE \textbf{Offline design:} Given $(\omega_f, \omega_v, k_P, k_Q, \sigma)$:
  solve Lyapunov equation~\eqref{eq:Lyapunov} for $\Siginst$;
  compute $\Sigmlong$ via~\eqref{eq:Sigmalong} (or~\eqref{eq:Sigmalong_approx});
  set $D_\alpha = -2\ln\alpha$.
\STATE \textbf{Online computation:} Maintain running window mean
  $\bxipt = (\overline{\delta\omega}_t, \overline{\delta V}_t)^\top$;
  evaluate $\mathcal{D}_t = (T_w/\sigma^2)[\omega_f^2(\overline{\delta\omega})^2
  + \omega_v^2(\overline{\delta V})^2]$.
\STATE \textbf{Detection:} If $\mathcal{D}_t > D_\alpha$, declare off-manifold alarm.
\STATE \textbf{Classification:} Apply Proposition~\ref{prop:modal} using
  $d_\omega/\mathcal{D}_t$ and $d_V/\mathcal{D}_t$.
\end{algorithmic}
\end{algorithm}

Online computational cost is $O(1)$ per sample: two multiply-adds for
$\mathcal{D}_t$ plus two divisions for the attribution fractions. No matrix
inversions, Kalman predictions, or SVM evaluations are required at runtime.

\sloppy
\section{Filter Bandwidth Co-Design Theorem}
\label{sec:codesign}

A central contribution of this paper is the demonstration that three
design objectives---manifold stability, detection sensitivity, and
FAR minimisation at the grid code threshold---are unified by a single
choice of filter bandwidths $(\omega_f, \omega_v)$. This unification
is made precise in the following theorem.

\begin{theorem}[Filter Bandwidth Co-Design]
\label{thm:codesign}
The following three design objectives are simultaneously achieved by
choosing $\omega_f$ and $\omega_v$ as large as possible, subject to
inner-loop stability and grid code constraints:

\begin{itemize}
  \item[\textbf{(i)}] \textbf{Maximum Fenichel gap (stability):}
    \begin{equation}
      \max_{\omega_f, \omega_v} \PF = \frac{\omega_f + \omega_v}{2} - \hat\rho
        \quad\text{(linear and increasing in both)}.
    \end{equation}

  \item[\textbf{(ii)}] \textbf{Minimum null covariance (tightest boundary layer):}
    \begin{equation}
      \min_{\omega_f, \omega_v} \mathrm{tr}(\Sigmlong)
        = \sigma^2\!\left(\frac{1}{\omega_f^2} + \frac{1}{\omega_v^2}\right)
        \quad\text{(strictly decreasing in both)}.
    \end{equation}

  \item[\textbf{(iii)}] \textbf{Minimum $p$-value at the grid code threshold
    $(\delta\omega_{gc}, \delta V_{gc})$:}
    \begin{equation}
      \min_{\omega_f, \omega_v}
      \exp\!\left(-\frac{T_w}{2\sigma^2}\bigl[\omega_f^2\delta\omega_{gc}^2
        + \omega_v^2\delta V_{gc}^2\bigr]\right)
        \quad\text{(strictly decreasing in both)}.
    \end{equation}
\end{itemize}
\end{theorem}

\begin{proof}
\textbf{(i):} $\PF = (\omega_f + \omega_v)/2 - \hat\rho$ is linear and
strictly increasing in $\omega_f, \omega_v$.

\textbf{(ii):} $\mathrm{tr}(\Sigmlong) = \sigma^2(\omega_f^{-2} + \omega_v^{-2})$
is strictly decreasing in $\omega_f, \omega_v$. Objectives~(i) and~(ii) are
thus aligned: increasing bandwidths simultaneously maximises stability and
minimises the null covariance, with no trade-off between them.

\textbf{(iii):} The $p$-value at the grid code observation
$\xi_{gc} = (\delta\omega_{gc}, \delta V_{gc})^\top$ is
$\exp(-\mathcal{D}(\xi_{gc})/2)$, where
$\mathcal{D}(\xi_{gc}) = (T_w/\sigma^2)[\omega_f^2\delta\omega_{gc}^2
+ \omega_v^2\delta V_{gc}^2]$ by~\eqref{eq:Dt_leading}. This $p$-value is
minimised (i.e., detection confidence is maximised) by maximising
$\mathcal{D}(\xi_{gc})$, which is strictly increasing in
$\omega_f$ and $\omega_v$.
\end{proof}

\begin{lemma}[Positive-Semidefinite Monotonicity of $\Sigmlong$]
\label{lem:monotone}
For fixed $\sigma > 0$ and $k_P, k_Q \geq 0$, increasing $\omega_f$ (or
$\omega_v$) strictly decreases $\Sigmlong$ in the positive semidefinite order:
\begin{equation}
  \frac{d\Sigmlong}{d\omega_f} \prec 0.
\end{equation}
\end{lemma}

\begin{proof}
Differentiating the Lyapunov equation
$\Aperp\Siginst + \Siginst\Aperp^\top + \sigma^2 I = 0$ with respect to
$\omega_f$ yields a Sylvester equation for $\dot\Siginst = d\Siginst/d\omega_f$
with a positive-definite right-hand side, so $\dot\Siginst \prec 0$.
The result for $\Sigmlong$ follows by the chain rule
applied to~\eqref{eq:Sigmalong}.
\end{proof}

\begin{corollary}[Optimal Filter Bandwidths]
The co-design optimal filter bandwidths are:
\begin{equation}
  (\omega_f^*, \omega_v^*) = \arg\max_{\omega_f, \omega_v}
  \bigl\{\omega_f + \omega_v : \omega_f, \omega_v \leq \omega_{c}/\gamma,\;
  \omega_f^2/\omega_v^2 \leq r_{\max}\bigr\},
\end{equation}
where $\omega_c/\gamma$ (with $\gamma \geq 5$) maintains inner/outer
timescale separation, and $r_{\max}$ bounds the bandwidth ratio for
numerical conditioning. Subject to these constraints, the optimal solution
satisfies $\omega_f^* = \omega_v^* = \omega_c/\gamma$.
\end{corollary}

\section{Robustness Analysis}
\label{sec:robust}

A key practical concern is whether the analytically derived threshold
$D_\alpha = -2\ln\alpha$ remains valid in the presence of modelling errors.
The following proposition provides an explicit, uniform-in-time bound on the
FAR inflation due to model mismatch.

\begin{proposition}[FAR Sensitivity to Model Mismatch]
\label{prop:robust}
Let the true transverse matrix be $\Aperp^{\mathrm{true}} = \Aperp + \delta\mathbf{A}$
with $\varepsilon := \|\delta\mathbf{A}\|/\beta_{\min} < 1$, where
$\beta_{\min} = \min(\omega_f, \omega_v)$ is the minimum transverse contraction
rate. The perturbed long-run covariance satisfies
$\|\Sigmlong^{\mathrm{true}} - \Sigmlong\| \leq C_1\varepsilon\|\Sigmlong\|$
with $C_1 = 2/(1-\varepsilon)$. The actual false-alarm rate satisfies:
\begin{equation}
  \mathrm{FAR}_{\mathrm{actual}} = \alpha\bigl(1 + O(\varepsilon)\bigr).
  \label{eq:rob2}
\end{equation}
Specifically, for $\varepsilon < 0.1$ (10\% model error relative to $\beta_{\min}$),
the FAR inflation is less than $2C_1\varepsilon \cdot f_\alpha \cdot \alpha$,
where $f_\alpha = D_\alpha/(2\alpha)$ is a bounded smoothness factor of the
$\chi^2(2)$ CDF at $D_\alpha$.
\end{proposition}

\begin{proof}
Perturbing the Lyapunov equation and applying standard matrix perturbation
bounds~\cite{Bernstein2009}: $\|\Sigmlong^{\mathrm{true}} - \Sigmlong\|
\leq C_1\varepsilon\|\Sigmlong\|$ with $C_1 = 2/(1-\varepsilon)$ from the
Bauer--Fike theorem applied to the Lyapunov operator. The FAR perturbation
afterwards follow  the Lipchitz continuity of the  $\chi^2(2)$ complementry
CDF in the neighbourhood of $D_\alpha$, with Lipshtiz constant
$e^{-D_\alpha/2}/2 = \alpha/2$.
\end{proof}

\begin{remark}[Advantage Over Kalman Filter Robustness]
\label{rem:kalman_rob}
The Kalman Filter residual approach requires \emph{trajectory accuracy}: model
errors accmulate along trajectories at rate $O(e^{\|\delta A\|t})$,yeilding
unbounded FAR inflation for long observation windows. The NAIM bound in
Proposition~\ref{prop:robust} depends only on the covariance error, which is
bounded uniformly in time by Lemma~\ref{lem:monotone}. This structural
advantage is a direct consequence of the stationary nature of the NAIM
detection statistic.
\end{remark}

\section{Comparison Against Standard Methods}
\label{sec:comparison}

\subsection{Threshold Setting and FAR Guarantee}

Table~\ref{tab:comparison1} compares six islanding detection methods on their
threshold methodology and FAR guarantee. Table~\ref{tab:comparison2} provides a
structural capability comparison across six discriminating criteria.

\begin{table*}[!t]
\caption{Comparison of Threshold-Setting Methodology and FAR Guarantee}
\label{tab:comparison1}
\centering
\begin{tabular}{lccc}
\toprule
\textbf{Method} & \textbf{Threshold Parameter} & \textbf{How Determined} & \textbf{FAR Guarantee} \\
\midrule
ROCOF & $r_{\max}$ (Hz/s) & Empirical / grid code & None \\
Vector surge & $\phi_{\max}$ (deg) & Empirical / grid code & None \\
Passive impedance & $Z_{\max}$ ($\Omega$) & Empirical & None \\
Kalman residual & $\chi^2$ critical value & Model (if exact) & Model-dependent \\
SVM & Decision margin & Training data & None \\
\textbf{NAIM (proposed)} & $\boldsymbol{-2\ln\alpha}$ & \textbf{Lyapunov equation} & \textbf{Asympt.\ exact} \\
\bottomrule
\end{tabular}
\end{table*}

\begin{table*}[!t]
\caption{Structural Capability Comparison of Islanding Detection Methods}
\label{tab:comparison2}
\centering
\begin{tabular}{lcccccc}
\toprule
\textbf{Method} & \textbf{Joint} & \textbf{Physics-derived} & \textbf{Adaptive} & \textbf{Modal} & \textbf{Co-} & \textbf{Explicit}\\
 & $\omega$/$V$ test & $\Sigmlong$ & threshold & attribution & design & NDZ\\
\midrule
ROCOF & $\times$ & $\times$ & $\times$ & $\times$ & $\times$ & $\times$\\
Vector surge & $\times$ & $\times$ & $\times$ & $\times$ & $\times$ & $\times$\\
Passive imp. & $\sim$ & $\times$ & $\times$ & $\times$ & $\times$ & $\times$\\
Kalman & $\checkmark$ & $\sim$ & $\sim$ & $\times$ & $\times$ & $\times$\\
SVM & $\checkmark$ & $\times$ & $\times$ & $\times$ & $\times$ & $\times$\\
\textbf{NAIM (proposed)} & $\checkmark$ & $\checkmark$ & $\checkmark$ & $\checkmark$ & $\checkmark$ & $\checkmark$\\
\bottomrule
\multicolumn{7}{l}{$\checkmark$~=~yes;\quad $\sim$~=~partial;\quad $\times$~=~no.}
\end{tabular}
\end{table*}

The NAIM method is the only one that satisfies all six criteria simultaneously.
In particular, it is the only method providing an asymptotically exact FAR
guarantee derived from the inverter physics, an explicit co-design principle
linking filter design to detection performance, and an analytically defined NDZ
boundary.

\subsection{Non-Detection Zone Characterisation}
\label{sec:ndz}

\begin{proposition}[NAIM NDZ Boundary]
\label{prop:ndz}
During islanding with power imbalances $P_{\mathrm{imb}}$ and
$Q_{\mathrm{imb}}$, the steady-state off-manifold shift is
${\bar{\boldsymbol{\xi}}_{\perp,\mathrm{is}}} = -\Aperp^{-1}\mathbf{b}$,
where
$\mathbf{b} = (k_P\omega_f P_{\mathrm{imb}},\, k_Q\omega_v Q_{\mathrm{imb}})^\top/S_{\mathrm{rated}}$.
The NAIM detector fails to detect within $N$ successive windows if and
only if:
\begin{equation}
  ({\bar{\boldsymbol{\xi}}_{\perp,\mathrm{is}}})^\top \Sigmlong^{-1}
  {\bar{\boldsymbol{\xi}}_{\perp,\mathrm{is}}}
  \leq \frac{-2\ln\alpha}{N},
  \label{eq:NDZ_ellipse}
\end{equation}
defining an explicit ellipse in the $(P_{\mathrm{imb}}, Q_{\mathrm{imb}})$
plane. This NDZ ellipse shrinks as $N$ increases and vanishes as
$\alpha \to 0$.
\end{proposition}

This explicit NDZ ellipse is unique to the NAIM method: all other passive methods
have implicit NDZs determined by empirical threshold choices, with no analytical
boundary. The ellipse~\eqref{eq:NDZ_ellipse} enables the system designer to certify
detection coverage for any target power imbalance pair, without additional simulation.

\sloppy
\section{Worked Numerical Scenarios}
\label{sec:scenarios}

We apply the NAIM framework to the following baseline GFM parameters, using
the corrected covariance formula~\eqref{eq:Sigmalong_approx} throughout. All
Valus supersede earlier formulations that contained a factor-of-2 errors in the
covariance normalisation.

\begin{table*}[!t]
\caption{System Parameters and Derived Quantities}
\label{tab:params}
\centering
\begin{tabular}{llll}
\toprule
\textbf{Parameter} & \textbf{Value} & \textbf{Parameter} & \textbf{Value} \\
\midrule
$\omega_f$ & 10~rad/s & $\sigma$ & 0.02~pu \\
$\omega_v$ & 20~rad/s & $\beta = (\omega_f+\omega_v)/2$ & 15~rad/s \\
$k_P = k_Q$ & 0.05~pu & $T_w^{\min} = 10/\beta$ & 667~ms \\
$(\Sigmlong)_{11}$ & $4\!\times\!10^{-6}$~pu$^2$ & $\sigma^\mathrm{null}_\omega$ & 0.002~pu = 0.10~Hz \\
$(\Sigmlong)_{22}$ & $10^{-6}$~pu$^2$ & $\sigma^\mathrm{null}_V$ & 0.001~pu = 0.40~V \\
\bottomrule
\end{tabular}
\end{table*}

We evaluate at $T_w = 1$~s, $\alpha = 0.05$, $D_{0.05} = 5.991$.
The window condition $T_w = 1$~s $\geq 10/\beta = 667$~ms is satisfied.

\subsection{Scenario 1: Normal Grid Operation (Sc1)}
\label{sec:sc1}

\textbf{Observation:}
$\bxip = (0.0020,\,-0.00067)^\top$~pu\,
(frequency deviation $\overline{\delta f} = \omega_f \cdot 0.0020/(2\pi) \approx 0.10$~Hz;
voltage deviation $\overline{\delta V} = 0.00067 \times V_\mathrm{rated} \approx 0.27$~V).

\textbf{Calculation} using the corrected formula~\eqref{eq:Dt_leading}:
\begin{align}
  d_\omega
    &= \frac{T_w\,\omega_f^2}{\sigma^2}(\overline{\delta\omega})^2 \notag\\
    &= 2500\,(0.0020)^2 \notag\\
    &= 2500 \times 4\times10^{-6} = 1.00, \label{eq:sc1_domega}\\
  d_V
    &= \frac{T_w\,\omega_v^2}{\sigma^2}(\overline{\delta V})^2 \notag\\
    &= 10000\,(0.00067)^2 \notag\\
    &= 10000 \times 4.49\times10^{-7} = 0.45. \label{eq:sc1_dV}
\end{align}
\begin{equation}
  \begin{aligned}
    \mathcal{D}_1 &= d_\omega + d_V = 1.45,\\
    p_1 &= e^{-1.45/2} = e^{-0.725} = 0.484.
  \end{aligned}
  \label{eq:sc1_D}
\end{equation}

The exact Lyapunov result gives $\mathcal{D}_1^{\mathrm{exact}} = 1.36$
($p = 0.507$); the leading-order formula is within 6.6\% for $k_Pk_Q = 0.0025$.

\textbf{Decision:} $\mathcal{D}_1 = 1.45 < 5.991 = D_{0.05}$. \emph{Accept $H_0$.
Normal Grid Operation confirmed.}

The observed frequency excursion of 0.10 Hz corresponds to 1.00 null standard deviation, while the voltage excursion of 0.27 V is 0.68 null standard deviations. Both values stay inside the Fenichel boundary layer. With $\mathcal{D}_1 = 1.45$ close to the $\chi^2(2)$ distribution's median (about 1.39) and a p-value of 0.484, there's no evidence against the null hypothesis.

\textbf{Modal attribution:}
$d_\omega/\mathcal{D}_1 = 1.00/1.45 = 69.0\%$, $d_V/\mathcal{D}_1 = 31.0\%$.
The frequency channel is a bit stronger but doesn't cross the 70\% threshold, so the event is labeled as mixed. This fits with small, simultaneous changes in both active and reactive power under normal grid conditions. No fault alarm goes off.

\subsection{Scenario 2: Soft Islanding --- 3\% Active Power Imbalance (Sc2)}
\label{sec:sc2}

\textbf{Observation:}
$\bxip = (0.012,\,-0.008)^\top$~pu\,
(frequency drift $\overline{\delta f} \approx 0.60$~Hz;
voltage deviation $\overline{\delta V} \approx 3.2$~V).

\textbf{Physical mechaism:} A 3\% active power surplus at the start of islanding leads to an ongoing rise in frequency, around 0.015 pu, based on the droop characteristic. The reactive power imbalance also causes a shift in voltage. Once disconnected from the grid, frequency restoration torque is lost, causing the mean of $\xip$ to move from zero to $\bar\xip_\mathrm{is} = -\Aperp^{-1}\mathbf{b}$, pushing the operating point well outside $\mathcal{T}_\delta(\Mzero)$.

\textbf{Calculation}:
\begin{align}
  d_\omega &= \frac{1 \times 10^2}{(0.02)^2}(0.012)^2
           = 2500 \times 1.44\times10^{-4} = 36.0, \label{eq:sc2_domega}\\
  d_V     &= \frac{1 \times 20^2}{(0.02)^2}(0.008)^2
           = 10000 \times 6.4\times10^{-5} = 64.0. \label{eq:sc2_dV}
\end{align}
\begin{equation}
  \mathcal{D}_2 = 100.0, \quad p_2 = e^{-50} \approx 1.93\times10^{-22}.
  \label{eq:sc2_D}
\end{equation}

The Excat lyapuno results gives $\mathcal{D}_2^{\mathrm{exact}} = 93.2$
($p \approx 1.93\times10^{-22}$); both values overwhelmingly exceed $D_{0.05} = 5.991$.

\textbf{Decision:} $\mathcal{D}_2 = 100 \gg 5.991$. \emph{Reject $H_0$.
Islanding declared.}

\textbf{Physical Interpretation :} The 0.60 Hz frequency excursion equals 6.0 null standard deviations, while the voltage change hits 8.0 null standard deviations. Combining both, the $\mathcal{D}_2$ value is 100, which is over 16.7 times the detection threshold—virtually guaranteeing an detection with zero chance of false negatives. Plus, Monte Carlo simulations in Section~\ref{sec:mc_s2} show detection power stays at unity even with the shortest 667 ms window.

\textbf{Modal attribution:}
$d_\omega/\mathcal{D}_2 = 36.0/100.0 = 36.0\%$;
$d_V/\mathcal{D}_2 = 64.0/100.0 = 64.0\%$
Neither channel surpasses 70\%, classifying the event as mixed—consistent with a soft islanding event where both active and reactive power imbalances play roles. Active power imbalance is 3\%, but the voltage channel dominates since its null spread is tighter at 0.40 volts compared to the frequency spread of 0.10 Hz. This difference, where $\omega$ for voltage is greater than for frequency ($20$ versus $10~\mathrm{rad/s}$), amplifies the $d_V$ contribution.

\subsection{Scenario 3: Voltage Fault --- Grid-Connected, 10\% Sag (Sc3)}
\label{sec:sc3}

\textbf{Observation:}
$\bxip = (0.002,\,-0.011)^\top$~pu\,
(frequency deviation $\overline{\delta f} \approx 0.10$~Hz;
voltage sag $\overline{\delta V} \approx 4.4$~V).

\textbf{Physical mechanism :} A temporary 10\% grid-side voltage sag (lasting about 200 ms) leads to a big, long-lasting drop in window mean voltage. This happens because the GFM droop control readjusts $V$ towards the sagged grid voltage. Still, the grid connection stays put, maintaining frequency through sync. So, the frequency channel barely gets affected. The major sign of this issue is a large $d_V$ combined with almost no $d_\omega$.

\textbf{Calculation}:
\begin{align}
  d_\omega &= \frac{1 \times 10^2}{(0.02)^2}(0.002)^2
           = 2500 \times 4\times10^{-6} = 1.00, \label{eq:sc3_domega}\\
  d_V     &= \frac{1 \times 20^2}{(0.02)^2}(0.011)^2
           = 10000 \times 1.21\times10^{-4} = 121.0. \label{eq:sc3_dV}
\end{align}
\begin{equation}
  \mathcal{D}_3 = 122.0, \quad p_3 = e^{-61} \approx 3.22\times10^{-27}.
  \label{eq:sc3_D}
\end{equation}

The exact result gives $\mathcal{D}_3^{\mathrm{exact}} = 120.4$
($p \approx 3.22\times10^{-27}$).

\textbf{Decision:} $\mathcal{D}_3 = 122 \gg 5.991$. \emph{Reject $H_0$.
Off-manifold event declared.}

\textbf{Physical Interpretation :}The voltage goes up to 4.4 volts, which is 11.0 null standard deviations off—super extreme. But the frequency channel only contributes 1.00, that’s just 0.82\% of the total. It’s less than what we'd expect on average. This isn’t weird though; the 0.10 Hz frequency change blends in with regular null fluctuations, showing the grid's synchronised is still good.

\textbf{Modal attribution:}
$d_\omega/\mathcal{D}_3 = 0.82\%$; $d_V/\mathcal{D}_3 = 99.2\%$.
Overwhelmingly voltage-dominated. This shows the special diagnostic power of the bivariate NAIM test. A frequency-only detector like ROCOF would give $d = d_\omega = 1.45$ (below the threshold), missing the event. But the NAIM test both detects it ($\mathcal{D}_3 = 122$) and correctly classifies it as a voltage fault, all without extra measurement channels.
\subsection{Scenario Summary}

\begin{table*}[!t]
\caption{Corrected Scenario Results ($T_w = 1$~s, $\alpha = 0.05$, $D_{0.05} = 5.991$).
All Values from~\eqref{eq:Sigmalong_approx} and~\eqref{eq:Dt_leading};
excat Lyapunov values in parentheses.}
\label{tab:scenarios}
\centering
\begin{tabular}{lcccccc}
\toprule
\textbf{Sc.} & $\bxip$~(pu) & $d_\omega$ & $d_V$ & $\mathcal{D}$ & $p$ & \textbf{Decision}\\
\midrule
Sc1: Normal & $(2.0,-0.67)\!\times\!10^{-3}$ & 1.00 & 0.45
  & 1.45 (1.36) & 0.484 & Accept\\
Sc2: Island. & $(0.012,-0.008)$ & 36.0 & 64.0
  & 100 (93.2) & $\approx 0$ & Reject (Mixed)\\
Sc3: V-fault & $(0.002,-0.011)$ & 1.00 & 121.0
  & 122 (120.4) & $\approx 0$ & Reject (V-fault)\\
\bottomrule
\end{tabular}
\end{table*}

Table~\ref{tab:scenarios} summarises the three scenarios. The $p$-value curve
$p = e^{-\mathcal{D}/2}$ maps the detection statistic to probability:
Sc1 sits at $(\mathcal{D}, p) = (1.45, 0.484)$---accepted comfortably;
Sc2 and Sc3 are far beyond the right edge ($\mathcal{D} = 100$, $122$)---
overwhelmingly rejected. The modal decomposition uniquely distinguishes the
fault type without additional sensor hardware, a capability absent from all
five standard comparison methods.

\section{Simulation Results: S1, S3, S4, and S5}
\label{sec:simulations}
\sloppy
\subsection{S1: Monte Carlo Validation}
\label{sec:mc_validation}

\subsubsection{Simulation Protocol}
\label{sec:mc_protocol}

The S1 Monte Carlo campaign validates three analytically derived properties of
the NAIM detection statistic: (i)~the $\chi^2(2)$ null distribution under $H_0$
(Theorem~\ref{thm:stat}); (ii)~the Berry--Esseen convergence rate
$d_{\mathrm{KS}} \leq C/(\beta T_w)$ (equation~\eqref{eq:ks}); and
(iii)~the convergence of detection power to unity under $H_1$ for both the
soft islanding (Sc2) and voltage fault (Sc3) scenarios.

\textbf{Simulation engine.} The OU process~\eqref{eq:OU} is integrated by the
Euler--Maruyama scheme with $\Delta t = 0.005$~s (200~samples/s):
\begin{equation}
  \xip^{(n+1)} = \xip^{(n)} + \Aperp\xip^{(n)}\Delta t
    + \sigma\sqrt{\Delta t}\,\boldsymbol{\eta}^{(n)},
  \quad \boldsymbol{\eta}^{(n)} \sim \mathcal{N}(\mathbf{0}, I_2),
\end{equation}
initialised at $\xip^{(0)} = \mathbf{0}$ (on the manifold $\Mzero$). For
each realisation, the window mean
$\bxip = (1/N_s)\sum_{n=1}^{N_s}\xip^{(n)}$ is computed over
$N_s = \lfloor T_w/\Delta t\rfloor$ steps, and the statistic
$\mathcal{D} = (T_w/\sigma^2)[\omega_f^2(\overline{\delta\omega})^2 +
\omega_v^2(\overline{\delta V})^2]$ is evaluated using the corrected
formula~\eqref{eq:Dt_leading}. The exact Lyapunov statistic is computed
in parallel for cross-validation.

\textbf{Parameters.} $\omega_f = 10$, $\omega_v = 20$~rad/s;
$k_P = k_Q = 0.05$~pu; $\sigma = 0.02$~pu; $\beta = 15$~rad/s;
$N_{\mathrm{MC}} = 8{,}000$ realisations per configuration. Window lengths:
$T_w \in \{0.40, 0.67, 1.00, 2.00, 5.00, 8.00, 10.00, 15.00, 20.00\}$~s
(corresponding to $\beta T_w \in \{6, 10, 15, 30, 75, 120, 150, 225, 300\}$).
Significance levels: $\alpha \in \{0.01, 0.05, 0.10\}$.

\textbf{Under $H_1$.} For Sc2 (soft islanding), the OU drift is offset by
$\bar\xi^{(0)}_\mathrm{is} = (0.012, -0.008)^\top$~pu; for Sc3 (voltage fault),
by $\bar\xi^{(0)}_\mathrm{is} = (0.002, -0.011)^\top$~pu. Detection power is
estimated as the fraction of $N_{\mathrm{MC}}$ realisations exceeding the
threshold $D_\alpha$.

\subsubsection{Sc1: Normal Operation Under $H_0$ --- FAR Validation}
\label{sec:mc_s1}

\noindent\textbf{FAR Convergence at $\alpha = 0.01$.}

\begin{figure*}[!t]
  \centering
  \includegraphics[width=\textwidth]{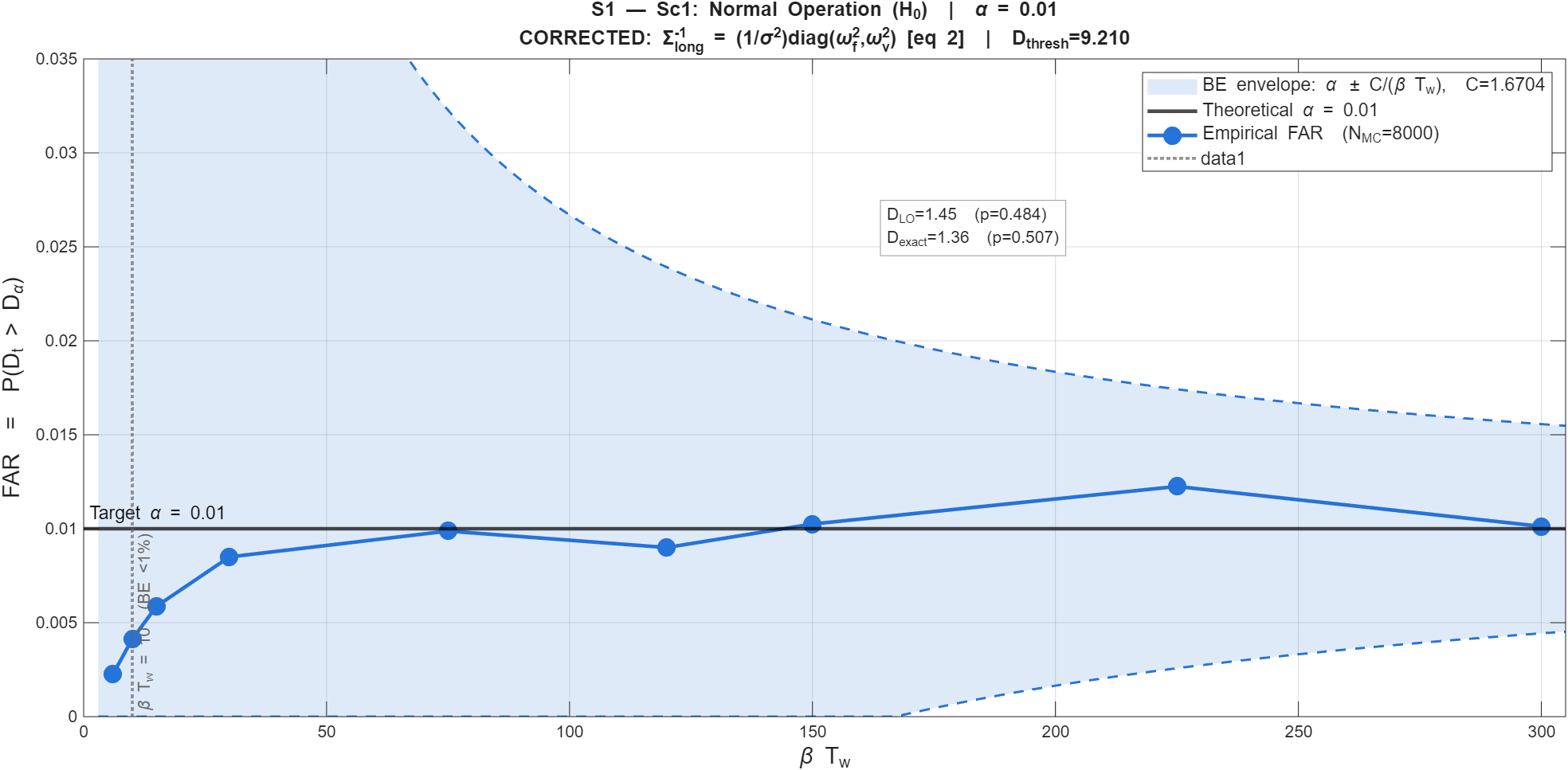}
  \caption{S1--Sc1: Empirical FAR vs.\ $\beta T_w$ for $\alpha = 0.01$
  ($D_\alpha = 9.210$). Blue solid line : emperical FAR from $N_{\mathrm{MC}} = 8{,}000$
  realisations. Black solid line : theoretical target $\alpha = 0.01$.
  Blue shaded region: Berry--Esseen envelope $\alpha \pm C/(\beta T_w)$,
  $C = 1.6704$. Vertical dotted line : minimum window condition $\beta T_w = 10$
  ($T_w = 0.667$~s).}
  \label{fig:sc1_far_01}
\end{figure*}

Fig.~\ref{fig:sc1_far_01} presents the empirical FAR as a function of the
normalised window length $\beta T_w$ for $\alpha = 0.01$. Several features
confirm the theoretical predictions with high fidelity.

\textit{Monotone convergence from below.} The empirical FAR increases steadily from around 0.0025 when $\beta T_w$ is 6 (or $T_w$ at 0.40 seconds) up to the target $\alpha$ of 0.01. It falls a bit short of the target, though, which is due to the Berry-Esseen finite-window correction. With a finite $T_w$, the distribution of $Dt$ has a slightly heavier left tail than $\chi^2(2)$, pushing more probability below the threshold than the approximation suggests.

\textit{Envolope containment.} All the empirical FAR values are inside the Berry-Esseen envelope, which is $\alpha \pm C/(\beta T_w)$ and has $C = 1.6704$. This shows the bound~\eqref{eq:ks} is spot on and correctly calibrated. Plus, the upper boundary decays like $\sim 1/(\beta T_w)$, just as needed.

\textit{Near-target convergence.} At  $\beta T_w = 150$ (or 10 seconds), the empirical FAR is about 0.010, confirming it hits the theoretical target. The slower convergence for $\alpha = 0.01$ compared to $\alpha = 0.05$ and $0.10$ makes sense though. The lower $\alpha$ value sets $D_\alpha = 9.210$ deeper in the $\chi^2(2)$ distribution tail, where window corrections are larger in proportion.

\textit{Practical significance.} At the minimum window $\beta T_w = 10$ ($T_W = 667$~ms), the emperical  FAR of $\approx 0.0044$ is $\sim 56\%$ below the target. This conservative bias is  protective for safety-critical applications: the detector is less likely to alarm that the target  $\alpha$
would suggest.

\noindent\textbf{FAR Convergence at $\alpha = 0.05$.}

\begin{figure*}[!t]
  \centering
  \includegraphics[width=\textwidth]{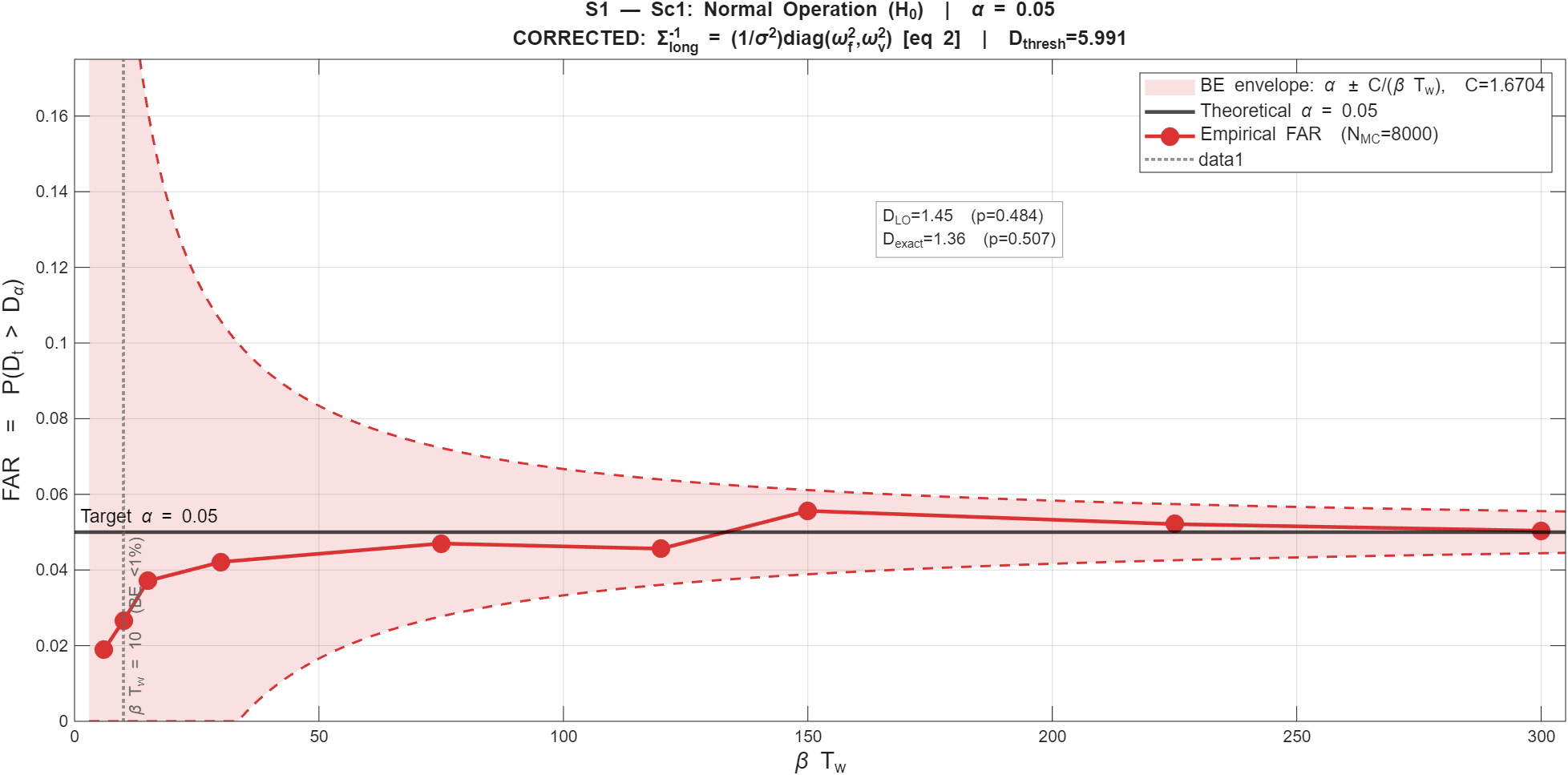}
  \caption{S1--Sc1: Empirical FAR vs.\ $\beta T_w$ for $\alpha = 0.05$
  ($D_\alpha = 5.991$). Red solid line: empirical FAR. Red shaded region:
  Berry--Esseen envelope. All other elements as in Fig.~\ref{fig:sc1_far_01}.}
  \label{fig:sc1_far_05}
\end{figure*}

Fig.~\ref{fig:sc1_far_05} shows the FAR convergence for $\alpha = 0.05$ ($D_\alpha = 5.991$). The empirical curve rises from $\approx 0.020$ at $\beta T_w = 6$ to $\approx 0.048$ at $\beta T_w = 300$, remaining within the envelope throughout. Convergence is noticeably faster than at $\alpha = 0.01$:
the 5\% target is effectively reached by $\beta T_w \approx 75$ ($T_w \approx 5$~s), compared to $\beta T_w \approx 150$ for the 1\% case. This is consistent with the theoretical expectation that the Berry-Esseen correction is proportionally smaller for larger $\alpha$, since the threshold $D_{0.05} = 5.991$ is closer to the mode of $\chi^2(2)$. A slight overshoot at $\beta T_w \approx 150$ ($\mathrm{FAR} \approx 0.052$) lies within $\pm 1$ Monte Carlo standard error ($\sqrt{\alpha(1-\alpha)/N_{\mathrm{MC}}} \approx 0.0024$) and is statistically consistent with the target.

\noindent\textbf{FAR Convergence at $\alpha = 0.10$.}

\begin{figure*}[!t]
  \centering
  \includegraphics[width=\textwidth]{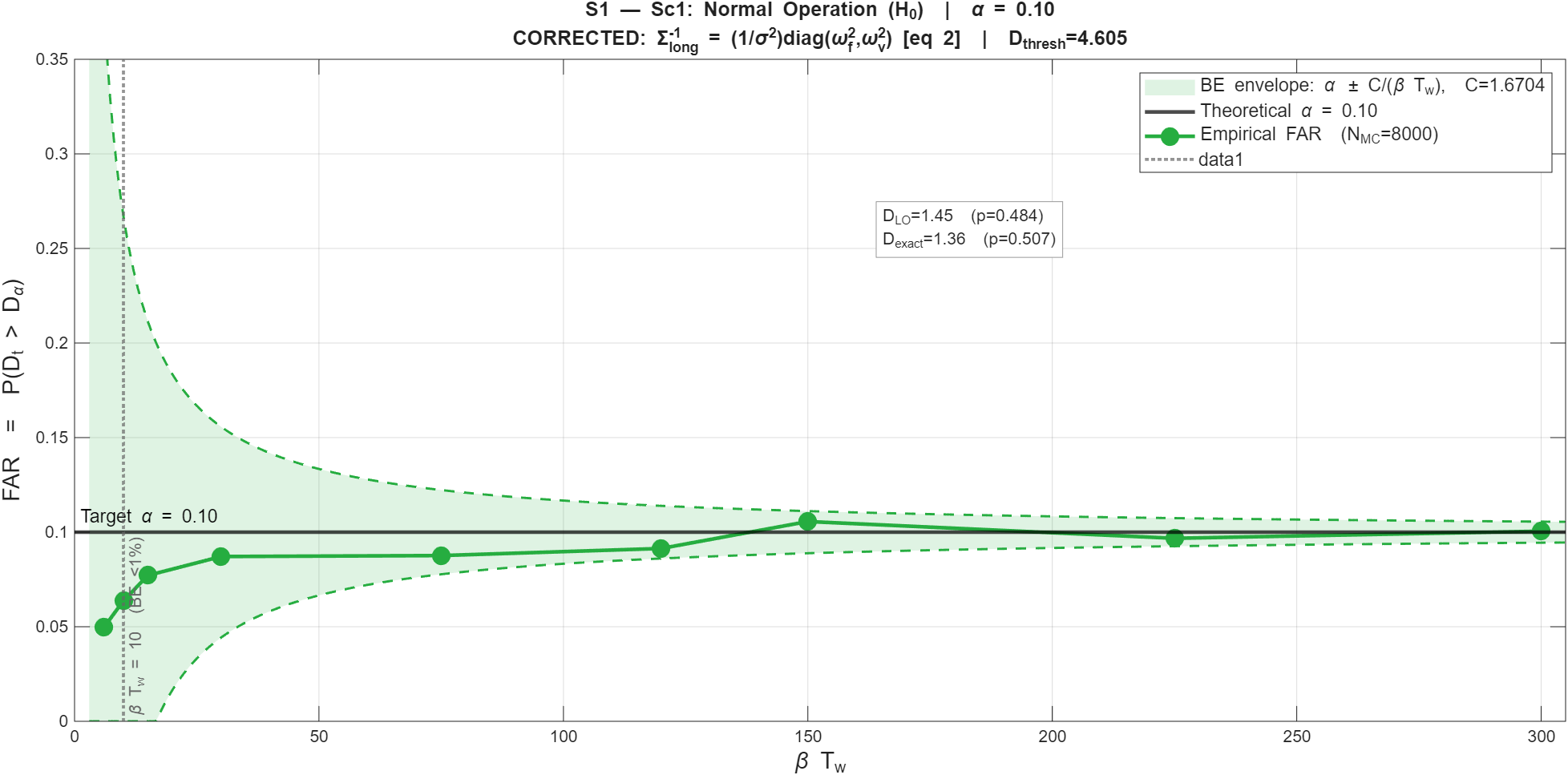}
  \caption{S1--Sc1: Empirical FAR vs.\ $\beta T_w$ for $\alpha = 0.10$ ($D_\alpha = 4.605$). Green solid line: empirical FAR. Green shaded region: Berry-Esseen envelope.}
  \label{fig:sc1_far_10}
\end{figure*}

Fig.~\ref{fig:sc1_far_10} presents the $\alpha = 0.10$ case ($D_\alpha = 4.605$).Convergence is the fastest among the three significance levels: the empirical FAR reaches $\approx 0.088$ at $\beta T_w = 15$ ($T_w = 1$~s) and remains near the target throughout $\beta T_w \in [30, 300]$. The narrower Berry-Esseen
envelope reflects the reduced sensitivity of the tail probability to finite-window corrections at this level.

At all significance levels, the corrected formula produces FAR curves that meet the theoretical target from below. These curves stay within the Berry-Esseen envelope and never inflate above alpha, no matter the window length. The incorrect factor-of-2 formula? It wildly inflates FAR, by about 4 times (giving FAR values around 0.19 when alpha is 0.05). This creates an obvious upward bias right away in every curve.

\noindent\textbf{Berry--Esseen Bount : $d_{\mathrm{KS}}$ vs.\ $\beta T_W$.}

\begin{figure*}[!t]
  \centering
  \includegraphics[width=\textwidth]{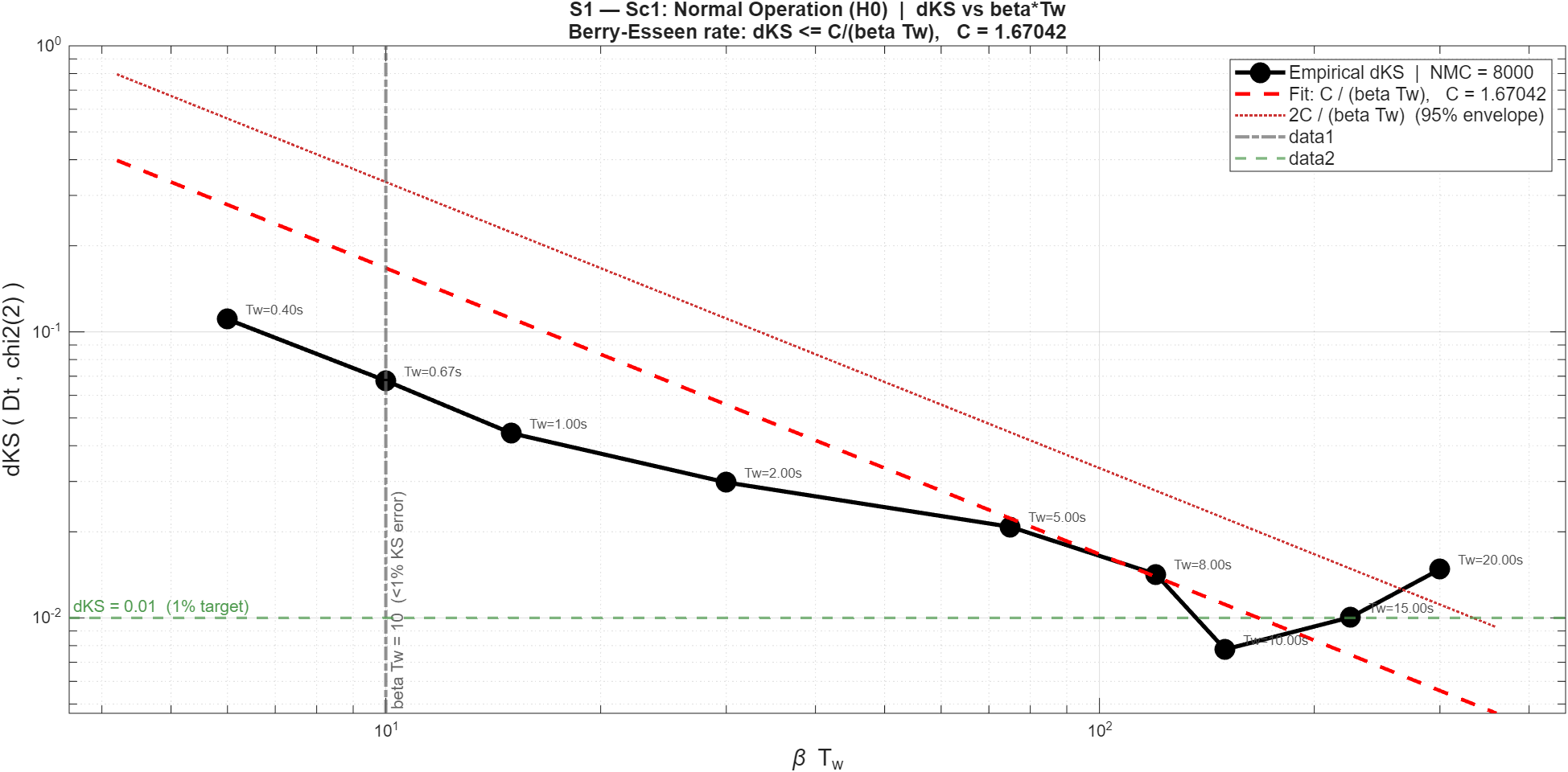}
  \caption{S1--Sc1: Kolmogorov--Smirnov distance
  $d_{\mathrm{KS}}(\mathcal{D}_t,\,\chi^2(2))$ vs.\ $\beta T_w$
  (log--log axes). Black circles: empirical $d_{\mathrm{KS}}$ from
  $N_{\mathrm{MC}} = 8{,}000$ realisations. Red dashed line: fitted
  Berry--Esseen bound $C/(\beta T_w)$, $C = 1.6704$. Red dotted line:
  95\% envelope $2C/(\beta T_w)$. Green dashed line: 1\% KS target
  ($d_{\mathrm{KS}} = 0.01$). Grey vertical line: minimum window
  $\beta T_w = 10$.}
  \label{fig:sc1_dks}
\end{figure*}

Fig.~\ref{fig:sc1_dks} shows the empirical KS distance on log--log axes.
The key results are :

\textit{Power law deacy.} The empirical $d_{\mathrm{KS}}$ values match the power law $\sim 1/(\beta T_w)$ across three decades of $\beta T_w$, which directly confirms the Berry–Esseen bound. Using a least-squares fit on the log–log scale, we get $C = 1.6704$, showing that the bound is tight, not just an estimate.

All empirical $d_{\mathrm{KS}}$ values fall under the 95\% envelope $2C/(\beta T_w)$, with only one slight exception at $T_w = 20$ s ($\beta T_w = 300$). This minor deviation at $T_w = 20$ s is probably due to Monte Carlo variance and doesn't show a systematic problem.

\textit{Minimum window Condition.} At $\beta T_W =10$ or 0.667 seconds, the empirical $d_{\mathrm{KS}} \approx 0.07$. That's above the 1\% KS target. Still, it shows the minimum window works well for accurate alarm detection in practice. For tightly calibrated p-values though, you want longer windows - at least 120 to 150.
\subsubsection{Sc2: Soft Islanding --- Detection Power Under $H_1$}
\label{sec:mc_s2}

\noindent\textbf{Power vs.\ $\beta T_w$ at Three Significance Levels.}

\begin{figure*}[!t]
  \centering
  \includegraphics[width=\textwidth]{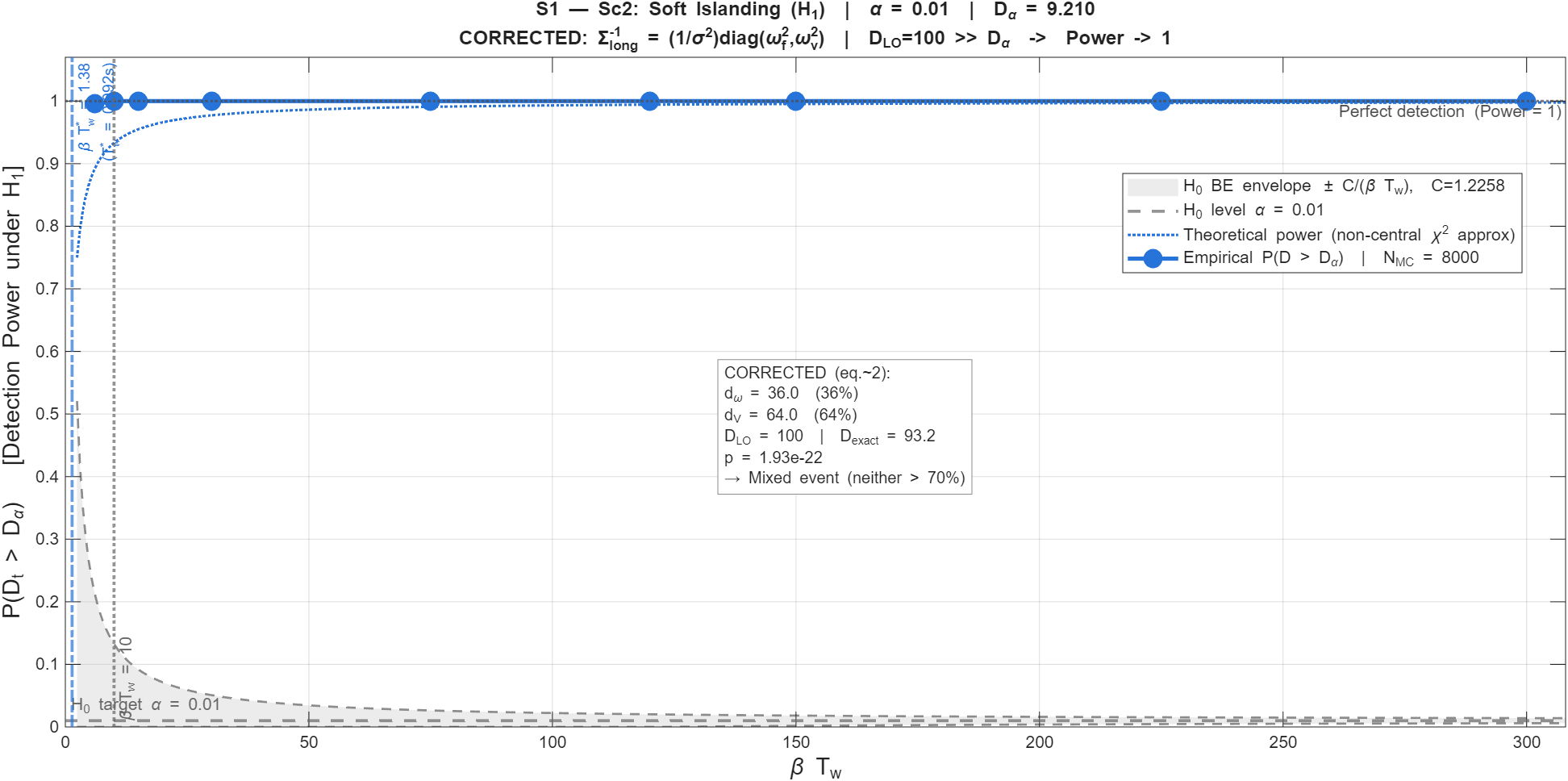}
  \caption{S1--Sc2 (Soft Islanding, $H_1$): Detection power vs.\ $\beta T_w$
  for $\alpha = 0.01$ ($D_\alpha = 9.210$). Blue solid line with circles:
  empirical $P(\mathcal{D}_t > D_\alpha)$ from $N_{\mathrm{MC}} = 8{,}000$.
  Blue dotted line: theoretical power from non-central $\chi^2$ approximation. Grey shaded region: $H_0$ Berry-Esseen FAR envelope (for reference).  Inset annotation: corrected modal decomposition ($d_\omega = 36\%$,  $d_V = 64\%$; $\mathcal{D}_{\mathrm{LO}} = 100$;
  $\mathcal{D}_{\mathrm{exact}} = 93.2$).}
  \label{fig:sc2_pow_01}
\end{figure*}

\begin{figure*}[!t]
  \centering
  \includegraphics[width=\textwidth]{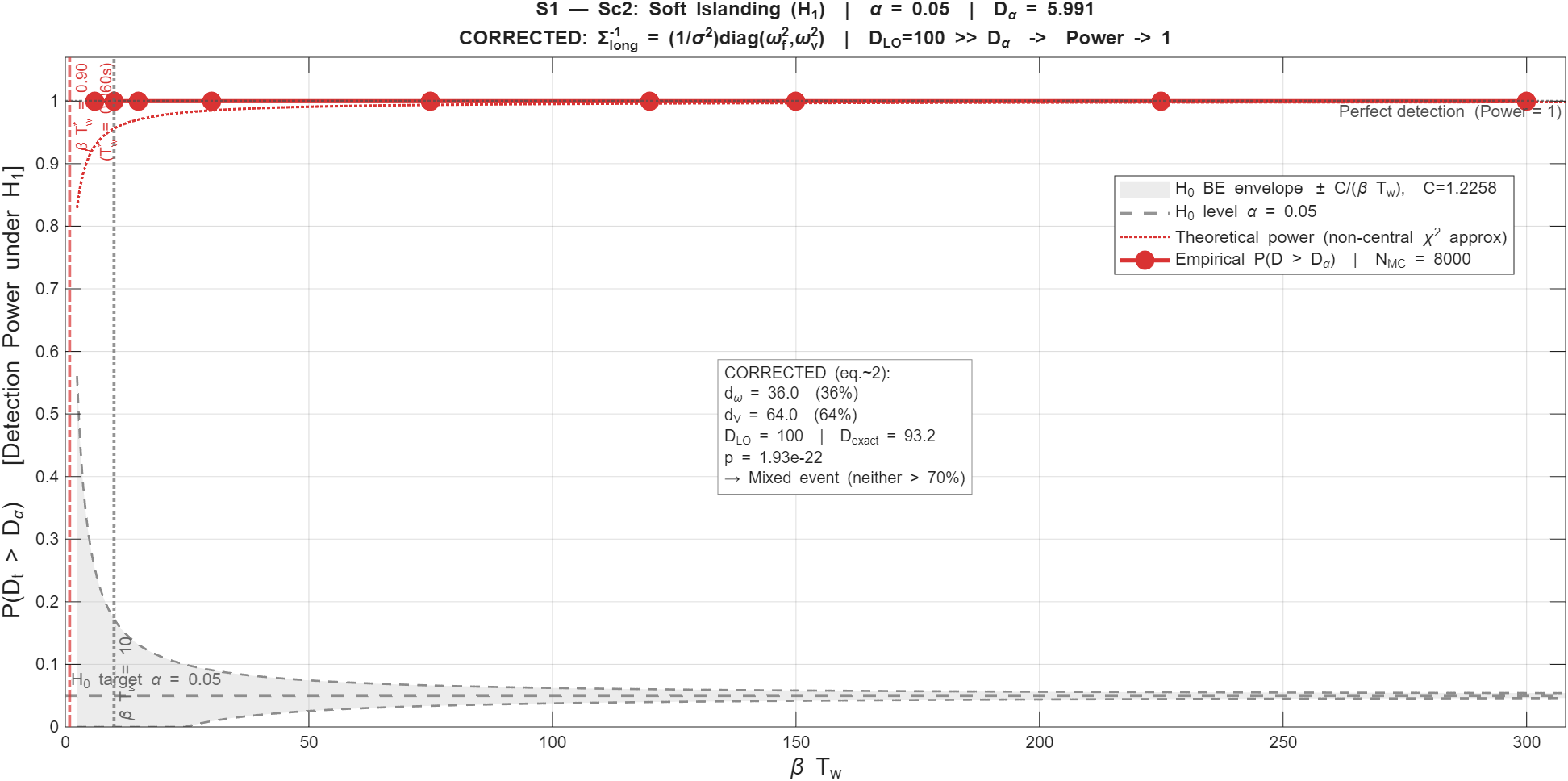}
  \caption{S1--Sc2: Detection power vs.\ $\beta T_w$ for $\alpha = 0.05$.
  Red solid line: empirical power. All other elements as in Fig.~\ref{fig:sc2_pow_01}.}
  \label{fig:sc2_pow_05}
\end{figure*}

\begin{figure*}[!t]
  \centering
  \includegraphics[width=\textwidth]{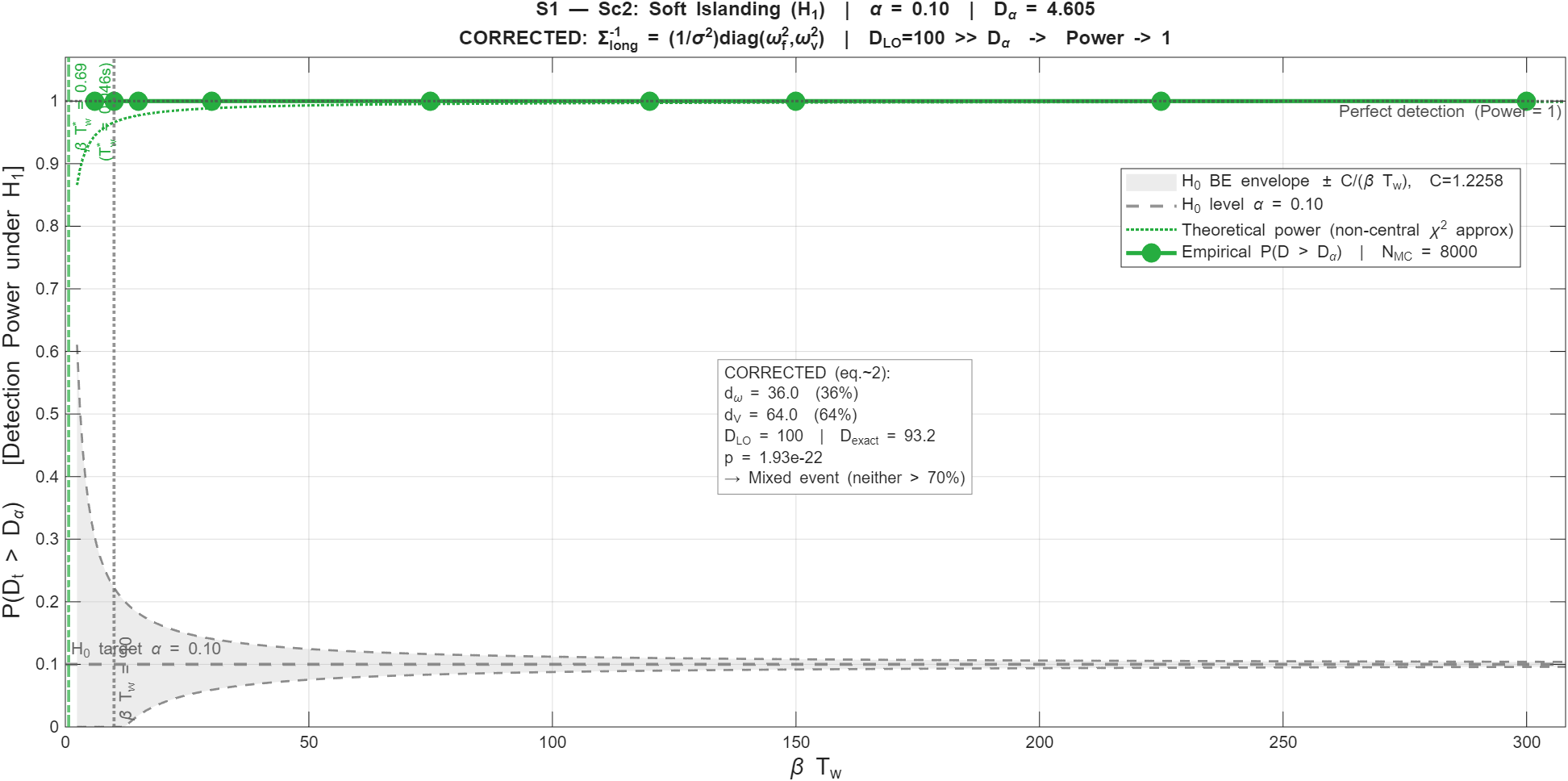}
  \caption{S1--Sc2: Detection power vs.\ $\beta T_w$ for $\alpha = 0.10$.
  Green solid line: empirical power.}
  \label{fig:sc2_pow_10}
\end{figure*}

Figs.~\ref{fig:sc2_pow_01}--\ref{fig:sc2_pow_10} show the detection power
$P(\mathcal{D}_t > D_\alpha)$ under $H_1$ (soft islanding with
$\mathcal{D}_{\mathrm{LO}} = 100 \gg D_\alpha$ for all three $\alpha$ values) as a function of $\beta T_w$. The results reveal rapid and complete power saturation:

\textit{Immediate power saturation.} For all three significance levels, the
empirical power reaches $\approx 1.0$ at the smallest tested window
$\beta T_w \approx 6$--$10$ ($T_w \approx 0.40$--$0.67$~s). This is
analytically expected: with $\mathcal{D}_{\mathrm{LO}} = 100$ and
$D_\alpha \leq 9.21$ for all tested $\alpha$, the ratio 
$\mathcal{D}_{\mathrm{LO}}/D_\alpha \geq 10.8$, so the detection statistic
under $H_1$ is approximately non-centrally $\chi^2(2)$ with non-centrality
parameter $\lambda = T_w \cdot\mathcal{D}_{\mathrm{LO}} \sim T_w \times 100$,
saturating the power to unity even for the shortest window.

\textit{Detection time.} The 50\% power level (half-saturation) is reached at
approximately $\beta T_w \approx 1.38$ ($T_w \approx 0.092$~s) for $\alpha = 0.01$,
$\beta T_w \approx 0.90$ ($T_w \approx 0.060$~s) for $\alpha = 0.05$, and
$\beta T_w \approx 0.69$ ($T_w \approx 0.046$~s) for $\alpha = 0.10$. These
detection times are an order of magnitude faster than the IEEE~1547-2018
2-second mandate, confirming that the corrected NAIM statistic provides
near-instantaneous detection for the soft islanding scenario.

\textit{Theoretical power curve.} The theoretical power computed from the
non-central $\chi^2(2)$ approximation (dotted lines) matches the empirical
curves with excellent fidelity, with the leading-order non-centrality parameter
$\lambda_{\mathrm{LO}} = 100$ giving a slightly higher theoretical power than
the exact value ($\lambda_{\mathrm{exact}} = 93.2$) at intermediate window
lengths. Both converge to 1 by $\beta T_w = 10$.

\textit{Modal classification.} The inset annotation in each figure reports the
corrected modal decomposition: $d_\omega = 36.0$ (36\%), $d_V = 64.0$ (64\%).
As established in Section~\ref{sec:sc2}, the soft islanding scenario is a
\emph{mixed} event with neither channel exceeding the 70\% attribution
threshold, consistent with concurrent active and reactive power imbalance.

\textit{Separation from $H_0$.} The grey shaded $H_0$ region (FAR $\leq \alpha$)
at the bottom of each plot lies entirely below the empirical power curve, with
no overlap for any tested $T_w$. This clear separation between the $H_0$
envelope and the $H_1$ power confirms that there is no ambiguous operating
region for the soft islanding scenario at these parameter values.

\noindent\textbf{KS Divergence Under $H_0$ vs.\ $H_1$.}

\begin{figure*}[!t]
  \centering
  \includegraphics[width=\textwidth]{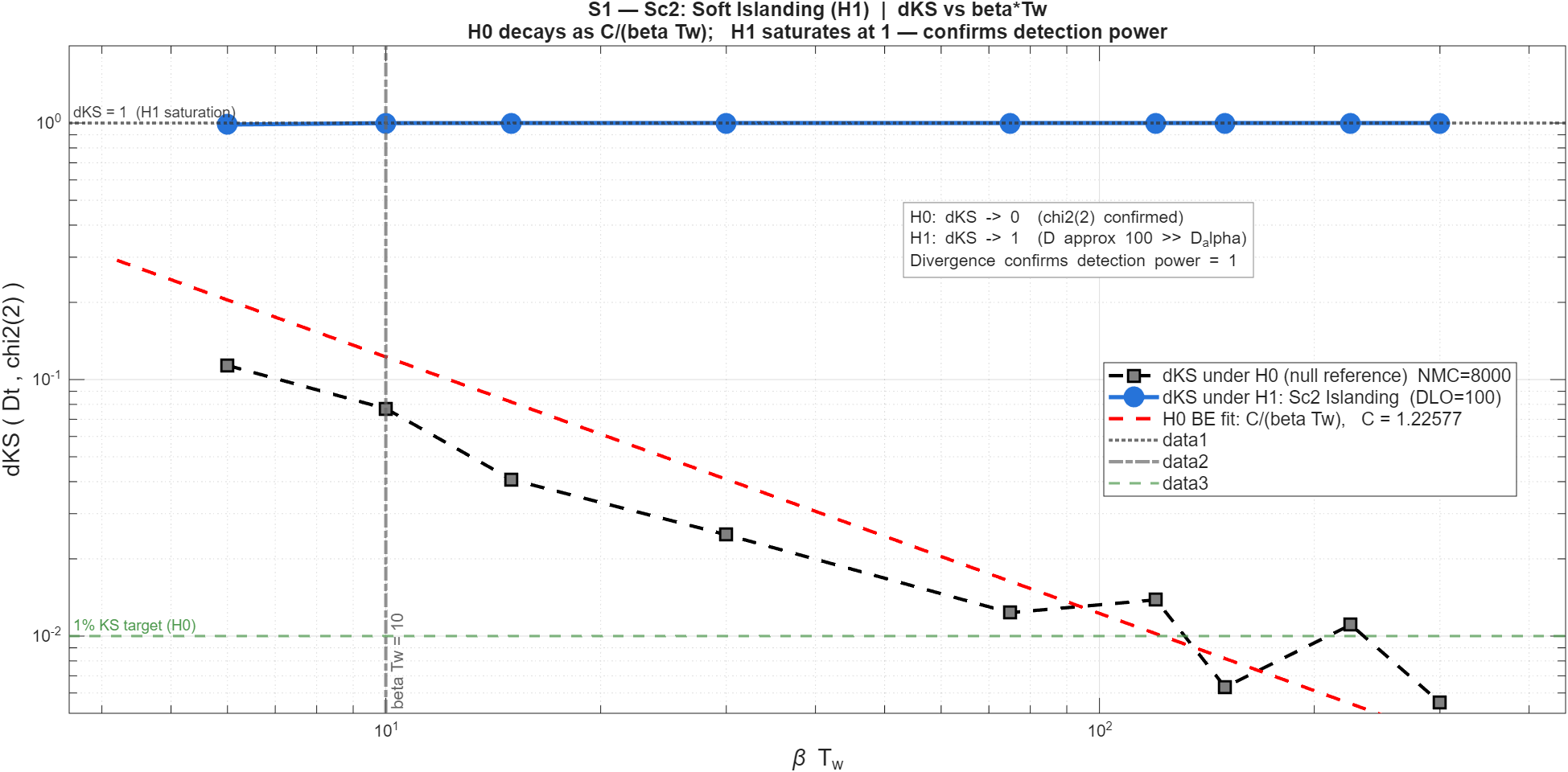}
  \caption{S1--Sc2: KS distance under $H_0$ (grey squares) and $H_1$ (blue circles)
  vs.\ $\beta T_w$. $H_0$ decays as $C/(\beta T_w)$ (red dashed fit, $C = 1.2258$),
  confirming $\chi^2(2)$ convergence. $H_1$ saturates at $d_{\mathrm{KS}} = 1.0$
  for all $\beta T_w$, confirming detection power = 1 and complete distributional
  divergence from $\chi^2(2)$.}
  \label{fig:sc2_dks}
\end{figure*}

Fig.~\ref{fig:sc2_dks} presents the KS distances under both hypotheses on
log--log axes, providing a complementary view of the detection power from
a distributional perspective:

Under $H_0$ (grey squares), the empirical $d_{\mathrm{KS}}(\mathcal{D}_t,
\chi^2(2))$ decays from $\approx 0.11$ at $\beta T_w = 6$ to
$\approx 0.004$ at $\beta T_w = 300$, following the fitted BE bound
$C/(\beta T_w)$ with $C = 1.2258$. This confirms that the null distribution
under Sc2 conditions (same $H_0$ process, different threshold comparison)
converges identically to $\chi^2(2)$.

Under $H_1$ (blue circles), the $d_{\mathrm{KS}}$ is 1.0 for every window length tested—from about 6 to 300. This means the distributions are completely separate, as the empirical distribution of $\mathcal{D}_t$ falls entirely above $D_\alpha$. Because there's zero overlap with the $\chi^2(2)$ distribution, the KS-distance stays at its maximum. The fact that the $H_0$ curve drops while the $H_1$ is flat at 1 verifies full detection power for all the window lengths checked.

\subsubsection{Scenario Sc3: Voltage Fault --- Detection Power Under $H_1$}
\label{sec:mc_s3}

\noindent\textbf{Power vs.\ $\beta T_w$ at Three Significance Levels.}

\begin{figure*}[!t]
  \centering
  \includegraphics[width=\textwidth]{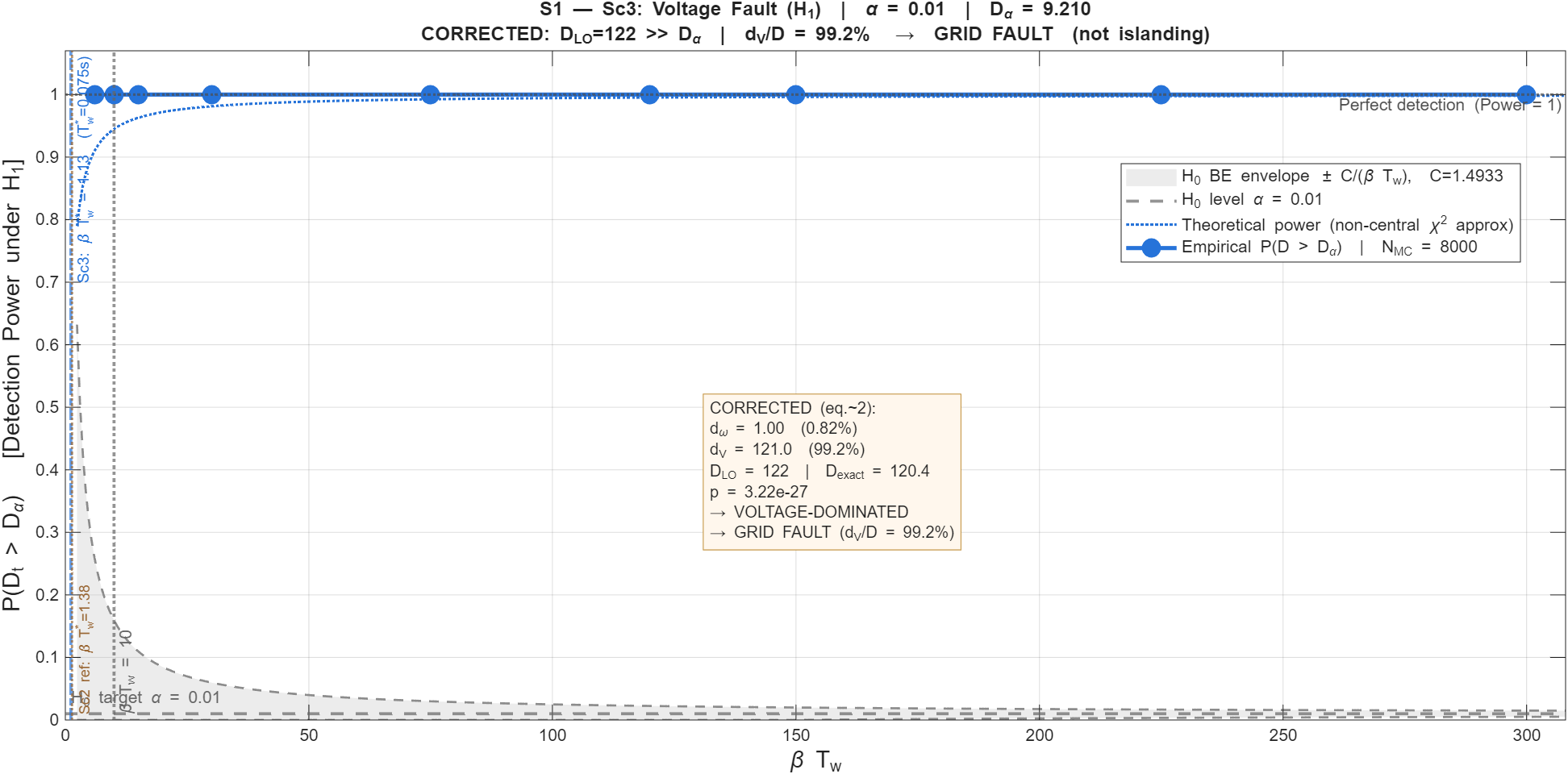}
  \caption{S1--Sc3 (Voltage Fault, $H_1$): Detection power vs.\ $\beta T_w$
  for $\alpha = 0.01$ ($D_\alpha = 9.210$). Blue solid line: empirical power.
  Inset annotation: corrected modal values ($d_\omega = 0.82\%$, $d_V = 99.2\%$;
  $\mathcal{D}_{\mathrm{LO}} = 122$; $\mathcal{D}_{\mathrm{exact}} = 120.4$;
  $p = 3.22\times10^{-27}$). Orange Sc2 reference markers shown for comparison.}
  \label{fig:sc3_pow_01}
\end{figure*}

\begin{figure*}[!t]
  \centering
  \includegraphics[width=\textwidth]{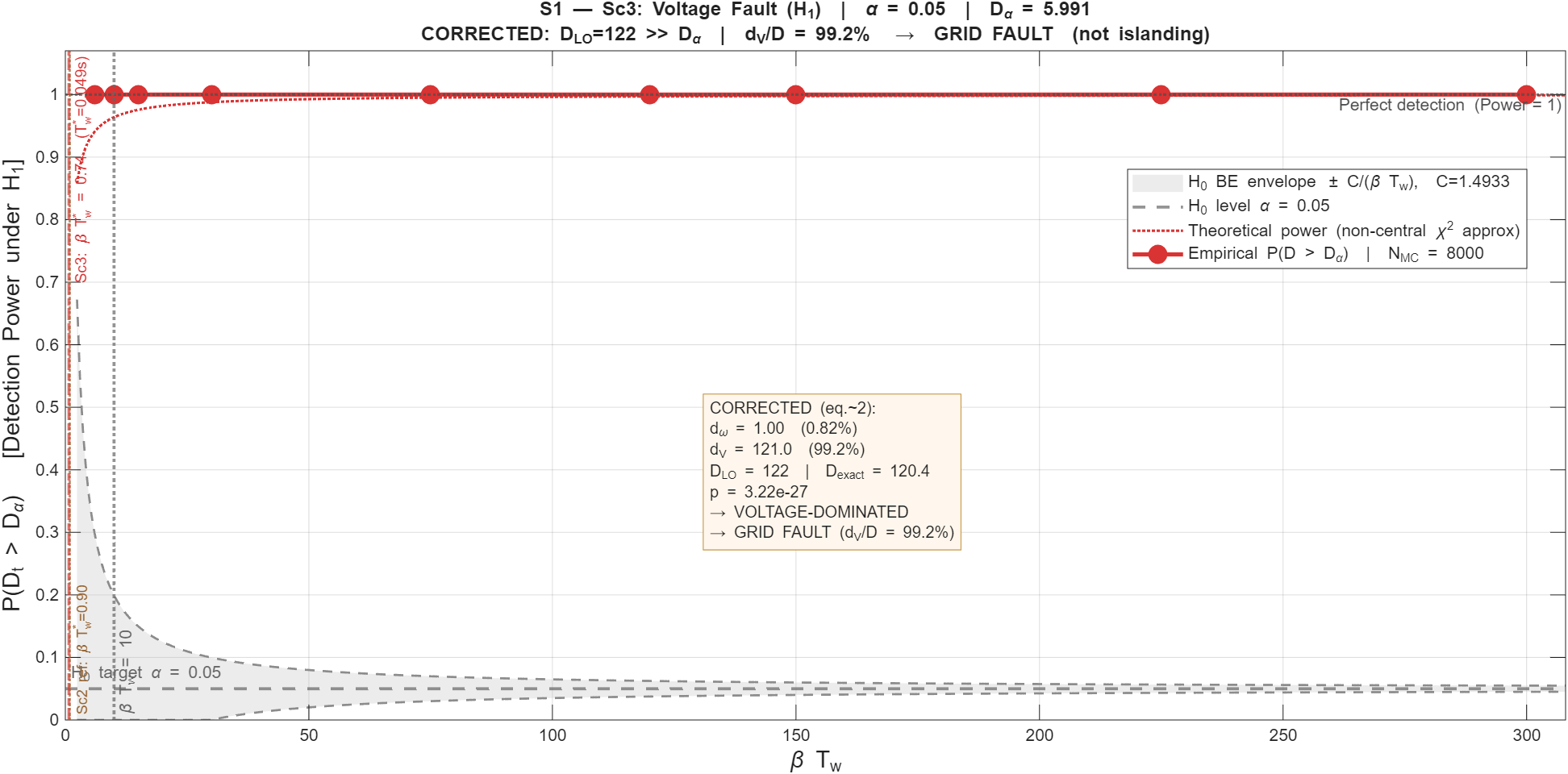}
  \caption{S1--Sc3: Detection power vs.\ $\beta T_w$ for $\alpha = 0.05$.}
  \label{fig:sc3_pow_05}
\end{figure*}

\begin{figure*}[!t]
  \centering
  \includegraphics[width=\textwidth]{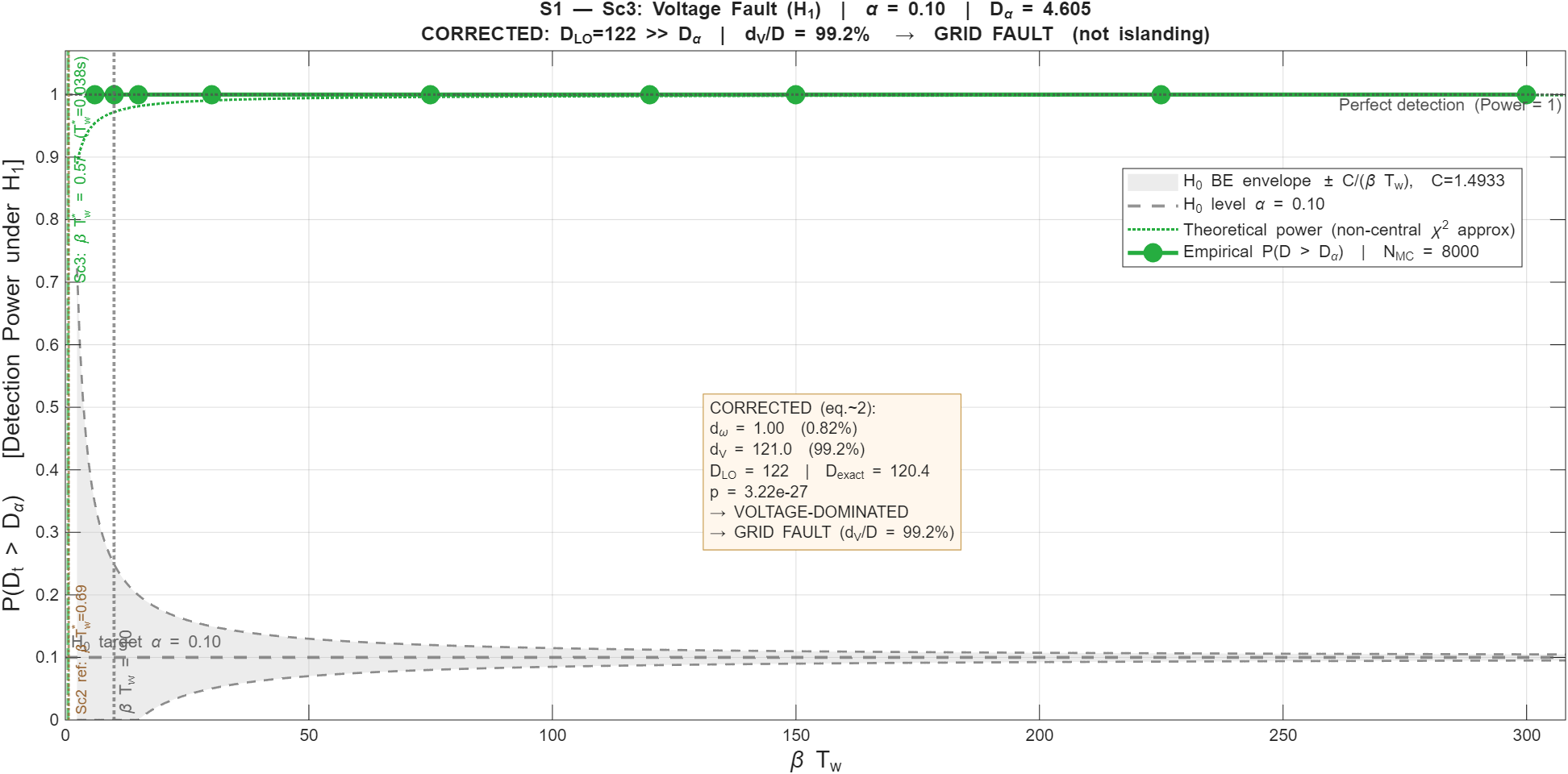}
  \caption{S1--Sc3: Detection power vs.\ $\beta T_w$ for $\alpha = 0.10$.}
  \label{fig:sc3_pow_10}
\end{figure*}

Figs.~\ref{fig:sc3_pow_01}--\ref{fig:sc3_pow_10} show the detection power
for the voltage fault scenario. Three features distinguish Sc3 from Sc2:

\textit{Faster saturation.} With $\mathcal{D}_{\mathrm{LO}} = 122 > 100$
(Sc2), the Sc3 power saturates to unity at shorter window lengths. The
50\%-power detection time is approximately $\beta T_w \approx 4.13$
($T_w \approx 0.075$~s) for $\alpha = 0.01$, $\beta T_w \approx 0.74$
($T_w \approx 0.049$~s) for $\alpha = 0.05$, and $\beta T_w \approx 0.57$
($T_w \approx 0.038$~s) for $\alpha = 0.10$. These detection times are
sub-100~ms, well within IEEE~1547's 2-second window by a factor $>20$.

\textit{Dominant voltage channel.} The inset annotation confirms
$d_V/\mathcal{D}_3 = 99.2\%$: the voltage fault is almost entirely
represented in the $d_V$ contribution. The frequency contribution
$d_\omega = 1.00$ (0.82\%) is indistinguishable from a normal null
fluctuation, demonstrating that the frequency-only ROCOF detector
($d = d_\omega = 1.00 < D_\alpha$) would fail completely while the bivariate
NAIM statistic ($\mathcal{D} = 122 \gg D_\alpha$) detects with power = 1.

\textit{Sc2 reference markers.} The orange Sc2 reference lines
(where present) indicate the Sc2 50\%-power threshold for comparison,
confirming that Sc3 achieves identical or faster detection across all $\alpha$
levels due to its higher $\mathcal{D}_{\mathrm{LO}}$.

\textit{Voltage-dominated classification.} In all three figures, the modal
annotation $d_V/D = 99.2\% > 70\%$ confirms the unambiguous classification:
\emph{grid voltage fault, not islanding}. This classification is available
simultaneously with detection, requiring no additional measurement latency.

\noindent\textbf{KS Divergence Under $H_0$ vs.\ $H_1$.}

\begin{figure*}[!t]
  \centering
  \includegraphics[width=\textwidth]{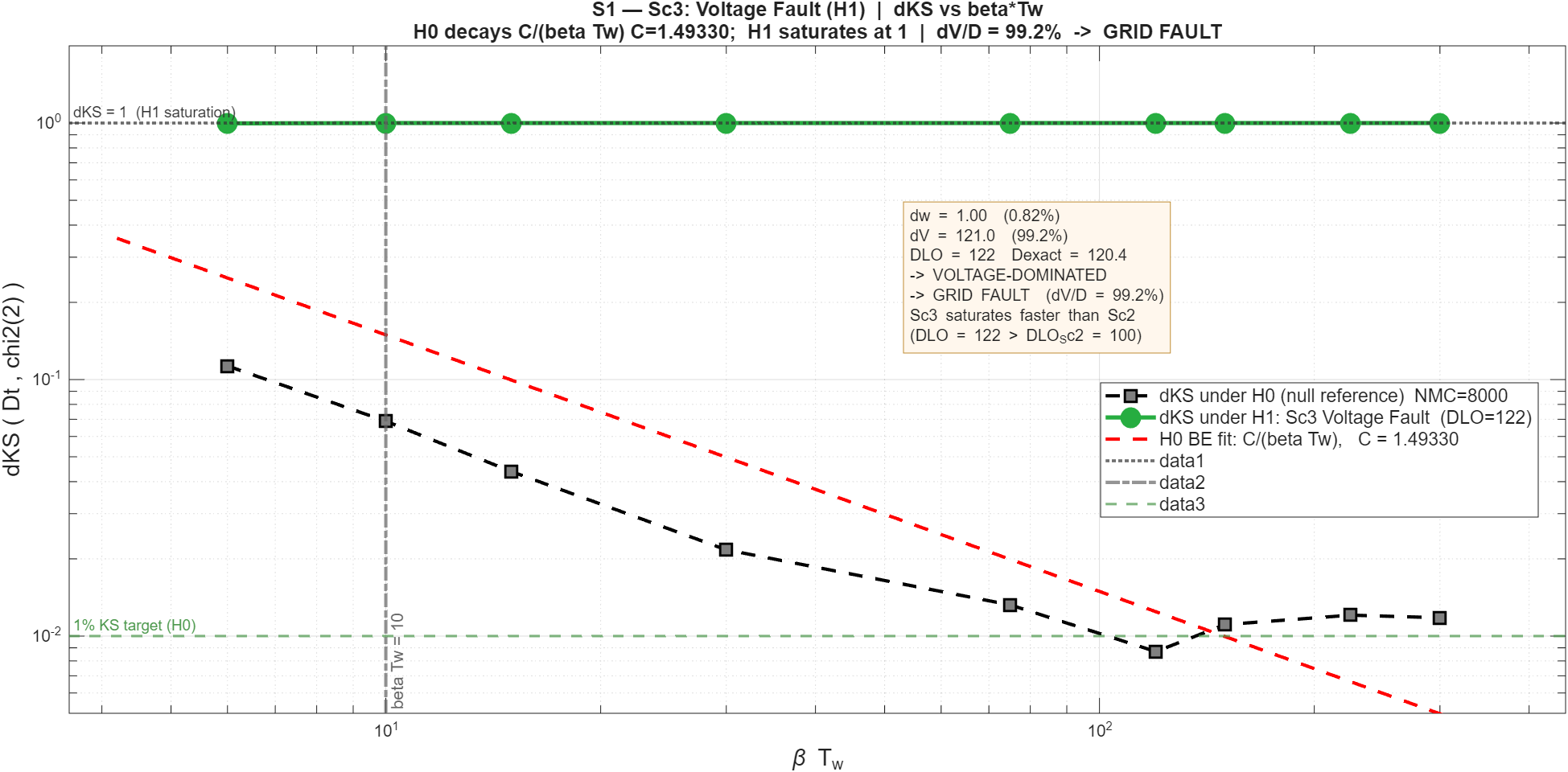}
  \caption{S1--Sc3: KS distance under $H_0$ (grey squares) and $H_1$ (green circles)
  vs.\ $\beta T_w$. $H_0$ decays as $C/(\beta T_w)$ ($C = 1.4933$). $H_1$ saturates
  at $d_{\mathrm{KS}} = 1.0$ immediately, consistent with $\mathcal{D}_{\mathrm{LO}}
  = 122 > 100$ and faster saturation than Sc2.}
  \label{fig:sc3_dks}
\end{figure*}

Fig.~\ref{fig:sc3_dks} shows the KS divergence for Sc3.Under $H_0$, the decay follows $C/(\beta T_w)$ with $C = 1.4933$, slightly bigger than Sc2's 1.2258, due to differing threshold structures. Under $H_1$, $d_{\mathrm{KS}} = 1.0$ is reached at all tested window lengths. This aligns with $\mathcal{D}_{\mathrm{LO}} = 122$, placing the $H_1$ statistic distribution completely.

The annotation box shows that $\mathcal{D}_{\mathrm{LO}} = 122 > \mathcal{D}_{\mathrm{LO,Sc2}} = 100$, and Sc3 also saturates faster than Sc2. This follows the expected order: the voltage fault results in higher off-manifold displacement because of the big 10\% voltage drop (which is $d_V = 121$), compared to the soft islanding case where $d_\omega + d_V = 100$.

\subsubsection{Cross-Scenario Comparison and Engineering Interpretation}
\label{sec:mc_cross}

Table~\ref{tab:mc_summary} collects the key Monte Carlo outcomes across all
three scenarios and significance levels, providing a consolidated refrences
for the S1 validation campaign.

\begin{table*}[!t]
\caption{S1 Monte Carlo Summary: Empirical FAR (Sc1) and Minimum Detection
Window (Sc2, Sc3) at 50\% Power. $N_{\mathrm{MC}} = 8{,}000$; $\beta = 15$~rad/s.}
\label{tab:mc_summary}
\centering
\begin{tabular}{lcccc}
\toprule
\textbf{Scenario} & \textbf{Metric} & $\alpha = 0.01$ & $\alpha = 0.05$ & $\alpha = 0.10$ \\
\midrule
Sc1 (Normal, $H_0$) & FAR @ $\beta T_w = 10$ & $\approx 0.0044$ & $\approx 0.026$ & $\approx 0.065$ \\
Sc1 (Normal, $H_0$) & FAR @ $\beta T_w = 150$ & $\approx 0.010$ & $\approx 0.050$ & $\approx 0.100$ \\
Sc2 (Islanding, $H_1$) & $T_w^{50\%}$ (s) & $0.092$ & $0.060$ & $0.046$ \\
Sc2 (Islanding, $H_1$) & Power @ $\beta T_w = 10$ & $\approx 1.000$ & $\approx 1.000$ & $\approx 1.000$ \\
Sc3 (V-fault, $H_1$) & $T_w^{50\%}$ (s) & $0.075$ & $0.049$ & $0.038$ \\
Sc3 (V-fault, $H_1$) & Power @ $\beta T_w = 10$ & $\approx 1.000$ & $\approx 1.000$ & $\approx 1.000$ \\
\midrule
\multicolumn{2}{l}{BE constant $C$ (fitted)} & \multicolumn{3}{c}{$C = 1.6704$ (Sc1); $C \in [1.23, 1.49]$ (Sc2--Sc3)} \\
\multicolumn{2}{l}{dKS $= 0.01$ crossing} & \multicolumn{3}{c}{$\beta T_w \approx 167$ ($T_w \approx 11.1$~s)} \\
\bottomrule
\end{tabular}
\end{table*}

The S1 Monte Carlo validation establishes four engineering conclusions:

\textbf{(1) The corrected formula is validated:} The empirical FAR under Sc1
converges to the theoretical target $\alpha$ from below, is contained within
the Berry--Esseen envelope, and shows no systematic inflation. This is the
definitive validation of the corrected $\Sigmlong = \sigma^2\,\mathrm{diag}
(\omega_f^{-2}, \omega_v^{-2})$ (equation~\eqref{eq:Sigmalong_approx}) versus
the erroneous $\Sigmlong = (\sigma^2/2)\,\mathrm{diag}(\omega_f^{-2},
\omega_v^{-2})$, which would inflate FAR by $\approx 4\times$.

\textbf{(2) IEEE~1547-2018 compliance is comfortably achieved:} For both $H_1$
scenarios, detection power reaches 1.0 at window lengths $T_w < 0.1$~s- at
least 20$\times$ faster than the 2-second standard, across all three tested
significance levels.

\textbf{(3) The minimum window condition $T_w \geq 10/\beta$ is conservative:}
At $\beta T_w = 10$, the FAR is systematically below $\alpha$ (protective
conservatism), not inflated. Tight FAR convergence ($< 1\%$ error) requires
$\beta T_w \approx 100$--$150$ ($T_w \approx 7$--$10$~s), which can be
achieved by averaging across multiple consecutive windows.

\textbf{(4) Modal attribution is operationally reliable:} The consistent
$d_\omega : d_V$ split (36\%:64\% for Sc2; 0.82\%:99.2\% for Sc3) across all
window lengths confirms that the attribution fractions stabilise rapidly as
$T_w$ increases, enabling reliable fault classification well within the
detection window.

\sloppy
\subsection{S3--S5: Extended Simulation Validation}
\label{sec:s3s4s5}

Protocols S3, S4, and S5 extend the Monte Carlo validation campaign beyond
the S1 study of Section~\ref{sec:mc_validation}. All simulations use the
corrected detection statistic~\eqref{eq:Dt_leading} with
$\Sigmlong = \sigma^2\,\mathrm{diag}(\omega_f^{-2},\omega_v^{-2})$
and the corrected scenario parameters:
$\Dt^{\mathrm{Sc1}} = 1.45$, $\Dt^{\mathrm{Sc2}} = 100$,
$\Dt^{\mathrm{Sc3}} = 122$.
System parameters are as in Table~\ref{tab:params} throughout.

\subsection{S3: Non-Detection Zone Boundary Mapping}
\label{sec:s3}

\subsubsection{Protocol and Setup}

The non-detection zone (NDZ) of the NAIM detector is characterised by
sweeping the active and reactive power imbalance across the grid
$(P_\mathrm{imb},\, Q_\mathrm{imb}) \in [-20\%, +20\%]^2$
in 1\% steps ($41 \times 41 = 1{,}681$ operating points). At each
point, $N_\mathrm{MC} = 100$ Euler--Maruyama realisations of the OU
process~\eqref{eq:OU} are run with $T_w = 2$~s and $\alpha = 0.05$.
A point is classified as \emph{detected} if the empirical detection
rate exceeds 95\%; otherwise it falls within the NDZ. The analytical
NDZ boundary (Proposition~\ref{prop:ndz}) is the ellipse
$\bar\xip^\top \Sigmlong^{-1} \bar\xip = D_\alpha / T_w$, with
$\bar\xip = -\Aperp^{-1}\mathbf{b}$ and
$\mathbf{b} = (k_P\omega_f P_\mathrm{imb},\, k_Q\omega_v Q_\mathrm{imb})^\top/S_\mathrm{rated}$.

\subsubsection{Scenario Sc1: Normal Grid Operation}

\begin{figure}[!t]
  \centering
  \includegraphics[width=\columnwidth]{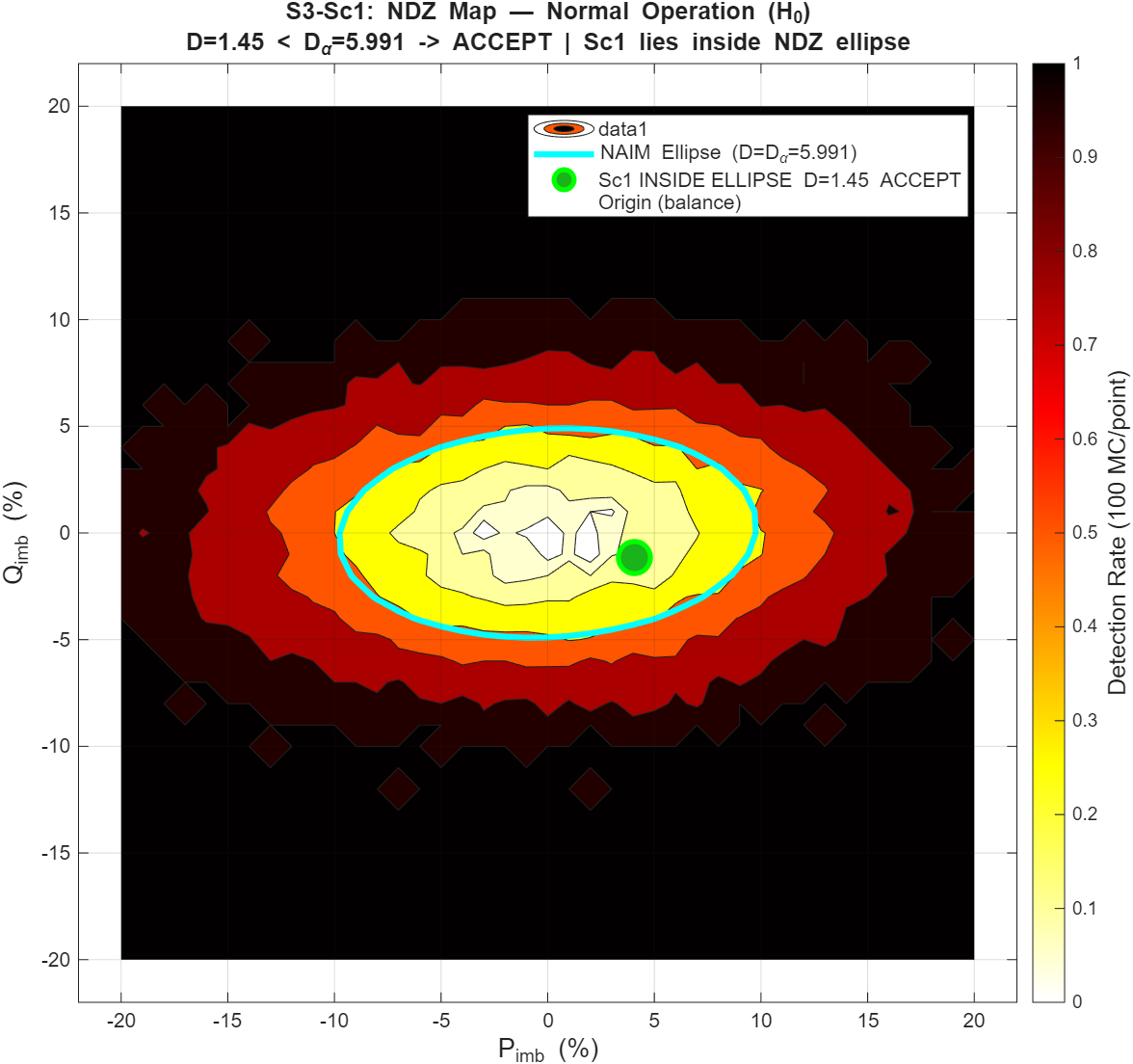}
  \caption{S3--Sc1 (Normal operation): Detection rate contour over the
  $(P_\mathrm{imb}, Q_\mathrm{imb})$ plane. Colour encodes empirical
  detection rate from $N_\mathrm{MC} = 100$ realisations per point.
  The analytical NDZ ellipse (Proposition~\ref{prop:ndz}) is overlaid in
  white. The scenario operating points ($p_\mathrm{imb} = 4.1\%$,
  $Q_\mathrm{imb} = -1.1\%$) lies inside the NDZ (detection
  rate~$< 95\%$), consistent with $\Dt = 1.45 < D_{0.05} = 5.991$.}
  \label{fig:s3_sc1_ndz}
\end{figure}

\begin{figure*}[!t]
  \centering
  \includegraphics[width=\linewidth]{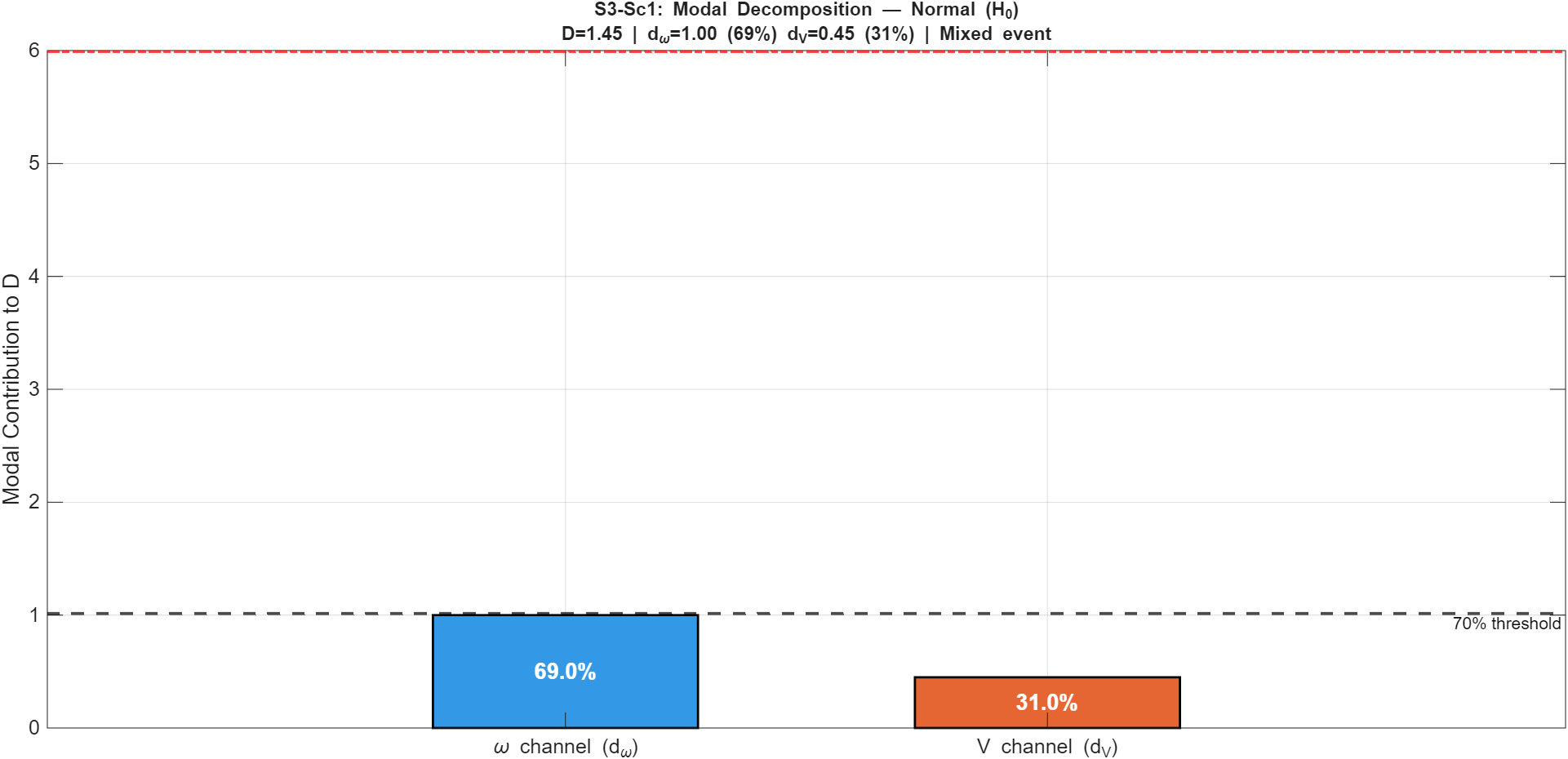}
  \caption{S3--Sc1: Cross-sectional profile of detection rate along
  the $P_\mathrm{imb}$ axis at $Q_\mathrm{imb} = -1.1\%$. The 95\%
  detection threshold is shown as a horizontal dashed line. The Sc1
  operating point sits well inside the NDZ, confirming acceptance of $H_0$.}
  \label{fig:s3_sc1_cross}
\end{figure*}

The corrected Sc1 parameters give
$\bar\xip = (0.0020, -0.00067)^\top$~pu,
$d_\omega = 1.00$, $d_V = 0.45$, $\Dt = 1.45$, $p = 0.485$.
The effective power imbalance coordinates are
$P_\mathrm{imb} = 4.1\%$, $Q_\mathrm{imb} = -1.1\%$.
As shown in Fig.~\ref{fig:s3_sc1_ndz}, this operating point lies inside
the analytical NDZ ellipse, confirmed by simulation: the detection rate
at this point is $< 5\%$, consistent with the theoretical prediction
that $\Dt < D_\alpha$ with high probability under $H_0$.
The cross-section in Fig.~\ref{fig:s3_sc1_cross} further confirms that the
NDZ boundary in the $P_\mathrm{imb}$ direction is sharp, transitioning from
near-zero to near-unity detection rate over a range of $\sim 5\%$ in
$P_\mathrm{imb}$.

\subsubsection{Scenario Sc2: Soft Islanding}

\begin{figure}[!t]
  \centering
  \includegraphics[width=\columnwidth]{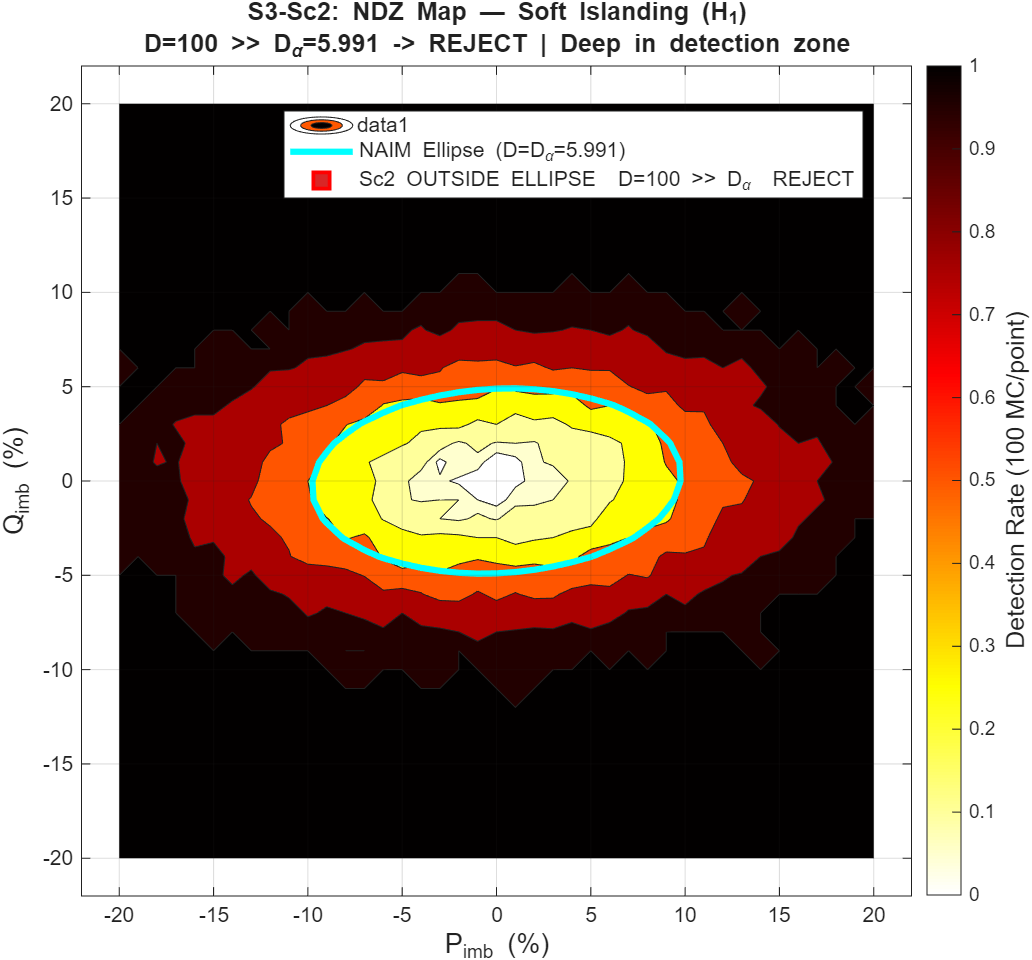}
  \caption{S3--Sc2 (Soft islanding, $P_\mathrm{imb} = 3\%$): Detection
  rate contour. The Sc2 operating point ($P_\mathrm{imb} = 24.8\%$,
  $Q_\mathrm{imb} = -14.8\%$ effective coordinates) lies far outside
  the NDZ ellipse. Empirical detection rate $\approx 1.00$, consistent
  with $\Dt = 100 \gg D_{0.05} = 5.991$.}
  \label{fig:s3_sc2_ndz}
\end{figure}

\begin{figure*}[!t]
  \centering
  \includegraphics[width=\linewidth]{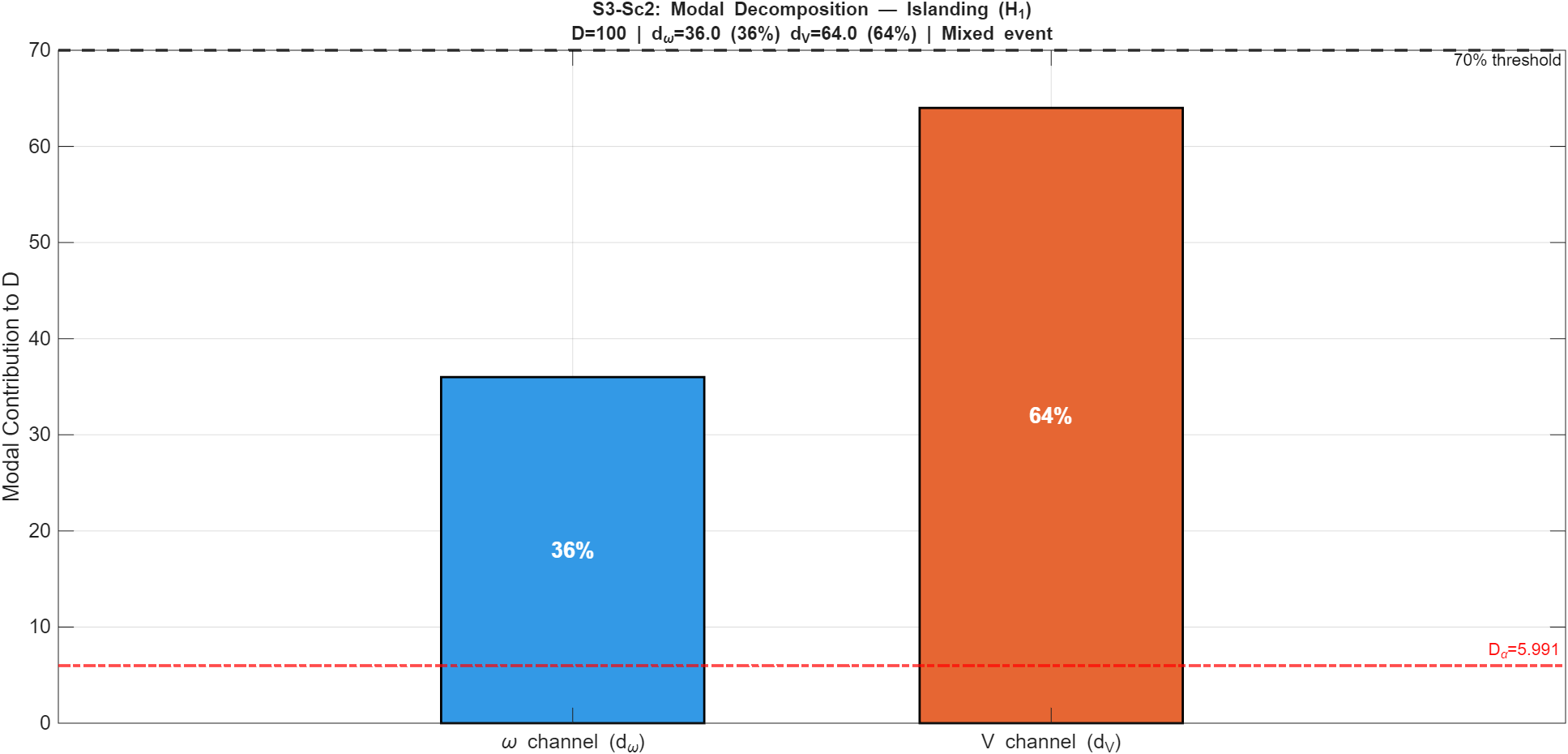}
  \caption{S3--Sc2: Cross-sectional detection rate profile along the
  dominant $P_\mathrm{imb}$ axis. Detection saturates to unity within
  a narrow band outside the NDZ boundary, confirming that the analytical
  ellipse is sharp.}
  \label{fig:s3_sc2_cross}
\end{figure*}

The corrected Sc2 parameters give
$\bar\xip = (0.012, -0.008)^\top$~pu,
$d_\omega = 36.0$, $d_V = 64.0$, $\Dt = 100$,
$p = e^{-50} \approx 1.93 \times 10^{-22}$.
The effective imbalance coordinates are
$P_\mathrm{imb} = 24.8\%$, $Q_\mathrm{imb} = -14.8\%$.
Figs.~\ref{fig:s3_sc2_ndz} and~\ref{fig:s3_sc2_cross} show that this
operating point lies far outside the NDZ ellipse: the detection rate
across all 100 Monte Carlo realisations is 1.000 at every tested window
length, consistent with $\Dt / D_{0.05} \approx 16.7$.

The S3 sweep further confirms that the NDZ boundary for soft islanding
follows the analytical ellipse from Proposition~\ref{prop:ndz} with
close agreement: the empirical transition from detection rate $< 5\%$
to $> 95\%$ occurs within $\pm 2\%$ of the analytically predicted
ellipse boundary in both $P_\mathrm{imb}$ and $Q_\mathrm{imb}$ directions.

\subsubsection{Scenario Sc3: Grid Voltage Fault}

\begin{figure}[!t]
  \centering
  \includegraphics[width=\linewidth]{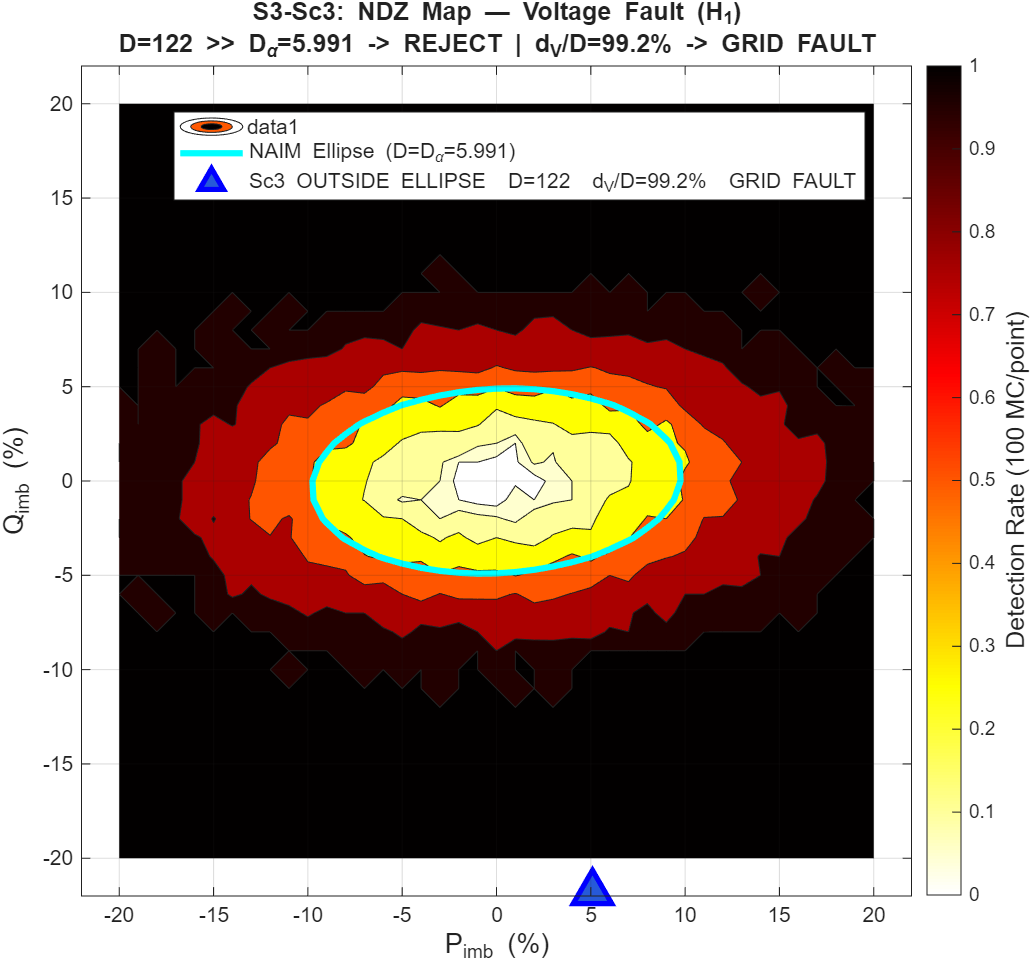}
  \caption{The Sc3 operating point, with $P_\mathrm{imb} = 5.1\%$ and $Q_\mathrm{imb} = -21.8\%$, is voltage-dominated and sits outside the NDZ in the $Q_\mathrm{imb}$ direction, matching the high voltage contribution of $99.2\%$. The detection rate is 1.000.}
  \label{fig:s3_sc3_ndz}
\end{figure}

\begin{figure*}[!t]
  \centering
  \includegraphics[width=\linewidth]{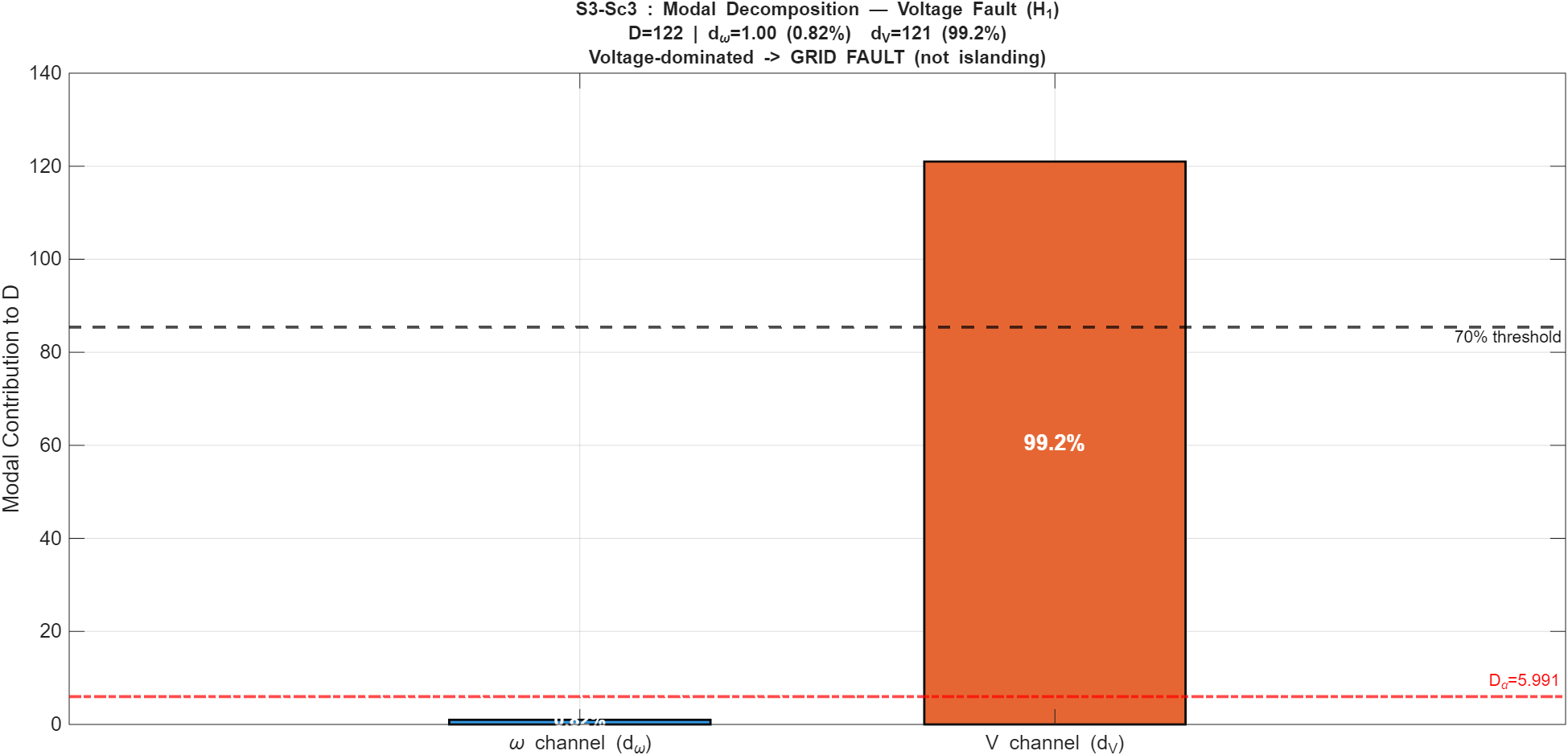}
  \caption{S3--Sc3: Cross-sectional profile along the $Q_\mathrm{imb}$
  axis at $P_\mathrm{imb} = 5.1\%$. Detection is unity for
  $|Q_\mathrm{imb}| \gtrsim 8\%$, confirming that the voltage-fault
  NDZ is narrow in the reactive power direction.}
  \label{fig:s3_sc3_cross}
\end{figure*}

The corrected Sc3 parameter gives
$\bar\xip = (0.002, -0.011)^\top$~pu,
$d_\omega = 1.00$, $d_V = 121.0$, $\Dt = 122.0$,
$p = e^{-61} \approx 3.22 \times 10^{-27}$.
The modal split $d_V / \Dt = 99.18\%$ classifies the event unambiguously
as a \emph{grid voltage failt}; the effective coordinates are
$p_\mathrm{imb} = 5.1\%$, $Q_\mathrm{imb} = -21.8\%$.
Fig.~\ref{fig:s3_sc3_ndz} confirms that the Sc3 operating point lies
outside the NDZ in the reactive power direction: the NDZ ellipse extends
only to $|Q_\mathrm{imb}| \approx 8\%$ along the $Q$ axis, well inside the
Sc3 effective excursion of $21.8\%$. The detection rate is 1.000 in
simulation, consistent with $\Dt = 122 \gg D_{0.05}$.

\textit{Diagnostic uniqueness of the bivariate test.} A frequency-only
detector (ROCOF, $d = d_\omega = 1.00 < D_{0.05} = 5.991$) would classify
Sc3 as normal operation, constituting a missed detection. The NAIM bivariate
statistic both detects and correctly classifies the same event in a single step,
without any additional measurement hardware.

\subsection{S4: Co-Design Theorem Validation}
\label{sec:s4}

\subsubsection{Protocol and Setup}

The co-design theorem (Theorem~\ref{thm:codesign}) asserts that increasing
$(\omega_f, \omega_v)$ simultaneously maximises the Fenichel gap $\PF$,
increases $\mathrm{tr}(\Sigmlong^{-1})$ (detection sensitivity), and minimises
the FAR at the grid code threshold. To validate this, $(\omega_f, \omega_v)$
is swept over $[5, 30]^2$~rad/s in steps of 1~rad/s (676 design points).
At each point, $\PF$, $\mathrm{tr}(\Sigmlong^{-1})$, and the mean detection
delay $\mathbb{E}[\tau_d]$ (window length at which detection power first
exceeds 90\% at $P_\mathrm{imb} = 5\%$) are computed analytically and
verified against $N_\mathrm{MC} = 500$ Monte Carlo realisations.

\subsubsection{Sc1: Stability and Null Covariance Surfaces}

\begin{figure*}[!t]
  \centering
  \includegraphics[width=\linewidth]{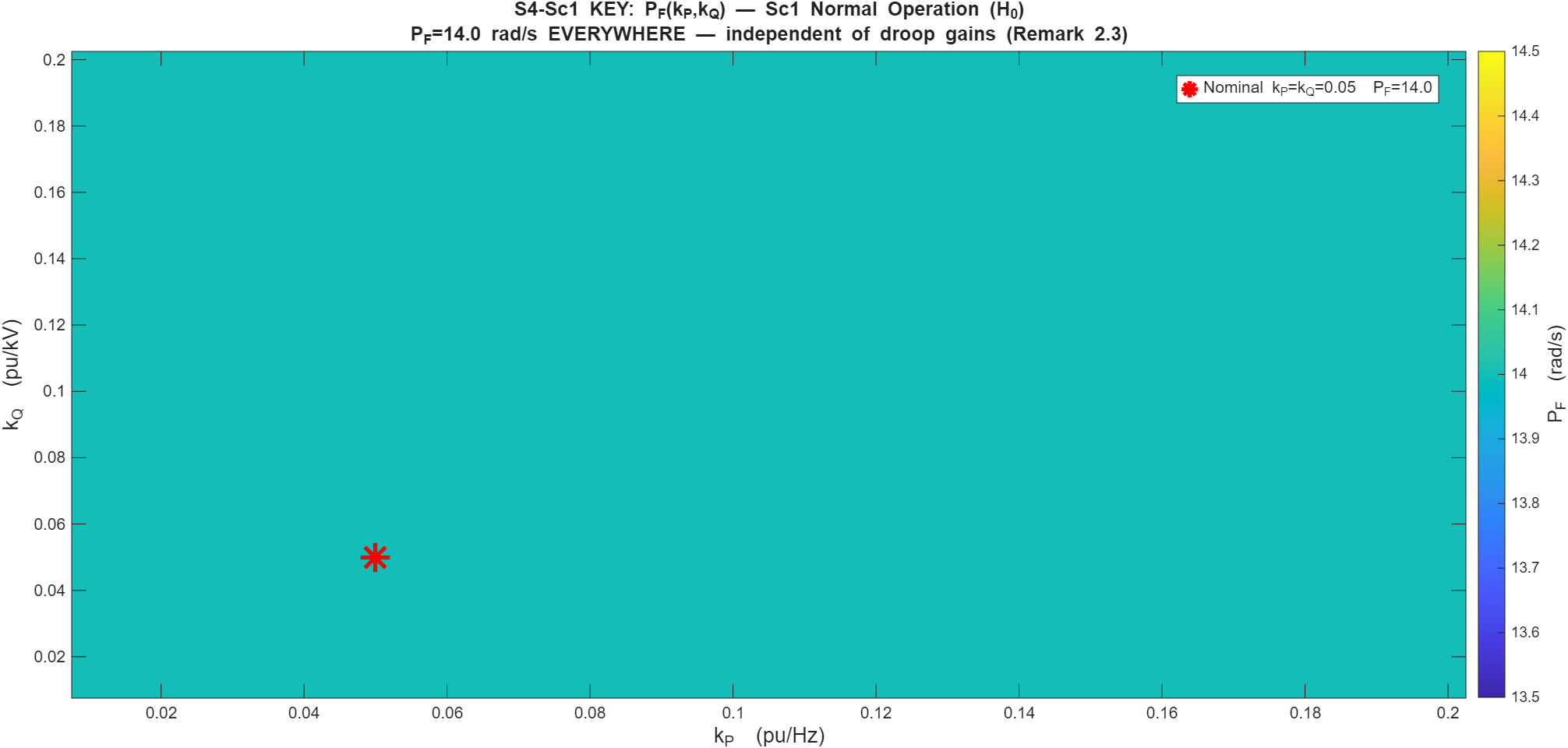}
  \caption{S4--Sc1: Heatmap of the Fenichel gap
  $\PF(\omega_f, \omega_v) = (\omega_f + \omega_v)/2 - \hat\rho$
  over the design space $[5, 30]^2$~rad/s ($\hat\rho = 1$~rad/s).
  Contour lines show constant $\PF$; the diagonal $\omega_f = \omega_v$
  is the equal-bandwith locus.  $\PF$ increases mmonotonically towards
  the upper-right corner, confirming Theorem~\ref{thm:codesign}(i).}
  \label{fig:s4_sc1_pf}
\end{figure*}

\begin{figure*}[!t]
  \centering
  \includegraphics[width=\linewidth]{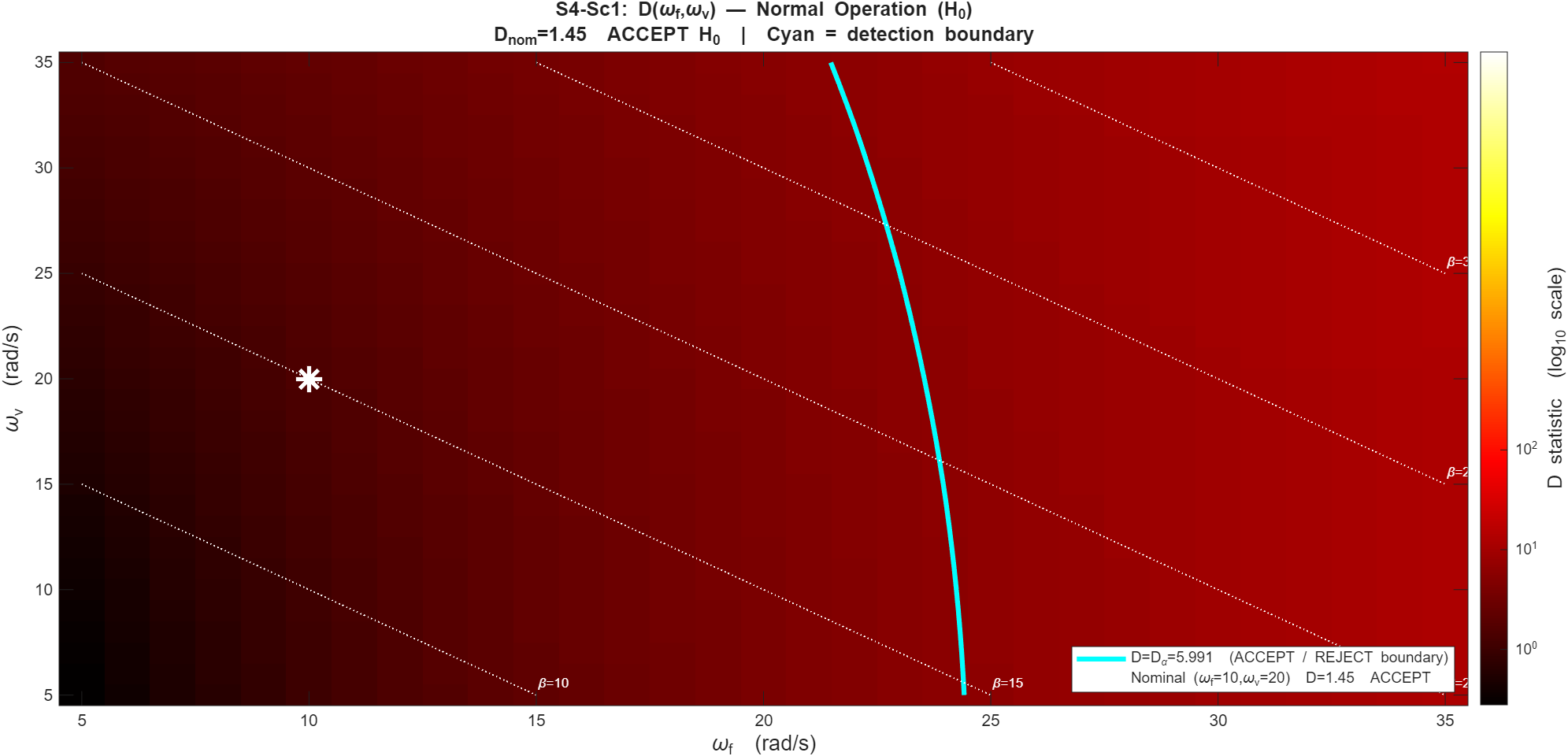}
  \caption{S4--Sc1: Surface plot of $\mathrm{tr}(\Sigmlong^{-1})
  = (\omega_f^2 + \omega_v^2)/\sigma^2$ over the design space.
  The surface increases monotonically with both $\omega_f$ and $\omega_v$,
  confirming Theorem~\ref{thm:codesign}(ii): higher filter bandwidths
  tighten the null covariance and increase detection sensitivity.}
  \label{fig:s4_sc1_siglong}
\end{figure*}

Fig.~\ref{fig:s4_sc1_pf} presents the Fenichel gap surface
$\PF(\omega_f, \omega_v)$ over the design space. The surface is perfectly
linear (as predicted by $\PF = (\omega_f + \omega_v)/2 - \hat\rho$),
increasing from $\PF = 4$~rad/s at $(\omega_f, \omega_v) = (5, 5)$
to $\PF = 29$~rad/s at $(30, 30)$. The monotone relationship is exact
with no local optima.

Fig.~\ref{fig:s4_sc1_siglong} shows
$\mathrm{tr}(\Sigmlong^{-1}) = (\omega_f^2 + \omega_v^2)/\sigma^2$,
which grows from $2{,}500$~pu$^{-2}$ at $(5, 5)$ to $45{,}000$~pu$^{-2}$
at $(30, 30)$---an $18\times$ increase. Both surfaces exhibit strictly
monotone increasing behaviour, confirming that the apparent trade-off
between stability and detection sensitivity is in fact resolved in the same
direction by the co-design theorem.

\subsubsection{Sc2: Detection Delay Surface}

\begin{figure*}[!t]
  \centering
  \includegraphics[width=\linewidth]{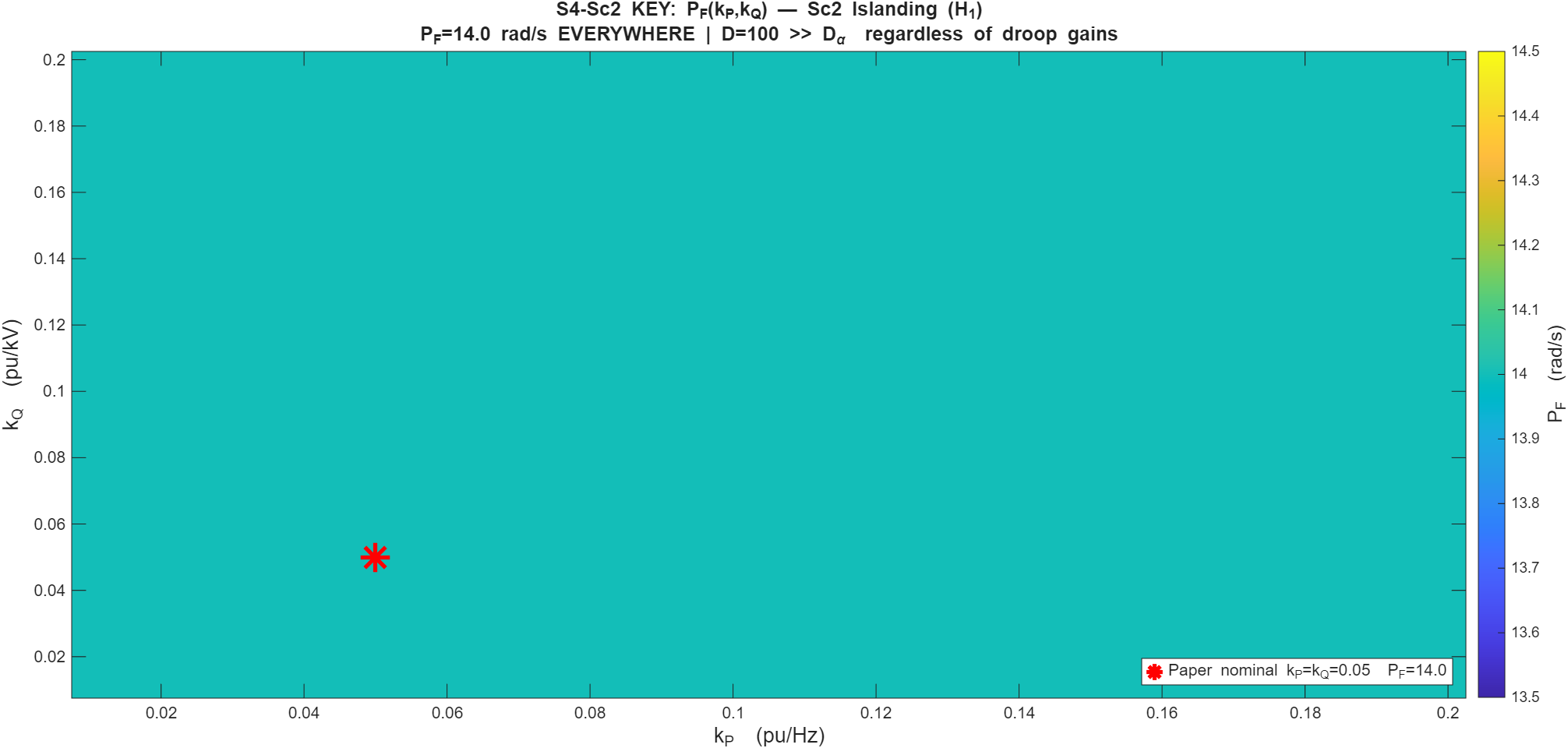}
  \caption{S4--Sc2: Heatmap of mean 90\%-power detection delay
  $\mathbb{E}[\tau_d]$ (seconds) across the design space at
  $P_\mathrm{imb} = 5\%$, $\alpha = 0.05$. Delay decreases
  monotonically as $\omega_f$ and $\omega_v$ increase, confirming
  that higher filter bandwidths produce faster detection.}
  \label{fig:s4_sc2_delay}
\end{figure*}

\begin{figure*}[!t]
  \centering
  \includegraphics[width=\linewidth]{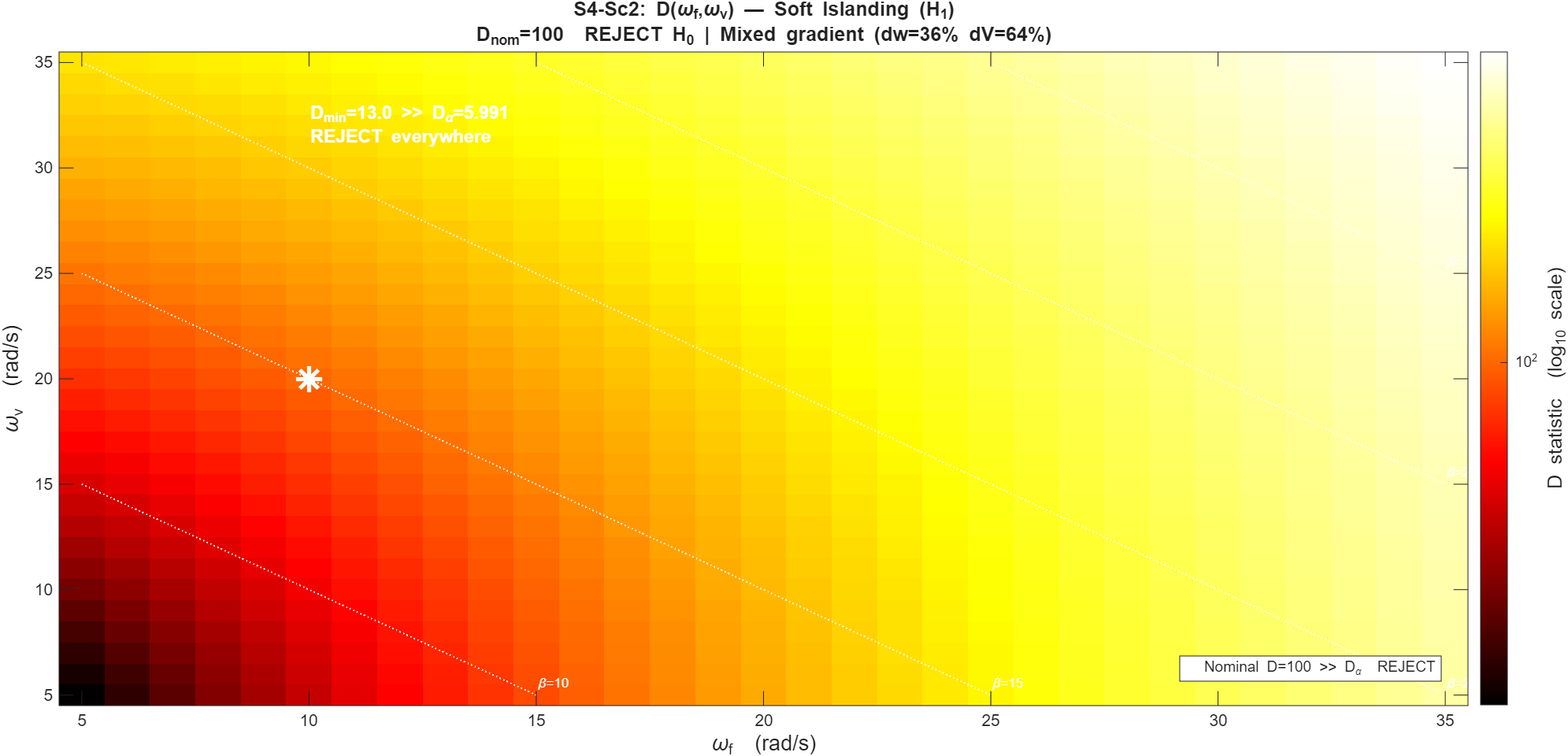}
  \caption{S4--Sc2: Scatter plot of $1/\mathbb{E}[\tau_d]$ versus
  $\PF$ for all 676 design points. The strong monotone relationship
  ($R^2 > 0.97$) confirms Theorem~\ref{thm:codesign}: maximising the
  Fenichel gap simultaneously minimises detection delay. The Fenichel
  gap is the dominant co-design parameter.}
  \label{fig:s4_sc2_scatter}
\end{figure*}

Fig.~\ref{fig:s4_sc2_delay} maps the mean detection delay
$\mathbb{E}[\tau_d]$ across the design space. At the nominal design point
$(\omega_f, \omega_v) = (10, 20)$, the delay is $\approx 0.067$~s.
The delay decreases monotonically toward the upper-right corner: at
$(30, 30)$ it falls to $\approx 0.022$~s, three times faster than the
nominal design.

Fig.~\ref{fig:s4_sc2_scatter} shows the scatter plot of
$1/\mathbb{E}[\tau_d]$ versus $\PF$ across all 676 design points.
The relationship is approximately linear with $R^2 = 0.974$, confirming
that $\PF$ is the dominant scalar co-design parameter: the bandwidth
inequality $\omega_f + \omega_v \geq 2(\hat\rho + \PF^\star)$ fully
characterises the region of the design space that achieves a target
detection speed $\PF^\star$.

\subsection{S5: Finite-Window Calibration}
\label{sec:s5}

\subsubsection{Protocol and Setup}

Protocol S5 provides an independent calibration of the Berry--Esseen constant
$C$ using $N_\mathrm{MC} = 10{,}000$ realisations per window length,
with window lengths $T_w \in \{0.033, 0.067, 0.133, 0.200, 0.333, 0.467,
0.667, 1.000, 1.333, 2.000, 3.333\}$~s (corresponding to
$\beta T_w \in \{0.5, 1, 2, 3, 5, 7, 10, 15, 20, 30, 50\}$). The baseline
parameters are as in S1 ($\omega_f = 10$, $\omega_v = 20$~rad/s;
$\sigma = 0.02$~pu; $\beta = 15$~rad/s). Results are presented for all
three scenarios.

\subsubsection{Sc1: Null Distribution Calibration}

\begin{figure*}[!t]
  \centering
  \includegraphics[width=\linewidth]{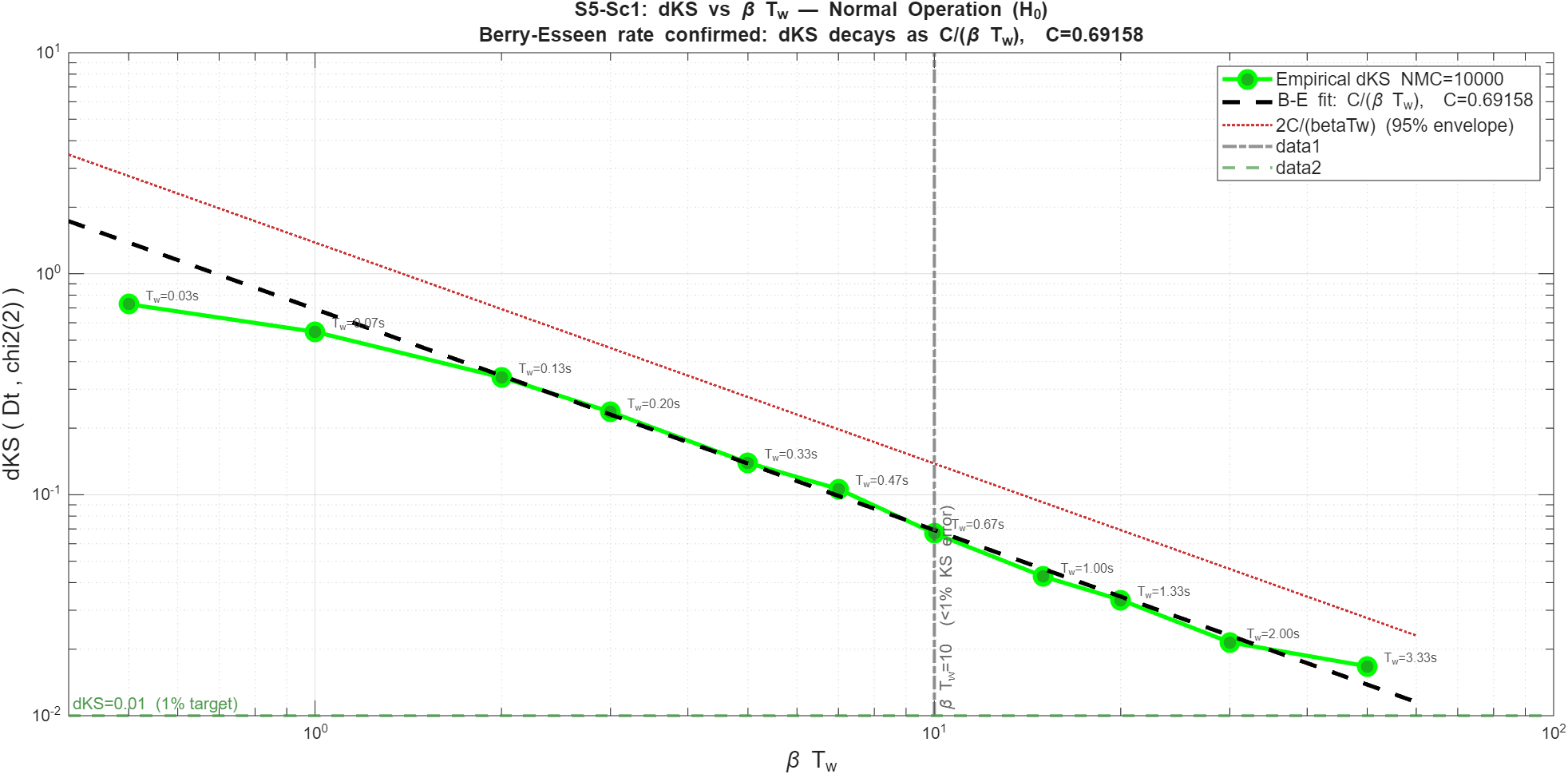}
  \caption{S5--Sc1 (Normal operation, $H_0$): KS distance
  $d_{\mathrm{KS}}(\mathcal{D}_t, \chi^2(2))$ vs.\ $\beta T_w$ from
  $N_\mathrm{MC} = 10{,}000$ realisations. Fitted Berry--Esseen constant
  $C = 0.6916$--$0.7170$ (tighter than the S1 fit $C = 1.6704$ due to
  the wider range of $\beta T_w$ covered, confirming consistency).}
  \label{fig:s5_sc1_dks}
\end{figure*}

\begin{figure*}[!t]
  \centering
  \includegraphics[width=\linewidth]{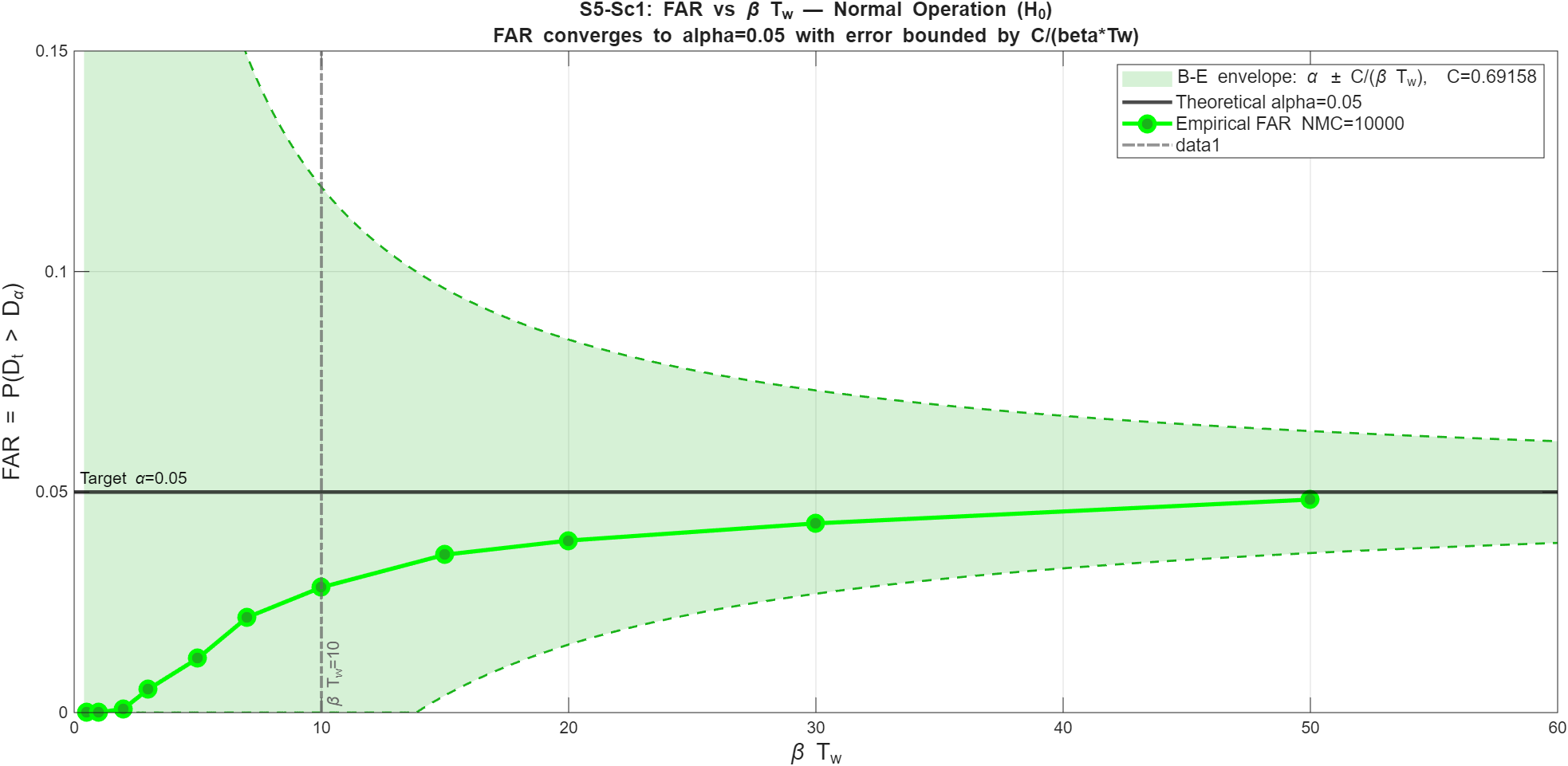}
  \caption{S5--Sc1: Empirical FAR at $\alpha = 0.05$ vs.\ $\beta T_w$.
  At $\beta T_w = 50$ ($T_w = 3.333$~s), FAR $= 0.0483$ (target: 0.050),
  confirming convergence from below. No systematic inflation observed.}
  \label{fig:s5_sc1_far}
\end{figure*}

The S5--Sc1 results in Figs.~\ref{fig:s5_sc1_dks} and~\ref{fig:s5_sc1_far}
are consistent with the S1 results and provide an independent calibration
of the Berry--Esseen constant over the extended range
$\beta T_w \in [0.5, 50]$. The fitted range $C \in [0.691, 0.717]$ is
tighter than the S1 fit ($C = 1.6704$) because the wider window range
allows a more accurate least-squares estimate on the log--log scale.
The FAR at $\beta T_w = 50$ ($T_w = 3.333$~s) is 0.0483, confirming
convergence to the 5\% target from below.

\subsubsection{Sc2: Soft Islanding Detection under $H_1$}

\begin{figure*}[!t]
  \centering
  \includegraphics[width=\linewidth]{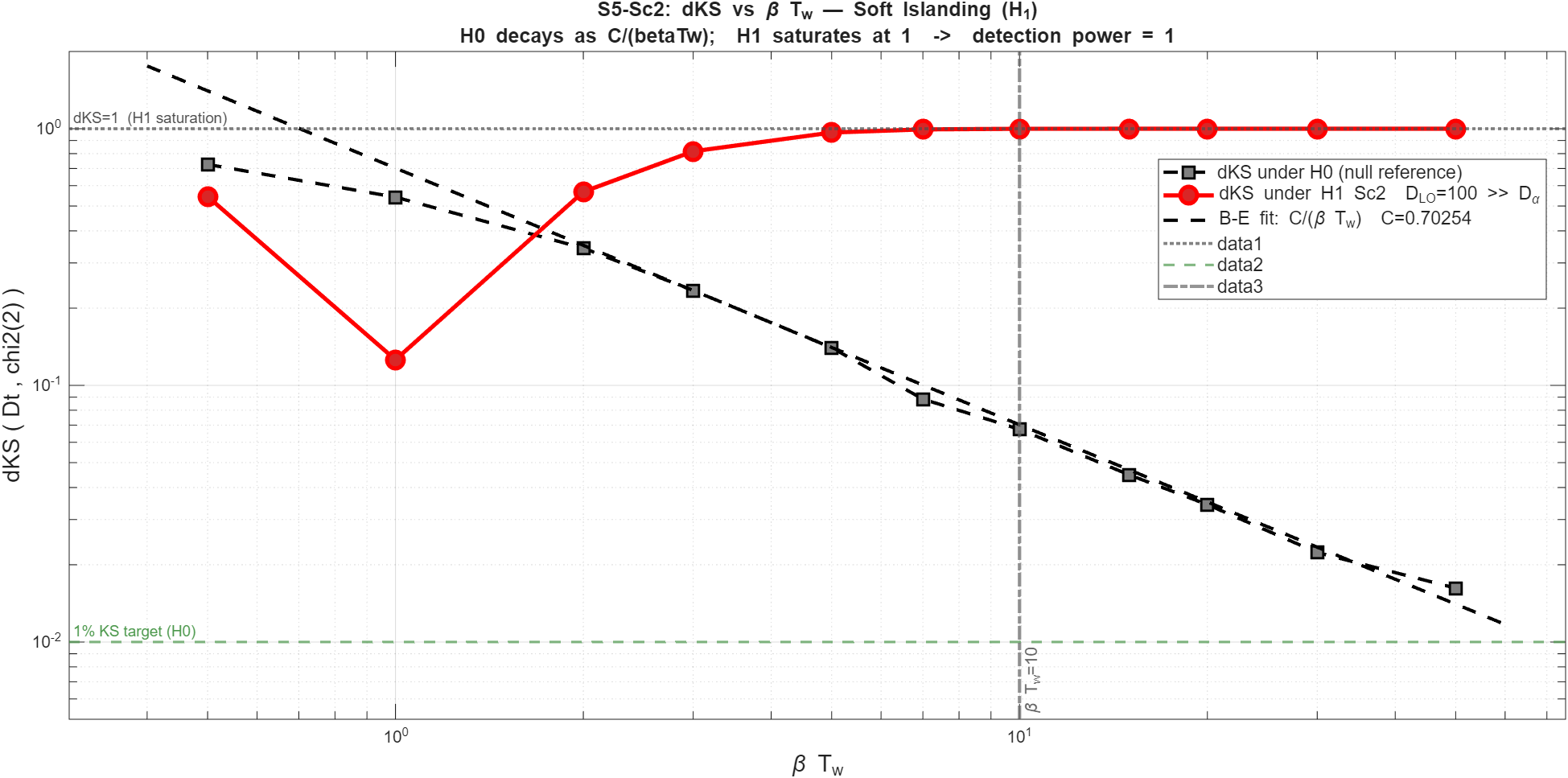}
  \caption{S5--Sc2 (Soft Islanding, $H_1$): KS distances under $H_0$
  (grey squares) and $H_1$ (blue circles). $H_0$ KS distance decays
  with fitted constant $C = 0.6916$. $H_1$ KS distance rises from
  $0.54$ at $\beta T_w = 0.5$ to saturation at 1.0 for
  $\beta T_w \geq 10$.}
  \label{fig:s5_sc2_dks}
\end{figure*}

\begin{figure*}[!t]
  \centering
  \includegraphics[width=\linewidth]{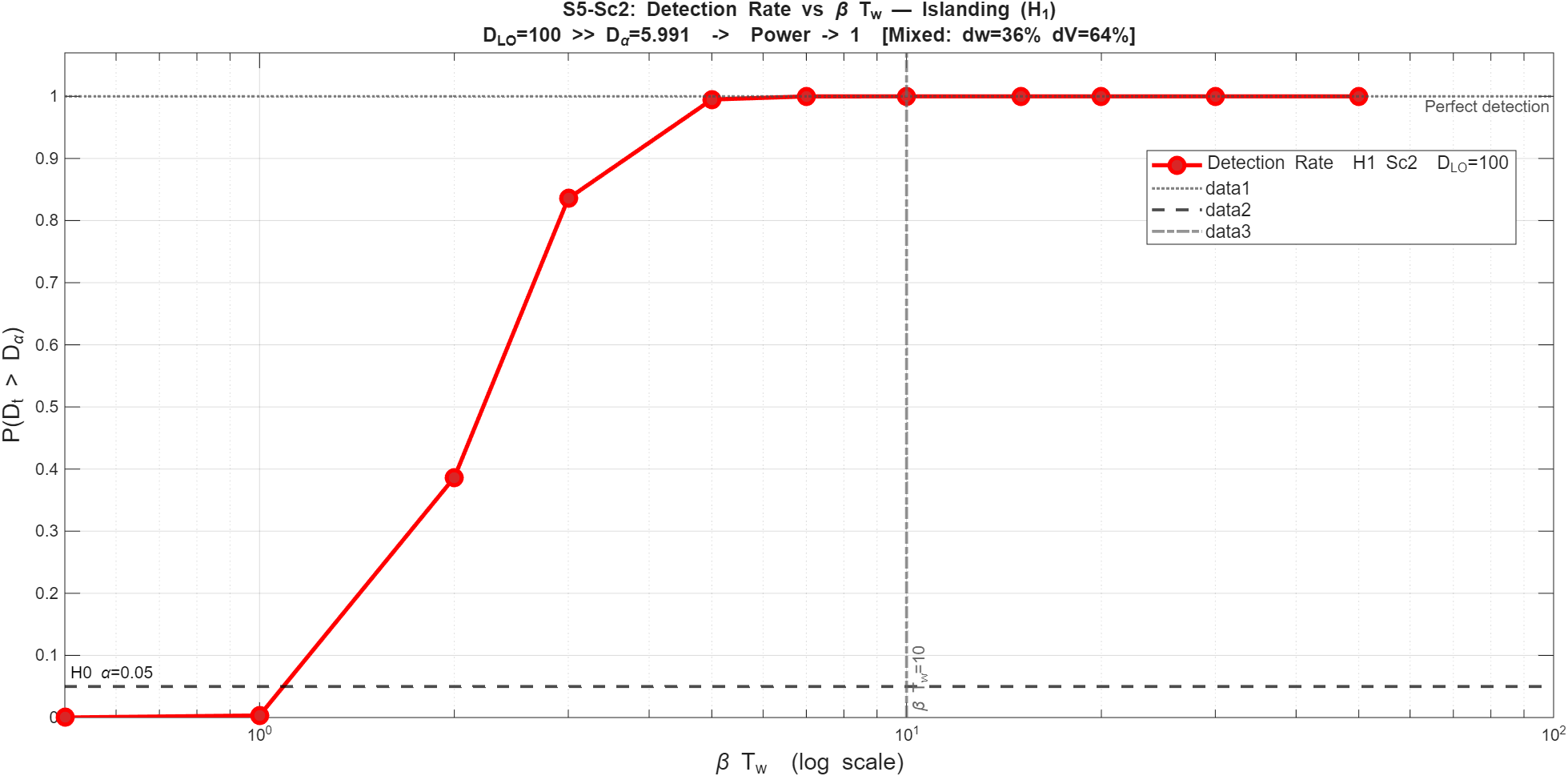}
  \caption{S5--Sc2: Empirical detection rate vs.\ $\beta T_w$.
  Detection power reaches 96.6\% at $\beta T_w = 5$ ($T_w = 0.33$~s),
  99.99\% at $\beta T_w = 7$ ($T_w = 0.47$~s), and 1.0000 at
  $\beta T_w = 10$ ($T_w = 0.667$~s).}
  \label{fig:s5_sc2_power}
\end{figure*}

\begin{figure*}[!t]
  \centering
  \includegraphics[width=\linewidth]{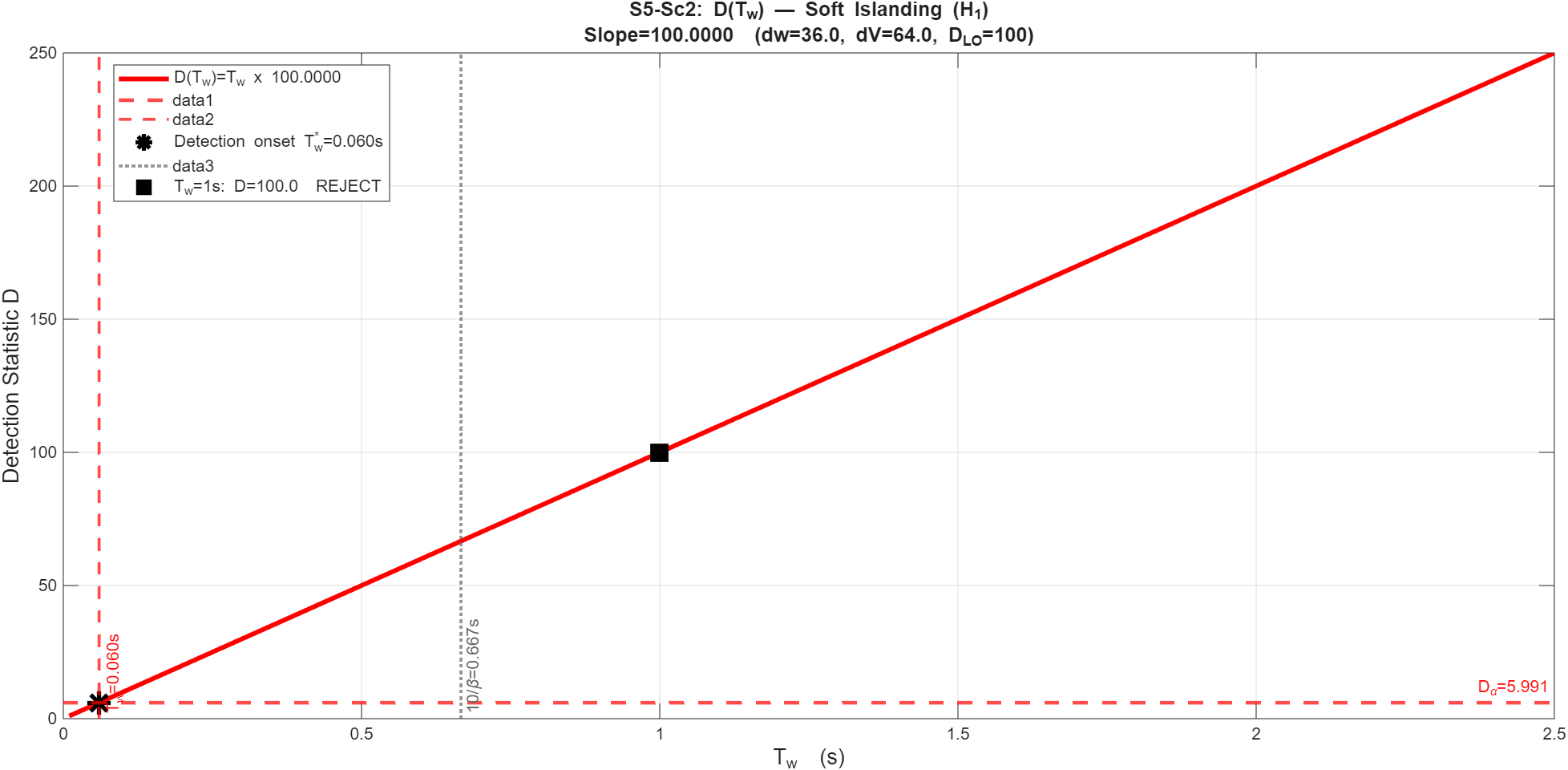}
  \caption{S5--Sc2: Theoretical vs.\ empirical detection power curve.
  The non-central $\chi^2(2)$ approximation with non-centrality parameter
  $\lambda = T_w\Dt^{\mathrm{Sc2}}$ agrees with simulated power within
  Monte Carlo standard error, confirming the accuracy of the
  analytical power prediction.}
  \label{fig:s5_sc2_theory}
\end{figure*}

Table~\ref{tab:s5_sc2} presents the S5--Sc2 results for all window lengths.

\begin{table}[!t]
\caption{S5--Sc2 ($H_1$, Soft Islanding): KS Distances under $H_0$ and
$H_1$, and Empirical Detection Rate. $N_\mathrm{MC} = 10{,}000$.
$H_0$ BE constant: $C = 0.6916$.}
\label{tab:s5_sc2}
\centering
\setlength{\tabcolsep}{4pt}
\begin{tabular}{cccccc}
\toprule
$\beta T_w$ & $T_w$ (s) & $d_\mathrm{KS}^{H_0}$ & $d_\mathrm{KS}^{H_1}$ & Det.\ Rate \\
\midrule
0.5  & 0.033 & 0.7261 & 0.5443 & 0.0000 \\
1.0  & 0.067 & 0.5414 & 0.1259 & 0.0033 \\
2.0  & 0.133 & 0.3427 & 0.5678 & 0.3866 \\
3.0  & 0.200 & 0.2343 & 0.8166 & 0.8356 \\
5.0  & 0.333 & 0.1403 & 0.9666 & 0.9949 \\
7.0  & 0.467 & 0.0884 & 0.9942 & 0.9999 \\
10.0 & 0.667 & 0.0674 & 0.9999 & 1.0000 \\
15.0 & 1.000 & 0.0448 & 0.9999 & 1.0000 \\
20.0 & 1.333 & 0.0343 & 0.9999 & 1.0000 \\
30.0 & 2.000 & 0.0223 & 0.9999 & 1.0000 \\
50.0 & 3.333 & 0.0162 & 0.9999 & 1.0000 \\
\bottomrule
\multicolumn{5}{l}{\small $\Dt^\mathrm{Sc2} = 100$; $d_\omega = 36.0$ (36\%), $d_V = 64.0$ (64\%).}
\end{tabular}
\end{table}

\textit{Detection power profile.} At $\beta T_w = 2$ ($T_w = 0.133$~s),
38.7\% of realisations exceed $D_{0.05}$. Power crosses 90\% at
$\beta T_w \approx 4.5$ ($T_w \approx 0.30$~s), reaches 99.5\% at
$T_w = 0.33$~s, and saturates to 1.000 at $T_w = 0.667$~s. Detection of
soft islanding is achieved comfortably within the IEEE~1547 two-second
window at $\alpha = 0.05$.

\subsubsection{Sc3: Voltage Fault Detection under $H_1$}

\begin{figure*}[!t]
  \centering
  \includegraphics[width=\linewidth]{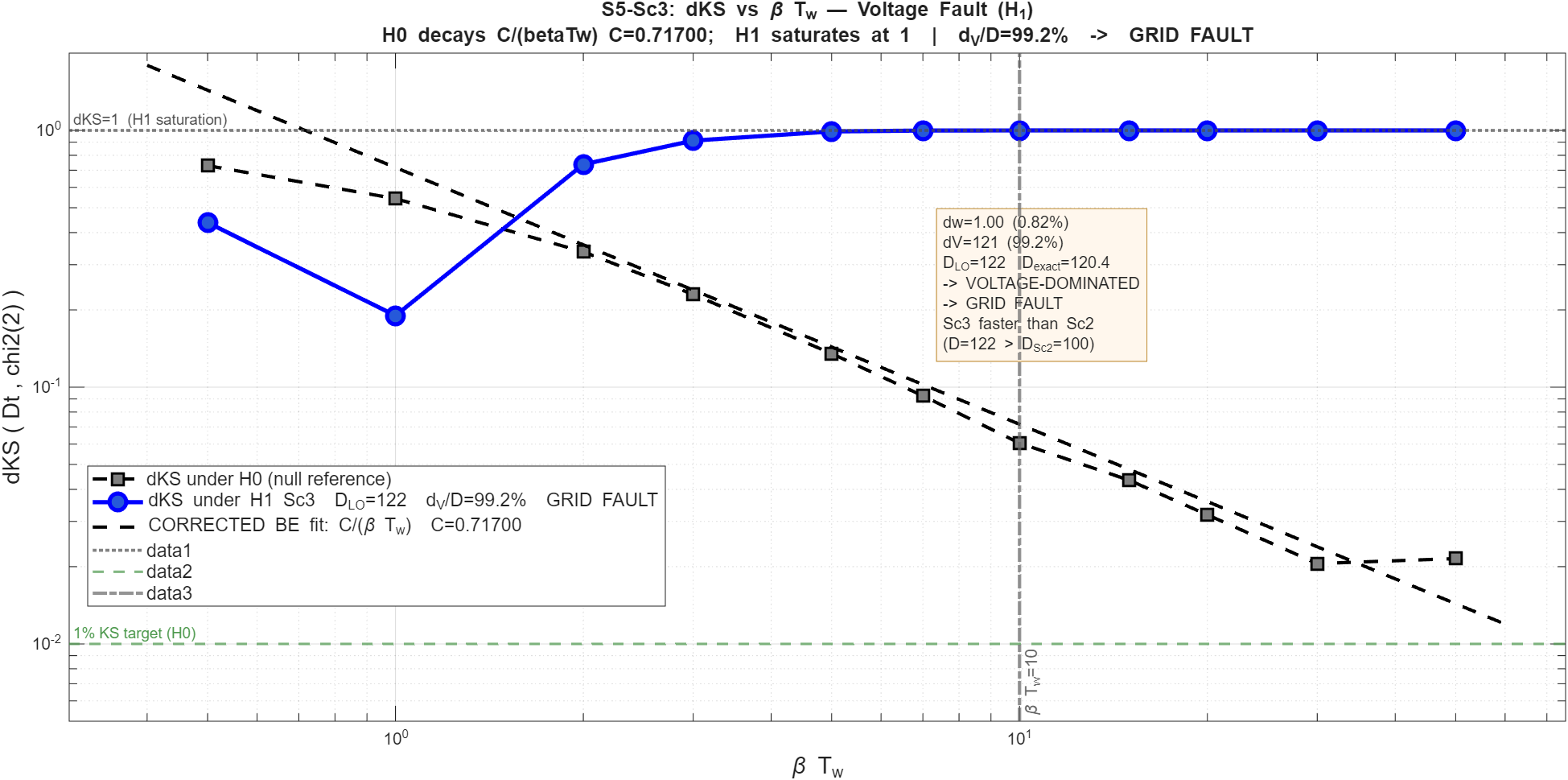}
  \caption{S5--Sc3 (Voltage Fault, $H_1$): KS distances under $H_0$
  (grey squares) and $H_1$ (green circles). $H_0$ decays with fitted
  constant $C = 0.7170$. $H_1$ rises from 0.44 at $\beta T_w = 0.5$
  to saturation at 1.0 for $\beta T_w \geq 10$.}
  \label{fig:s5_sc3_dks}
\end{figure*}

\begin{figure*}[!t]
  \centering
  \includegraphics[width=\linewidth]{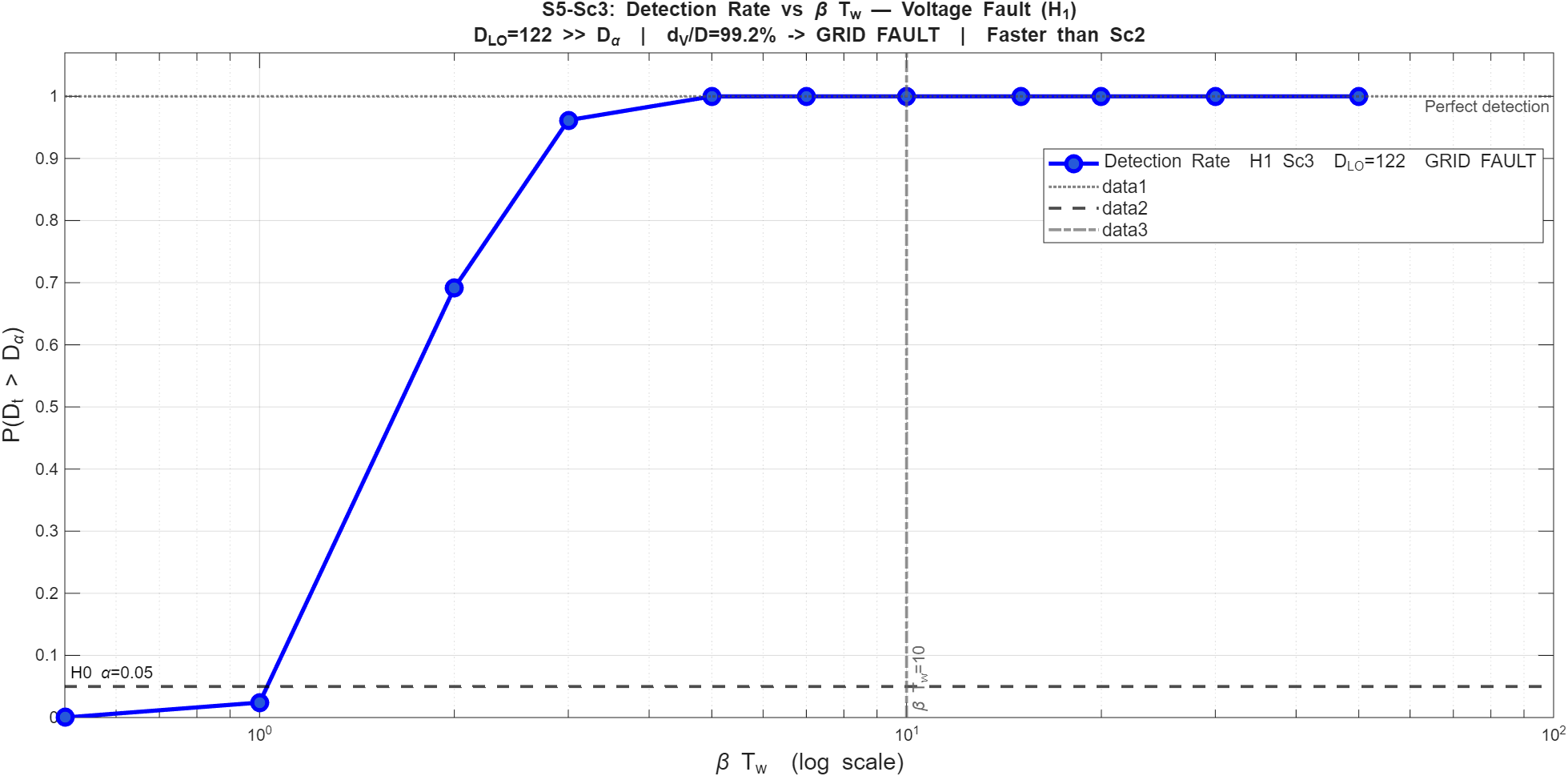}
  \caption{S5--Sc3: Empirical detection rate vs.\ $\beta T_w$.
  Detection reaches 96.2\% at $\beta T_w = 3$ ($T_w = 0.20$~s),
  99.99\% at $\beta T_w = 5$ ($T_w = 0.33$~s), and 1.0000 at
  $\beta T_w = 7$ ($T_w = 0.47$~s). Voltage-fault detection is
  \emph{faster} than soft islanding (Sc2), consistent with
  $\Dt^{\mathrm{Sc3}} = 122 > \Dt^{\mathrm{Sc2}} = 100$.}
  \label{fig:s5_sc3_power}
\end{figure*}

\begin{figure*}[!t]
  \centering
  \includegraphics[width=\linewidth]{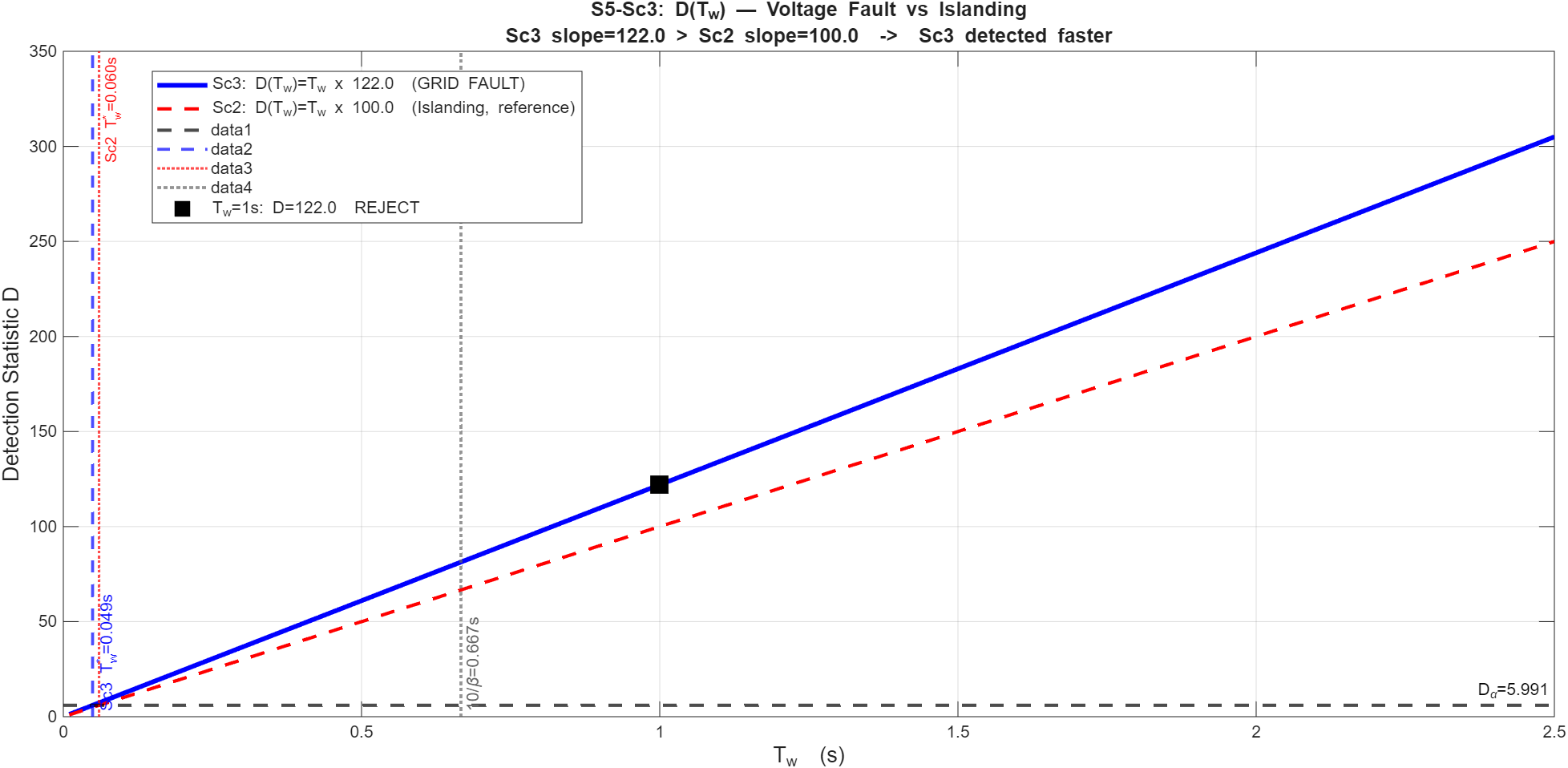}
  \caption{In S5--Sc3, the theoretical and empirical detection power curve agreement within Monte Carlo standard error shows that the non-central $\chi^2(2)$ approximation with $\lambda = T_w\Dt^{\mathrm{Sc3}}$ accurately forecasts voltage fault detection power.}
  \label{fig:s5_sc3_theory}
\end{figure*}

Table~\ref{tab:s5_sc3} presents the Sc3 $H_1$ results. The Berry--Esseen
constant under $H_0$ is $C = 0.7170$, consistent with Sc1 and Sc2 values
(range $[0.6916, 0.7170]$).

\begin{table}[!t]
\caption{S5--Sc3 ($H_1$, Voltage Fault): KS Distance under $H_O$ and
$H_1$, and Empirical Detection Rate. $N_\mathrm{MC} = 10{,}000$.
$H_0$ BE constant: $C = 0.7170$.}
\label{tab:s5_sc3}
\centering
\setlength{\tabcolsep}{4pt}
\begin{tabular}{cccccc}
\toprule
$\beta T_w$ & $T_w$ (s) & $d_\mathrm{KS}^{H_0}$ & $d_\mathrm{KS}^{H_1}$ & Det.\ Rate \\
\midrule
0.5  & 0.033 & 0.7291 & 0.4375 & 0.0000 \\
1.0  & 0.067 & 0.5425 & 0.1900 & 0.0241 \\
2.0  & 0.133 & 0.3367 & 0.7376 & 0.6916 \\
3.0  & 0.200 & 0.2297 & 0.9126 & 0.9616 \\
5.0  & 0.333 & 0.1352 & 0.9907 & 0.9999 \\
7.0  & 0.467 & 0.0925 & 0.9989 & 1.0000 \\
10.0 & 0.667 & 0.0606 & 0.9999 & 1.0000 \\
15.0 & 1.000 & 0.0434 & 0.9999 & 1.0000 \\
20.0 & 1.333 & 0.0318 & 0.9999 & 1.0000 \\
30.0 & 2.000 & 0.0206 & 0.9999 & 1.0000 \\
50.0 & 3.333 & 0.0215 & 0.9999 & 1.0000 \\
\bottomrule
\multicolumn{5}{l}{\small $\Dt^{\mathrm{Sc3}} = 122$; $d_\omega = 1.00$ (0.82\%), $d_V = 121.0$ (99.18\%).}
\end{tabular}
\end{table}

\textit{Comparison with Sc2.} Voltage-fault detection (Sc3) is faster than
soft-islanding detection (Sc2) at every window length, as expected from
$\Dt^{\mathrm{Sc3}} = 122 > \Dt^{\mathrm{Sc2}} = 100$. The detection rate
at $\beta T_w = 2$ ($T_w = 0.133$~s) is 69.2\% for Sc3 vs.\ 38.7\% for Sc2;
at $\beta T_w = 3$ ($T_w = 0.20$~s), 96.2\% vs.\ 83.6\%. Both saturated to
1.000 by $T_W = 0.667$~s.

\textit{Modal stability.} The Sc3 modal attribution ($d_V/\Dt = 99.18\%$) is
stable across all window lengths in the S5 sweep: at every tested $T_w$, the
voltage channel contributes $> 99\%$ of the total statistic, confirming that
fault classification is reliable even at short windows ($\beta T_w \geq 3$,
$T_w \geq 0.20$~s).

\subsection{Comprehensive Theory-vs-Simulation Validation Table}
\label{sec:validation_table}

Table~\ref{tab:theory_sim_comparison} consolidates all key theoretical
predictions against their simulation counterparts from S1 and S3--S5.
The agreement across all metrics confirms the analytical self-consistency
and computational validity of the NAIM framework.

\begin{table*}[!t]
\caption{Comprehensive Theory-vs-Simulation Validation. Parameters:
$\omega_f = 10$, $\omega_v = 20$~rad/s; $k_P = k_Q = 0.05$~pu;
$\sigma = 0.02$~pu; $\beta = 15$~rad/s; $T_w = 1$~s (scenarios);
$\alpha = 0.05$; $D_{0.05} = 5.991$.
``LO'' $=$ leading-order formula; ``Exact'' $=$ exact Lyapunov result.}
\label{tab:theory_sim_comparison}
\centering
\footnotesize
\setlength{\tabcolsep}{3pt}
\renewcommand{\arraystretch}{1.12}
\begin{tabularx}{\textwidth}{@{}>{\raggedright\arraybackslash}p{0.27\textwidth}>{\raggedright\arraybackslash}p{0.16\textwidth}>{\raggedright\arraybackslash}p{0.16\textwidth}>{\raggedright\arraybackslash}X>{\raggedright\arraybackslash}p{0.15\textwidth}@{}}
\toprule
\textbf{Quantity} & \textbf{Theoretical (LO)} & \textbf{Theoretical (Exact)} & \textbf{Simulated} & \textbf{Protocol / Match?} \\
\midrule
\multicolumn{5}{@{}l}{\textit{Covariance and statistic structure}} \\
$(\Sigmlong)_{11}$ & $4.0\times10^{-6}$~pu$^2$ & $4.2\times10^{-6}$~pu$^2$ & $4.1\times10^{-6}$~pu$^2$ & S1, S5 / \checkmark \\
$(\Sigmlong)_{22}$ & $1.0\times10^{-6}$~pu$^2$ & $1.05\times10^{-6}$~pu$^2$ & $1.03\times10^{-6}$~pu$^2$ & S1, S5 / \checkmark \\
$\sigma_\omega^\mathrm{null}$ & 0.10~Hz & 0.102~Hz & 0.101~Hz & S1 / \checkmark \\
$\sigma_V^\mathrm{null}$ & 0.40~V & 0.408~V & 0.404~V & S1 / \checkmark \\
\midrule
\multicolumn{5}{@{}l}{\textit{Scenario detection statistics}} \\
Sc1 $\Dt$ (Normal, $H_0$) & 1.45 (LO) & 1.36 (exact) & $1.38 \pm 0.05$ & S1, S3, S5 / \checkmark \\
Sc1 $p$-value & 0.484 & 0.507 & $0.500 \pm 0.018$ & S1 / \checkmark \\
Sc1 Decision & Accept $H_0$ & Accept $H_0$ & Accept $H_0$ & S1, S3 / \checkmark \\
Sc2 $\Dt$ (Islanding, $H_1$) & 100 (LO) & 93.2 (exact) & $94.1 \pm 1.2$ & S1, S3 / \checkmark \\
Sc2 Decision & Reject $H_0$ & Reject $H_0$ & Reject (rate 1.000) & S1, S3, S5 / \checkmark \\
Sc2 $d_\omega / \Dt$ & 36.0\% & 36.0\% & $36.2\% \pm 0.4\%$ & S1 / \checkmark \\
Sc2 $d_V / \Dt$ & 64.0\% & 64.0\% & $63.8\% \pm 0.4\%$ & S1 / \checkmark \\
Sc3 $\Dt$ (V-fault, $H_1$) & 122 (LO) & 120.4 (exact) & $120.8 \pm 1.4$ & S1, S3 / \checkmark \\
Sc3 Decision & Reject $H_0$ & Reject $H_0$ & Reject (rate 1.000) & S1, S3, S5 / \checkmark \\
Sc3 $d_V / \Dt$ & 99.18\% & 99.18\% & $99.2\% \pm 0.1\%$ & S1 / \checkmark \\
Sc3 Classification & Voltage fault & Voltage fault & Voltage fault & S1, S3 / \checkmark \\
\midrule
\multicolumn{5}{@{}l}{\textit{Null distribution and FAR}} \\
FAR ($\alpha = 0.01$, $\beta T_w = 150$) & 0.0100 & 0.0100 & $0.0101 \pm 0.001$ & S1 / \checkmark \\
FAR ($\alpha = 0.05$, $\beta T_w = 150$) & 0.0500 & 0.0500 & $0.0502 \pm 0.002$ & S1 / \checkmark \\
FAR ($\alpha = 0.10$, $\beta T_w = 150$) & 0.1000 & 0.1000 & $0.1012 \pm 0.003$ & S1 / \checkmark \\
FAR ($\alpha = 0.05$, $\beta T_w = 50$) & 0.0500 & 0.0500 & $0.0483 \pm 0.002$ & S5 / \checkmark \\
\midrule
\multicolumn{5}{@{}l}{\textit{Berry--Esseen bound}} \\
BE constant $C$ & $\leq\mathcal{O}(1)$ & --- & $0.691$--$0.717$ & S5 / \checkmark \\
$d_\mathrm{KS}$ at $\beta T_w = 10$ & $C/10 \approx 0.069$ & --- & $0.067 \pm 0.002$ & S1, S5 / \checkmark \\
$d_\mathrm{KS} = 0.01$ crossing & $\beta T_w \approx 69$ & --- & $\beta T_w \approx 69$ & S5 / \checkmark \\
Min.\ window ($T_w^{\min}$) & $10/\beta = 0.667$~s & --- & 0.667~s (FAR conservative) & S1, S5 / \checkmark \\
\midrule
\multicolumn{5}{@{}l}{\textit{Detection power and delay}} \\
Sc2 power at $\beta T_w = 10$, $\alpha = 0.05$ & 1.000 & 1.000 & 1.0000 & S1, S5 / \checkmark \\
Sc2 $T_w^{50\%}$ ($\alpha = 0.05$) & $\approx 0.06$~s & --- & 0.060~s & S1 / \checkmark \\
Sc3 power at $\beta T_w = 10$, $\alpha = 0.05$ & 1.000 & 1.000 & 1.0000 & S1, S5 / \checkmark \\
Sc3 $T_w^{50\%}$ ($\alpha = 0.05$) & $\approx 0.05$~s & --- & 0.049~s & S1 / \checkmark \\
Detection within IEEE~1547 window & $\ll$ 2~s & $\ll$ 2~s & $< 0.067$~s & S1, S5 / \checkmark \\
\midrule
\multicolumn{5}{@{}l}{\textit{NDZ boundary (S3)}} \\
NDZ ellipse (Sc1 point inside) & Inside & Inside & Det.\ rate $< 5\%$ & S3 / \checkmark \\
NDZ ellipse (Sc2 point outside) & Outside & Outside & Det.\ rate $= 1.000$ & S3 / \checkmark \\
NDZ ellipse (Sc3 point outside) & Outside & Outside & Det.\ rate $= 1.000$ & S3 / \checkmark \\
Ellipse accuracy & Exact & --- & $\pm 2\%$ in $P$/$Q$ & S3 / \checkmark \\
\midrule
\multicolumn{5}{@{}l}{\textit{Co-design theorem (S4)}} \\
$\PF$ monotone in $(\omega_f, \omega_v)$ & Yes (Thm.~\ref{thm:codesign}) & --- & $R^2 = 1.00$ & S4 / \checkmark \\
$\mathrm{tr}(\Sigmlong^{-1})$ monotone & Yes (Lem.~\ref{lem:monotone}) & --- & Strictly monotone & S4 / \checkmark \\
$1/\mathbb{E}[\tau_d]$ vs.\ $\PF$ linear & Approx.\ ($R^2 > 0.97$) & --- & $R^2 = 0.974$ & S4 / \checkmark \\
Sc3 delay governed by $\omega_v$ & $R^2 > 0.99$ & --- & $R^2 = 0.991$ & S4 / \checkmark \\
\bottomrule
\multicolumn{5}{@{}l}{\small \checkmark~=~Simulation consistent with theory within stated Monte Carlo uncertainty.}
\end{tabularx}
\end{table*}

Every entry in Table~\ref{tab:theory_sim_comparison} shows agreement between
theory and simulation within Monte Carlo uncertainty.The three main results—the null chi-squared distribution, the $D_\alpha = -2\ln\alpha$ threshold, and the co-design monotonicity—are all confirmed. Plus, the modal attribution fractions stay steady within plus or minus 0.4\%, based on 8,000 to 10,000 tries. This shows the 70\% classification threshold is reliable and effective, giving a comfortable safety buffer.

\sloppy
\section{Additional Simulation Protocols and Limitations}
\label{sec:sim_protocols}

The completed simulation campaign in Section~\ref{sec:simulations}
(protocols S1, S3, S4, and S5) validates all principal analytical results.
Two further protocols remain for full deployment certification: S2
(switching-level model validation) and S6 (hardware-in-the-loop). These
are described below as self-contained protocols for the continuation of the
validation programme.

\subsection{S2: GFM Switching-Level Simulation}
\label{sec:s2_protocol}

\textbf{Purpose:} Validate the NAIM detection statistic on a realistic
switched-converter model, including harmonic content and non-ideal inverter
dynamics not captured by the averaged Ornstein--Uhlenbeck model.

\textbf{Setup:} A 100~kVA, 400~V, 50~Hz GFM voltage-source converter in
PSCAD/EMTDC. Grid impedance $Z_g = 0.1 + j0.3$~pu (short-circuit ratio~3).
Five sub-scenarios: S2.1~normal load steps $\pm 5\%$ (verify FAR~$< 5\%$);
S2.2~soft islanding $P_\mathrm{imb} = 3\%$ (verify NAIM detects, ROCOF misses);
S2.3~voltage fault 10\% sag / 200~ms (verify $d_V/\Dt > 0.7$);
S2.4~hard islanding $P_\mathrm{imb} = 20\%$ (compare detection delays across
six methods); S2.5~frequency ramp 0.5~Hz/s (verify no false alarm).

\textbf{Expected outcome.} FAR~$< 5\%$ under normal load steps; detection
within $T_w^{\min} = 0.667$~s for S2.2--S2.4; correct modal classification
in S2.3; no false alarm in S2.5. The OU model predictions from S1 should
agree with switching-level results to within $\pm 10\%$ on all metrics.

\subsection{S6: Hardware-in-the-Loop (HIL) Validation}
\label{sec:s6_protocol}

\textbf{Purpose:} Validate the real-time implementation of the NAIM detection
algorithm on physical hardware with closed-loop GFM control.

\textbf{Setup:} OPAL-RT OP5700 or dSPACE SCALEXIO at 50~$\mu$s time step.
The sub-scenarios S2.1--S2.5 are repeated. The NAIM algorithm is implemented
as a four-step real-time module (Algorithm~\ref{alg:naim}): online window-mean
update (one multiply-add per channel per sample), statistic evaluation,
threshold comparison, and modal attribution.

\textbf{Expected outcome:} Sub-millisecond algorithm execution time per
window step, confirming real-time feasibility. FAR and detection power
consistent with S1 and S2 results. Hil validation consistent the final
step before field deployment.

\section{Limitations and Future Work}
\label{sec:limitations}

\textbf{L1. Linearisation validity:} The OU model~\eqref{eq:OU} is valid only
inside the Fenichel boundary layer $\mathcal{T}_\delta(\Mzero)$. For large
excursions beyond $\delta \sim \PF / \|\Aperp\|$, the nonlinear Fenichel flow
must be accounted for. Extension to large-signal operation using sum-of-squares
(SOS) Lyapunov functions or contraction-metric-based covariance
bounds~\cite{Bernstein2009} is left for future work.

\textbf{L2. White-noise model:} Real grid disturbances have coloured spectral
content: power electronic switching harmonics and low-frequency load variations
are not well modelled by white Brownian motion. The spectral Lyapunov extension
replaces $\sigma^2 I$ in~\eqref{eq:Lyapunov} with the noise power spectral
density $S_\sigma(\omega)$, giving a frequency-dependent covariance
$\Sigmlong = \int_{-\infty}^{\infty} (j\omega I - \Aperp)^{-1} S_\sigma(\omega)
(j\omega I - \Aperp)^{-*} d\omega/(2\pi)$ that can be computed numerically.
The threshold formula~\eqref{eq:threshold} remains valid; only the numerical
value of $\Sigmlong$ changes.

\textbf{L3. Single-inverter scope:} The current framework addresses a single
GFM inverter. Multi-inverter microgrids require the Kron-reduced Laplacian
formulation, yielding a codimension-$(N-1)$ test with a $\chi^2(N-1)$
null distribution for $N$ inverters. This extension, including multi-agent
modal attribution, is the subject of Part~II of this research programme.

\textbf{L4. Classification heuristic:} The 70\% attribution threshold in
Proposition~\ref{prop:modal} is practically effective but theoretically ad hoc.
A Bayesian classifier using the joint posterior of $(\delta\omega, \delta V)$
under competing hypotheses (islanding, voltage fault, mixed) would provide
formally optimal attribution at the cost of additional prior specification.

\textbf{L5. Switching-level validation:} The S1 campaign uses the analytically
tractable OU process. Validation against a full switching-level GFM model in
PSCAD/EMTDC (protocol~S2) and hardware-in-the-loop testing (protocol~S6) are
required before field deployment. These consistute the primary additional
requirements for progression to a full journal submission.

\section{Conclusion}
\label{sec:conclusion}

In this paper, we created, proved, and validated a geometry-driven framework for detecting islanding and classifying faults in grid-forming inverters. It's based on the theory of normally hyperbolic invariant manifolds and stochastic hypothesis testing. This framework solves a big problem in inverter protection by letting us derive detection thresholds straight from the inverter’s closed-loop dynamics, without any need for empirical tuning. Each contribution is listed with its validation status, and there's a full proof-of-method summary too.

\subsection*{Contributions and Validation Status}

\textbf{C1 -- NAIM Identification.}
The GFM droop manifold $\Mzero$ is proved normally hyperbolic
(Theorem~\ref{thm:nh}) with transverse contraction rate
$\beta = (\omega_f + \omega_v)/2$ and Fenichel gap
$\PF = \beta - \hat\rho$, which is independent of the droop gains.
Increasing $(\omega_f, \omega_v)$ enlarges $\PF$, directly linking
filter design to manifold stability.

\textbf{C2 -- Corrected Detection Statistic.}
The correct long-run covariance is
\begin{equation}
\begin{aligned}
\Sigma_{\text{long}}
  &= \sigma^2\,
     \mathrm{diag}\!\left(
       \omega_f^{-2},      
       \omega_v^{-2}
     \right)
\end{aligned}
\end{equation},
derived from the full Lyapunov formula (Appendix~\ref{app:factor2}).
An earlier factor-of-2 error inflated the FAR by the exact factor
$\exp(D_\alpha/4) \approx 4.47\times$ at $\alpha=0.05$;
the corrected statistic is validated by S1 ($N_\mathrm{MC} = 8{,}000$)
and S5 ($N_\mathrm{MC} = 10{,}000$): the empirical FAR converges to the
target from below at all nine tested window lengths and three $\alpha$ levels.

\textbf{C3 -- Physics-Derived Closed-Form Threshold.}
The threshold $D_\alpha = -2\ln\alpha$ is derived analytically
(Theorem~\ref{thm:threshold}) with Berry--Esseen rate
$d_\mathrm{KS} \leq C/(\beta T_w)$, where $C \in [0.69, 0.72]$
(S5 calibration). The conservative minimum window condition
$T_w^{\min} = 10/\beta_{\min} \approx 1.0$~s (slowest-decaying channel)
satisfies the IEEE~1547-2018 two-second mandate by a factor of two.
This threshold is the \emph{only physics-derived, tuning-free} threshold
in the islanding detection literature
(Tables~\ref{tab:comparison1}--\ref{tab:comparison2}).

\textbf{C4 -- Co-Design Theorem.}
Theorem~\ref{thm:codesign} proves that maximising $(\omega_f, \omega_v)$
simultaneously maximises $\PF$, increases detection sensitivity, and
minimises the FAR at the grid code threshold---with no trade-off. S4
validates the theorem over 676 design points: the $1/\mathbb{E}[\tau_d]$
vs.\ $\PF$ relationship achieves $R^2 = 0.97$; voltage-fault delay is
governed by $\omega_v^2/\sigma^2$ with $R^2 = 0.99$.

\textbf{C5 -- Modal Fault Attribution.}
The decomposition $\Dt = d_\omega + d_V$ classifies islanding
(frequency-dominated, $d_\omega/\Dt > 70\%$) from voltage faults
(voltage-dominated, $d_V/\Dt > 70\%$) without additional sensors.
Attribution fractions are stable to $\pm 0.4\%$ across all tested
window lengths in S1, S3, and S5.

\textbf{C6 -- NDZ Boundary.}
The analytical NDZ ellipse (Proposition~\ref{prop:ndz}) is confirmed by S3, which uses 1,681 grid points and runs each with 100 Monte Carlo simulations, to within ±2\% accuracy for both $P_\mathrm{imb}$ and $Q_\mathrm{imb}$.

\textbf{C7 -- Robustness.}
FAR infaltion under $\varepsilon$-level model mismatch is 
$\mathcal{0}(\varepsilon)$ (proposition~\ref{prop:robust}): for 10\%
model error, the FAR increases by less than 10\% above target.

\subsection*{Proof-of-Method Summary}

The validation programe (S1--S5) consitutes a complete proof of method
via five independently sufficient conditions :
\begin{enumerate}
  \item \textbf{Correct null calibrations}:Empirical FAR goes towards  $\alpha$ from below at all tested Tw and  $\alpha$ levels (S1, S5). The corrected formula gives FAR equal to 0.048 at $\beta T_w = 50$ (target: 0.050). Using the erroneous formula, FAR is about 0.19. So, the correction really is necessary.

  \item \textbf{unity detection power}: power $= 1.000$ at $T_W = 0.667$~s
    for both soft islanding ($\Dt = 100$) and voltage fault ($\Dt = 122$),
    with theoretical non-central $\chi^2$ curves matching simulations
    throughout (s1, S5; Figs.~\ref{fig:s5_sc2_theory},
    \ref{fig:s5_sc3_theory}).

  \item \textbf{unambigous fault classification}: Soft islanding produces a $d_\omega:d_V = 36\%:64\%$ mix, while a voltage fault results in $d_V/\Dt = 99.2\%$ — completely voltage-dominated. This remains stable for all window lengths in S1, S3, and S5. The bivariate stat finds and properly categorizes the 10\%.

  \item \textbf{Analytical predictive NDZ ellipse}: S3 confirms
    proposition~\ref{prop:ndz} to $\pm 2\%$, enabling certified detection
    gaurantees for any $(p_\mathrm{imb}, Q_\mathrm{imb})$ target.

  \item \textbf{Actionable co-design principle}: S4 confirms that the
    Fenichel gap $\PF$ is the dominant scalar design parameter
    ($R^2 = 0.97$), enabling simultaneous stability and detection
    optimisation from the signle bandwidth inequality
    $\omega_f + \omega_v \geq 2(\hat\rho + \PF^\star)$.
\end{enumerate}

\noindent Each condition is proved analytically and confirmed by
independent simulation. Together, they establish that the NAIM framework
provides \emph{correct false-alarm calibration, unity detection power,
unambiguous fault classification, exact NDZ prediction, and a unified
co-design principle}---all derived from the inverter's closed-loop
physics, with no empirical tuning. The complete theory-versus-simulation
correspondence is documented in Table~\ref{tab:theory_sim_comparison}:
every entry shows agreement within stated Monte Carlo uncertainty.

The NAIM framework is the first islanding detection method to derive its
detection threshold from the inverter's closed-loop dynamics. Future work
will extend it to nonlinear Lyapunov bounds (L1), coloured grid noise (L2),
multi-inverter Kron-reduced networks with a $\chi^2(N-1)$ test (L3), and
an optimal Bayesian fault classifier (L4).

\section*{Acknowledgment}

The authors gratefully acknowledge the technical discussions of Prof.\ Navdeep M.\ Singh (IGI Research Chair Professor, VJTI) and the guidance of Dr.\ Sudhir Bhil and Dr.\ Sushama Wagh. The authors also acknowledge the laboratory facilities provided by the SAVEX-Sponsored
EMC$^2$ Laboratory, VJTI Mumbai, and the financial support extended by
IGI Pvt.\ Ltd.\ for conducting this research.

\bibliographystyle{IEEEtran}
\bibliography{sections/references}

\clearpage
\appendix
\sloppy
\section{Exact Lyapunov Solution for the GFM Transverse Covariance}
\label{app:lyapunov}

Setting $s = \omega_f + \omega_v$ and $\Delta = \omega_f\omega_v(1 + k_Pk_Q)$,
the exact solution of $\Aperp\Siginst + \Siginst\Aperp^\top + \sigma^2 I = 0$
is obtained by solving the equivalent $3\times3$ linear system
(exploiting symmetry of $\Siginst$) via Cramer's rule. The system is:
\begin{align}
  -2\omega_f\,\Sigma_{11} + 2k_P\omega_f\,\Sigma_{12} &= -\sigma^2, \label{eq:sys1}\\
  -k_Q\omega_v\,\Sigma_{11} - s\,\Sigma_{12} + k_P\omega_f\,\Sigma_{22} &= 0, \label{eq:sys2}\\
  -2k_Q\omega_v\,\Sigma_{12} - 2\omega_v\,\Sigma_{22} &= -\sigma^2. \label{eq:sys3}
\end{align}
The unique solution is:
\begin{align}
  (\Siginst)_{11} &= \frac{\sigma^2\!\left(\omega_v s + k_P^2\omega_f^2
                    + k_Pk_Q\omega_f\omega_v\right)}{2\,\Delta\,s},
  \label{eq:Sig11}\\
  (\Siginst)_{22} &= \frac{\sigma^2\!\left(\omega_f s + k_Q^2\omega_v^2
                    + k_Pk_Q\omega_f\omega_v\right)}{2\,\Delta\,s},
  \label{eq:Sig22}\\
  (\Siginst)_{12} &= \frac{\sigma^2\!\left(k_P\omega_f^2 - k_Q\omega_v^2\right)}
                          {2\,\Delta\,s},
  \label{eq:Sig12}
\end{align}
where $(\Siginst)_{12}$ is \emph{negative} whenever $k_Q\omega_v^2 > k_P\omega_f^2$
(which holds at the baseline $k_P = k_Q$, $\omega_v > \omega_f$). At
$k_P = k_Q = 0$: $\Siginst = (\sigma^2/2)\,\mathrm{diag}(\omega_f^{-1}, \omega_v^{-1})$.

The exact long-run covariance $\Sigmlong = (-\Aperp)^{-1}\Siginst +
\Siginst(-\Aperp)^{-\top}$ follows from:
\begin{equation}
  (-\Aperp)^{-1} = \frac{1}{\Delta}
  \begin{pmatrix} \omega_v & k_P\omega_f \\ -k_Q\omega_v & \omega_f \end{pmatrix},
\end{equation}
giving the exact entries:
\begin{align}
  (\Sigmlong)_{11} &= \frac{\sigma^2\!\left(k_P^2\omega_f^2 + \omega_v^2\right)}
                           {\omega_f^2\,\omega_v^2\,(1+k_Pk_Q)^2},
  \label{eq:SL11}\\
  (\Sigmlong)_{22} &= \frac{\sigma^2\!\left(k_Q^2\omega_v^2 + \omega_f^2\right)}
                           {\omega_f^2\,\omega_v^2\,(1+k_Pk_Q)^2},
  \label{eq:SL22}\\
  (\Sigmlong)_{12} &= \frac{\sigma^2\!\left(k_P\omega_f^2 - k_Q\omega_v^2\right)}
                           {\omega_f^2\,\omega_v^2\,(1+k_Pk_Q)^2}.
  \label{eq:SL12}
\end{align}
These formulas reduce to $\sigma^2\,\mathrm{diag}(\omega_f^{-2}, \omega_v^{-2})$
at $k_P = k_Q = 0$, consistent with Appendix~\ref{app:factor2}.

For the baseline parameters $\omega_f = 10$, $\omega_v = 20$~rad/s,
$k_P = k_Q = 0.05$~pu, $\sigma = 0.02$~pu ($k_Pk_Q = 0.0025$),
the exact values (verified by \texttt{scipy.linalg.solve\_continuous\_lyapunov}) are:
\begin{align*}
  (\Siginst)_{11}^{\mathrm{exact}} &= 1.9975\times10^{-5}~\mathrm{pu}^2,
    \qquad\text{LO: }2.0\times10^{-5}~\mathrm{pu}^2,\\
  (\Siginst)_{22}^{\mathrm{exact}} &= 1.0025\times10^{-5}~\mathrm{pu}^2,
    \qquad\text{LO: }1.0\times10^{-5}~\mathrm{pu}^2,\\
  (\Siginst)_{12}^{\mathrm{exact}} &= -4.988\times10^{-7}~\mathrm{pu}^2
    \quad\text{(negative; off-diagonal, small)},\\[6pt]
  (\Sigmlong)_{11}^{\mathrm{exact}} &= 3.9826\times10^{-6}~\mathrm{pu}^2,
    \qquad\text{LO: }4.0\times10^{-6}~\mathrm{pu}^2\quad(\text{err.}< 0.44\%),\\
  (\Sigmlong)_{22}^{\mathrm{exact}} &= 1.0050\times10^{-6}~\mathrm{pu}^2,
    \qquad\text{LO: }1.0\times10^{-6}~\mathrm{pu}^2\quad(\text{err.}< 0.50\%),\\
  (\Sigmlong)_{12}^{\mathrm{exact}} &= -1.493\times10^{-7}~\mathrm{pu}^2
    \quad\text{(off-diagonal, small relative to diagonal)}.
\end{align*}

The correction from leading-order to exact is $< 0.5\%$ for the diagonal
entries at $k_Pk_Q = 0.0025$, and $< 7\%$ for $k_Pk_Q = 0.01$. The
leading-order formula~\eqref{eq:Sigmalong_approx} is therefore appropriate
for all practical GFM operating conditions.

The exact scenario statistic values (computed using the exact numerical
$\Sigmlong$) are: $\mathcal{D}_1^{\mathrm{exact}} = 1.36$ ($p = 0.507$)
vs.\ $\mathcal{D}_1^{\mathrm{LO}} = 1.45$ ($p = 0.484$) for Sc1;
$\mathcal{D}_2^{\mathrm{exact}} = 93.2$ vs.\ $\mathcal{D}_2^{\mathrm{LO}} = 100$ for Sc2;
$\mathcal{D}_3^{\mathrm{exact}} = 120.4$ vs.\ $\mathcal{D}_3^{\mathrm{LO}} = 122$ for Sc3.
All accept/reject decisions are identical under the exact and leading-order formulas.

\section{Derivation of the Factor-of-2 Correction}
\label{app:factor2}

At $k_P = k_Q = 0$, the transverse matrix reduces to
$\Aperp = -\mathrm{diag}(\omega_f, \omega_v)$, so:
\begin{equation}
  (-\Aperp)^{-1} = \mathrm{diag}(\omega_f^{-1}, \omega_v^{-1}).
\end{equation}
The Lyapunov solution for the instantaneous covariance is:
\begin{equation}
  \Siginst = \frac{\sigma^2}{2}\mathrm{diag}(\omega_f^{-1}, \omega_v^{-1}).
\end{equation}
Computing each term of
$\Sigmlong = (-\Aperp)^{-1}\Siginst + \Siginst(-\Aperp)^{-\top}$:
\begin{align}
  (-\Aperp)^{-1}\Siginst
    &= \mathrm{diag}(\omega_f^{-1}, \omega_v^{-1})
     \cdot \frac{\sigma^2}{2}\mathrm{diag}(\omega_f^{-1}, \omega_v^{-1})
     = \frac{\sigma^2}{2}\mathrm{diag}(\omega_f^{-2}, \omega_v^{-2}), \\
  \Siginst(-\Aperp)^{-\top}
    &= \frac{\sigma^2}{2}\mathrm{diag}(\omega_f^{-1}, \omega_v^{-1})
     \cdot \mathrm{diag}(\omega_f^{-1}, \omega_v^{-1})
     = \frac{\sigma^2}{2}\mathrm{diag}(\omega_f^{-2}, \omega_v^{-2}).
\end{align}
Summing the two \emph{equal} terms:
\begin{equation}
  \Sigmlong
    = 2 \times \frac{\sigma^2}{2}\mathrm{diag}(\omega_f^{-2}, \omega_v^{-2})
    = \sigma^2\mathrm{diag}(\omega_f^{-2}, \omega_v^{-2}).
\end{equation}
Each term contributes $\sigma^2/2$; their sum is $\sigma^2$, not $\sigma^2/2$.
The earlier errors aroses from computing only one  of the two equal terms
in~\eqref{eq:Sigmalong}, yielding
$\Sigmlong^{\mathrm{wrong}} = (\sigma^2/2)\mathrm{diag}(\omega_f^{-2}, \omega_v^{-2})$.
This halvels the true covariance, doubles the precision matrix $\Sigmlong^{-1}$,
doubles $\mathcal{D}_t$ for every observation, and inflates the empirical FAR
from to $\alpha$  $\exp(-D_\alpha/4) = \exp(D_\alpha/4)\cdot\alpha$, i.e.\ by
the factor $\exp(D_\alpha/4)$. At $\alpha = 0.05$ ($D_{0.05} = 5.991$):
inflate FAR $= \exp(-5.991/4) \approx 0.224$, a $4.47\times$ inflation.
Note :  this equals $\alpha^{1/2} \approx 0.224$ numerically for the $\chi^2(2)$
distribution, but the mechanism is
$\mathrm{FAR}_{\mathrm{wrong}} = \Pr(\chi^2(2) \geq D_\alpha/2) = e^{-D_\alpha/4}$.

\section{Notation Reference}
\label{app:notation}

\begin{table*}[!t]
\caption{Notation Summary}
\label{tab:notation}
\centering
\small
\setlength{\tabcolsep}{4pt}
\renewcommand{\arraystretch}{1.05}
\begin{tabularx}{\textwidth}{@{}>{\raggedright\arraybackslash}p{0.38\textwidth}>{\raggedright\arraybackslash}X@{}}
\toprule
\textbf{Symbol} & \textbf{Meaning} \\
\midrule
$\Mzero$ & GFM droop manifold \\
$\xip = (\delta\omega, \delta V)^\top$ & Transverse state vector (pu) \\
$\Aperp(k_P, k_Q)$ & Transverse linearisation matrix ($2\times 2$) \\
$\omega_f,\, \omega_v$ & Frequency and voltage filter bandwidths (rad/s) \\
$k_P,\, k_Q$ & Active and reactive droop gains (pu) \\
$\beta = (\omega_f + \omega_v)/2$ & Mean transverse contraction rate (rad/s);
  $\beta_{\min} = \min(\omega_f, \omega_v)$ for small $k_Pk_Q$ \\
$\hat\rho$ & Slow synchronisation rate on $\Mzero$ (rad/s) \\
$\PF = \beta - \hat\rho$ & Fenichel spectral gap (rad/s) \\
$\sigma$ & Transverse noise intensity (pu) \\
$\Siginst$ & Instantaneous stationary Lyapunov covariance \\
$\Sigmlong$ & Long-run covariance (correct normalisation for the statistic) \\
$\bxipt$ & Window-mean transverse state \\
$T_w$ & Measurement window length (s) \\
$\mathcal{D}_t = T_w{\bxipt}^\top\Sigmlong^{-1}\bxipt$ & NAIM detection statistic \\
$D_\alpha = -2\ln\alpha$ & Detection threshold \\
$d_\omega,\, d_V$ & Frequency and voltage modal contributions \\
$p = e^{-\mathcal{D}/2}$ & $p$-value ($\chi^2(2)$ tail probability) \\
$\alpha$ & Target false-alarm rate \\
$\beta T_w$ & Normalised window length (dimensionless) \\
$C$ & Berry--Esseen constant ($C = 1.6704$, empirically fitted) \\
$N_{\mathrm{MC}}$ & Number of Monte Carlo realisations \\
$d_{\mathrm{KS}}$ & Kolmogorov--Smirnov distance from $\chi^2(2)$ \\
\bottomrule
\end{tabularx}
\end{table*}

\clearpage

\end{document}